\documentclass{article}

\usepackage[english]{babel}
\usepackage[utf8]{inputenc}
\usepackage{authblk}
\usepackage{amsmath,amssymb}
\usepackage{parskip}
\usepackage{mathrsfs}
\usepackage{graphicx}
\usepackage{mathtools}
\usepackage{listings}
\usepackage{enumerate}
\usepackage[hidelinks]{hyperref} 
\usepackage{amsthm}
\usepackage{amsmath}
\usepackage{indentfirst} 
\newtheorem{theorem}{Theorem}
\newtheorem{lemma}{Lemma}
\newtheorem{proposition}{Proposition}

\newtheorem{remark}{Remark}
\usepackage{enumitem}
\usepackage[font=small,labelfont=bf]{caption}
\setlist[enumerate]{leftmargin=*}
\usepackage{stmaryrd}
\newcommand*{\ldblbrace}{\{\mskip-5mu\{}
\newcommand*{\rdblbrace}{\}\mskip-5mu\}}
\usepackage{caption}
\usepackage{subcaption}
\DeclareMathOperator{\Div}{div}
\DeclareMathOperator*{\argmin}{argmin}
\usepackage{todonotes}
\usepackage{titling}

\usepackage[top=2.5cm, left=2.5cm, right=2.5cm, bottom=4.0cm]{geometry}
\usepackage{xcolor}
\usepackage{float}
\usepackage{bm}
\allowdisplaybreaks[4]
\newcommand*{\dif}{\mathop{}\!\mathrm{d}}

\newcommand{\RN}[1]{%
  \textup{\uppercase\expandafter{\romannumeral#1}}%
}

\title{Modeling and computation of the effective elastic behavior of parallelogram origami metamaterials}
\author{Hu Xu$^{1}$, Fr\'ed\'eric Marazzato$^{2}$, Paul Plucinsky$^{1,\ast}$}
\date{\small $^1$Aerospace and Mechacanical Engineering, University of Southern California, Los Angeles, CA 90014, USA \\ $^2$Department of Mathematics, The University of Arizona, Tucson, AZ 85721, USA \\  $^\ast$Corresponding author email: plucinsk@usc.edu \\ [2ex] \normalsize \today} 

\begin{document}

\maketitle

\begin{abstract}
\noindent Origami metamaterials made of repeating unit cells of parallelogram panels joined at folds dramatically change their shape through a collective motion of their cells. Here we develop an effective elastic model and numerical method to study the large deformation response of these metamaterials under a broad class of loads. The model builds on an effective plate theory derived in our prior work  \cite{xu2024derivation}. The theory captures the overall shape change of all slightly stressed  parallelogram origami deformations through nonlinear geometric compatibility constraints  that couple the origami's (cell averaged) effective deformation to an auxiliary angle field quantifying its cell-by-cell actuation. It also assigns to each such origami deformation a plate energy associated to these effective fields. Seeking a constitutive model that is faithful to the theory but also practical to simulate, we relax the geometric constraints via corresponding elastic energy penalties; we also simplify the plate energy density to   embrace its essential character as a regularization to the geometric  penalties. The resulting model for parallelogram origami  is  a generalized elastic continuum that is nonlinear in the effective  deformation gradient and angle field  and regularized by high-order gradients thereof.   We provide a finite element formulation of this model using the $C^0$ interior penalty method to handle second gradients of deformation,  and implement it using the open source computing platform \texttt{Firedrake}. We end by using the model and numerical method to study  two canonical  parallelogram origami patterns, in Miura and Eggbox origami, under a variety of loading conditions.   
\end{abstract}

\section{Introduction}

Flexible mechanical metamaterials are cell-based patterns of stiff panels, connected at flexible hinges or folds, designed to achieve large overall shape-change at little overall stress. They are a promising class of  materials for many applications, including for locomotion  and grasping in soft robotics \cite{kim2018printing,rafsanjani2019programming,yang2021grasping} and for the deployment of medical devices, space structures, and habitats \cite{kuribayashi2006self,melancon2021multistable,velvaluri2021origami,zirbel2013accommodating}. However, their mechanical behavior arises from a multiscale coupling spanning collective cell-wise interactions, large panel rotations, and localized distortion at the hinges or folds that is challenging to model and simulate. Here we study   \textit{parallelogram origami} ---   a large class  of flexible mechanical metamaterials that exemplifies these challenges. Our goal is a modeling and computational framework that predicts the  bulk macroscale response of these metamaterials under  a broad range of loads.

Parallelogram origami is a family of mechanical metamaterials made of repeating unit cells of four
parallelograms panels joined at folds. The most recognizable example is the Miura origami. Introduced as a means for packaging and deploying large space membranes  decades ago \cite{koryo1985method}, this pattern has since become an archetype for exploring functionality enabled by shape-morphing. It is the underlying motif in the inverse design of  targeted shapes on morphing \cite{dang2022inverse,dudte2016programming,sardas2024continuum}, can be used for self-folding at small and large scales \cite{na2015programming,tolley2014self}, and achieves    auxetic compression and saddle-like  bending modes  under loads \cite{schenk2013geometry, wei2013geometric}. Other well-studied examples of parallelogram origami include the Eggbox and Morph patterns \cite{pratapa2019geometric,schenk2011origami}. These patterns are distinct from the Miura in many ways.  Both are (generically) non-Euclidean in that their vertex sector angles do not sum to $ 2\pi$. In addition, Eggbox has a positive Poisson's ratio and is thus not an auxetic in contrast to the Miura. It also takes  a  cap-like shape when bent, rather than a saddle shape.  The Morph, meanwhile, is distinguished by its versatility. It is either Miura-like  in its mechanical behavior or Eggbox-like depending   on the choice of mountain-valley assignment. Importantly, these patterns share one  key mechanical property: Each can fold  as a periodic mechanism \cite{pellegrino1986matrix} or floppy mode \cite{lubensky2015phonons}, i.e.,  a shape-changing continuous motion involving  a  periodic and rigid rearrangement of the panels about the flexible folds.  It turns out that all parallelogram origami patterns are capable of folding as a periodic mechanism \cite{mcinerney2022discrete,nassar2022strain,xu2024derivation}.  This property enables these patterns to dramatically change their shape under a wide range of loads, leading to a rich variety of \textit{soft modes of deformation} (e.g., Fig.\;\ref{Fig:IntroFig}) whose stored elastic energy is  far less than bulk, and inviting basic questions on how to best model their elastic behavior.

\begin{figure}[t!]
\centering
\includegraphics[width=1\textwidth]{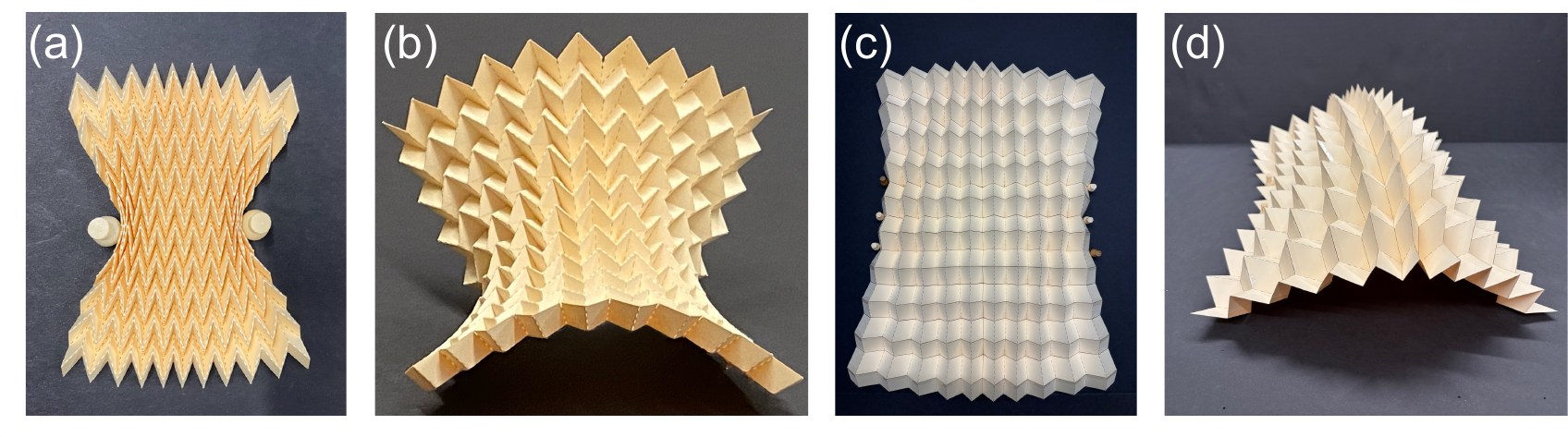} 
\caption{Examples of soft modes of deformation in paper models of Miura and Eggbox origami. (a-b) Pinching and bending in Miura Origami. (c-d) Analogous examples for Eggbox origami. }
\label{Fig:IntroFig}
\end{figure}

By far the most popular modeling approach for these and other flexible mechanical  metamaterials is  bar-and-hinge elasticity \cite{filipov2017bar,liu2017nonlinear,schenk2011origami} or related spring methods \cite{deng2020characterization,zhou2023low}.  These approaches replace the metamaterial with assemblies of elastic bars and hinges or introduce linear and torsional springs between the rigid elements, allowing the augmented material system to be analyzed using standard numerical methods in structural mechanics.  Bar-and-hinge elasticity is particularly versatile and convenient for systems composed of a small number of building blocks,  and the standard bearer in such settings \cite{zhu2022review}. However, it can become computationally expensive for  large  systems with many unit cells. It can also be challenging to fit the many spring stiffnesses in this method to yield accurate global behavior. Most importantly, the method does not provide a characterization of general soft modes of deformation, beyond  simulating them. In particular, it does not explain why the Miura prefers a saddle shape when bent and the Eggbox  prefers a cap.

An emerging alternative to bar-and-hinge elasticity is \textit{homogenization}.  The general idea is to replace the complex micro-motions of the metamaterial's building blocks by an effective field theory that  captures the material’s collective elastic interactions.  The earliest success in this direction is  perhaps the work of Alibert \textit{et al.}\;\cite{alibert2003truss} in their study of truss beams with pentographic substructure. Using the asymptotic variational method of $\Gamma$-convergence \cite{braides2002gamma}, they demonstrated that these structures coarse-grain to strain-gradient and higher-order elastic theories, assuming a linear elastic response.   They later generalized these results to  a myriad of other linear truss models  of metamaterials, including origami and kirigami \cite{abdoul2018strain,durand2022predictive}, placing higher-order  continuum theories   on a rigorous foundation in this setting.  A wide variety of researchers from engineering, mathematics and  physics have added to this literature \cite{ariza2024homogenization, eskandari2024unravelling, giomi2025stretching,nassar2020microtwist, saremi2020topological,sun2020continuum, vasudevan2024homogenization, ye2024asymptotic}, illustrating the richness and utility of coupling  the metamaterial's design to its effective elastic properties through homogenization. However, all of this research assumes up front that the metamaterial behaves linearly, and thus has only limited  applicability  to examples like parallelogram origami that dramatically change their shape.  

We believe geometry should play the leading role when coarse-graining flexible metamaterials that possess mechanism deformations.  For parallelogram origami, a basic physical heuristic is that their soft modes are \textit{locally mechanistic}: at a lengthscale comparable to the size of a unit cell, each such mode looks like a mechanism. However, these mechanistic features can vary on the scale of many cells to produce a heterogenous actuation. The task then is to quantify exactly how the locally mechanistic response informs the overall  shape. The first efforts in this direction are due to  Schenk and Guest \cite{schenk2013geometry}  and Wei \textit{et al.}\;\cite{wei2013geometric}. By studying the kinematics of fitting together slightly bent neighboring unit cells, both works uncovered a geometric link between the Miura's in-­plane Poisson’s ratio and its normal curvatures, and used this link to explain why the Miura prefers  a  saddle shape when bent. Nassar \textit{et al.}\;\cite{nassar2017curvature} built on these ideas by  deriving the same Poisson's ratio link for Eggbox origami. More perceptively, they viewed this link as encoding a global kinematic restriction on the effective surfaces possible in Eggbox origami, namely, by coupling it to the Gauss and Codazzi-Mainardi compatibility conditions from differential geometry \cite{do2016differential}.  Multiple lines of research have since demonstrated a fundamental coupling of Poisson's ratio and normal curvatures in all parallelogram origami patterns \cite{mcinerney2022discrete,nassar2022strain}, in fact all periodic shells \cite{nassar2024periodic,nassar2024effective}, and connected this result to the origami's effective shapes through differential geometry \cite{czajkowski2023orisometry,nassar2022strain}.   Notably though, all of these results are a purely geometric attempt at capturing the soft modes observed in origami. They do not  seek to systematically quantify the elastic energy of these modes through asymptotic analysis. 

Our key theoretical contributions are to marry the nonlinear geometric constraints typical of flexible metamaterials with the systematic asymptotic analysis typical of linear homogenization. Our first work in this direction focused on the simpler 2D setting of planar kirigami \cite{zheng2022continuum}. There, we produced a coarse-graining rule that coupled  the actuation of the kirigami’s panels and slits to its effective deformation through a metric
constraint.  We also justified this rule by showing that, for any solution to this metric constraint, there is a corresponding sequence of planar kirigami deformations converging to the given effective deformation, with elastic energy far less than bulk.  We then turned our attention to parallelogram origami \cite{xu2024derivation}, which is more challenging due to the out-of-plane deformations that arise from panel bending.  In brief,  the coupling  of effective deformation and cell-by-cell actuation is done through a metric constraint, similar to the kirigami setting, while the origami's rigidity in bending constrains the curvature of the pattern to its Poisson's ratio in line with prior results.  These conditions can be expressed in a fully nonlinear way as two algebraic constraints coupling the first and second fundamental forms of the origami's effective deformation to an angle field  quantifying  its cell-by-cell actuation (Eqs.\;(\ref{eq:firstFundConstraint}) and (\ref{eq:secFundConstraint}) below). As our main technical achievement,  we showed that every solution to these compatibility equations is the limit of a sequence of soft modes --- origami deformations that involve large actuation of the folds, slight panel bending, and 
 negligible stretching. Furthermore, we used these sequences to derive a plate theory for parallelogram origami patterns with an explicit coarse-grained quadratic energy depending on the second fundamental form of the effective deformation and the gradient of the actuation angle field. Section \ref{sec:Theory} provides an overview of this theory. 

The purpose of the current work is to introduce a constitutive model for parallelogram origami that is faithful to the theory but also practical to simulate. A basic challenge is that it is hard to numerically implement the exact geometrical constraints that emerge when coarse-graining a flexible metamaterial. This is especially true under prototypical Dirichlet boundary conditions on effective deformation, as well as the tractions, moments and transverse loads  typical of any engineering analysis of a plate. While mathematical efforts are under way by one of the co-authors to develop convergent numerical schemes that exactly solve such geometric constraints \cite{marazzato2022mixed,marazzato2023computation,marazzato2024h2}, the schemes are not yet capable of handling a wide variety of boundary conditions of practical interest. Here we embrace an alternative ``relaxation and regularization" view that builds on the success of our work simulating the effective behavior of soft modes in planar kirigami \cite{zheng2023modelling}, as well as other successful approaches to continuum modeling in the kirigami setting \cite{czajkowski2022conformal,deng2020characterization,mcmahan2022effective, roy2023curvature}. Specifically,  we relax the geometric constraints ((\ref{eq:firstFundConstraint}) and (\ref{eq:secFundConstraint})) via corresponding elastic energy penalties, and simplify the plate energy density from the  theory to embrace its essential character as a regularization of the geometric penalties. The resulting model for parallelogram origami  is  a generalized elastic continuum reminiscent of, but distinct from,  Eringen's classical theories \cite{eringen2012microcontinuum}. It is nonlinear in the effective  deformation gradient and angle field  and regularized by high-order gradients thereof. Importantly, it is also variationally well-posed (see Theorem \ref{ExistenceTheorem}), and thus ready-made for a  numerical treatment based on finite elements.  We provide a finite element formulation of the model using the $C^0$ interior penalty method, a well-established numerical technique for handling the higher-order  gradients present in the model \cite{engel2002continuous}.   Implementation is then done via the open source computing platform \texttt{Firedrake} \cite{FiredrakeUserManual}.

Our new continuum  model and numerical method prove to be a versatile platform for solving elastic boundary
value problems.  This rings especially true in  Section \ref{sec:Examples}, where we showcase a compelling variety of simulations and analysis of the large deformation response of Miura and Eggbox origami under loads.  One of the more interesting points exemplified by these examples   is the  strong coupling of the pattern's Poisson's ratio to its qualitative elastic behavior,  beyond curvature. Our model produces equilibrium solutions where, depending on the boundary value problem, certain components of effective deformation or actuation are driven to approximate geometric  partial differential equations (PDEs) that are either elliptic or hyperbolic,  based solely on whether the pattern is auxetic or not.  As elliptic and hyperbolic PDEs are dramatically different, so too is the behavior of  Miura  and Eggbox origami in their response to loads. Interestingly, the exact same connection between auxeticity and PDE type is also found in  planar kirigami \cite{zheng2022continuum}, suggesting perhaps that this characterization is a universal  feature of the effective behavior of flexible mechanical metamaterials. 

This paper is organized as follows. Section \ref{sec:Theory} describes the geometry of parallelogram origami and  summarizes the effective theory derived in \cite{xu2024derivation}. Section \ref{sec:modelForSimulations} introduces our generalized elastic continuum model, and its equilibrium equations and boundary conditions. Section \ref{sec:FEM} provides  the finite element formulation and discusses its implementation. Section \ref{sec:Examples} deals with examples and  Section \ref{sec:Conclusion} concludes the work. Throughout, we make liberal use of tensor notation for conciseness. Vectors are indicated by boldface lowercase letters ($\mathbf{a} \in \mathbb{R}^m$),   second-order tensors by boldface uppercase letters ($\mathbf{A} \in \mathbb{R}^{m \times n}$), third-order tensors by calligraphic uppercase letters ($\mathcal{A} \in \mathbb{R}^{m \times n \times p}$), and fourth-order tensors by ``blackboard bold" uppercase letters ($\mathbb{A} \in \mathbb{R}^{m \times n \times p \times q}$).

\section{Theoretical framework}\label{sec:Theory}


We begin by recalling key ideas from  \cite{xu2024derivation}, where we used asymptotic analysis  to derive an effective continuum elastic plate theory for parallelogram origami. This theory forms the basis of the constitutive model and numerical method introduced later on to simulate the mechanical response of these patterns. 

\subsection{Designs and an effective description of origami mechanisms}\label{ssec:Design}

 Each parallelogram origami metamaterial is built from a repeating unit cell of four  parallelograms panels joined at folds. As illustrated in Fig.\;\ref{Fig:1stFig}(a) and (c), the design of such a pattern is fully parameterized by four design vectors $\mathbf{t}^r_i \in \mathbb{R}^3$, $i = 1,2,3,4$, labeling the inner creases of a single unit cell in a counterclockwise fashion.  The parallelogram panels are formed by pairs of opposite sides described by $\mathbf{t}^r_i$ and $\mathbf{t}^r_{i+1}$  (Fig.\;\ref{Fig:1stFig}(a)). We assume throughout that these vectors correspond to a partially folded cell by enforcing the conditions 
 \begin{equation}
 \begin{aligned}\label{eq:MVConstraints}
 \mathbf{t}^r_i \cdot (\mathbf{t}^r_j \times \mathbf{t}_k^r) \neq 0 \quad \text{ for all } ijk \in \{123,234,341,412\}.
 \end{aligned}
 \end{equation}  
 The cell is  then tessellated along the Bravais lattice vectors 
 \begin{equation}
 \begin{aligned}
 \mathbf{u}_0 := \mathbf{t}_1^r - \mathbf{t}_3^{r}, \quad \mathbf{v}_0 := \mathbf{t}_2^r - \mathbf{t}_4^r
 \end{aligned}
 \end{equation} 
 to produce an overall pattern (Fig.\;\ref{Fig:1stFig}(c)). For future reference, we assume without loss of generality that $\mathbf{u}_0$ and $\mathbf{v}_0$ satisfy
 \begin{equation}
 \begin{aligned}
 \mathbf{u}_0 \cdot \mathbf{e}_3 = \mathbf{v}_0 \cdot \mathbf{e}_3 = 0, \quad \mathbf{e}_3 \cdot (\mathbf{u}_0 \times \mathbf{v}_0) > 0,
 \end{aligned}
 \end{equation} 
 and define their projection onto $\mathbb{R}^2$ as $\tilde{\mathbf{u}}_0, \tilde{\mathbf{v}}_0$ such that $\mathbf{u}_0 := (\tilde{\mathbf{u}}_0, 0)$ and $\mathbf{v}_0: = (\tilde{\mathbf{v}}_0, 0)$.  The design space of parallelogram origami is large. It includes many well-known examples like Miura origami \cite{koryo1985method}, Eggbox origami \cite{schenk2011origami} and the Morph patterns \cite{pratapa2019geometric}. The theory we outline treats the design vectors as input, and thus encompasses all these  patterns as special cases.
 
\begin{figure}[t!]
\centering
\includegraphics[width=1\textwidth]{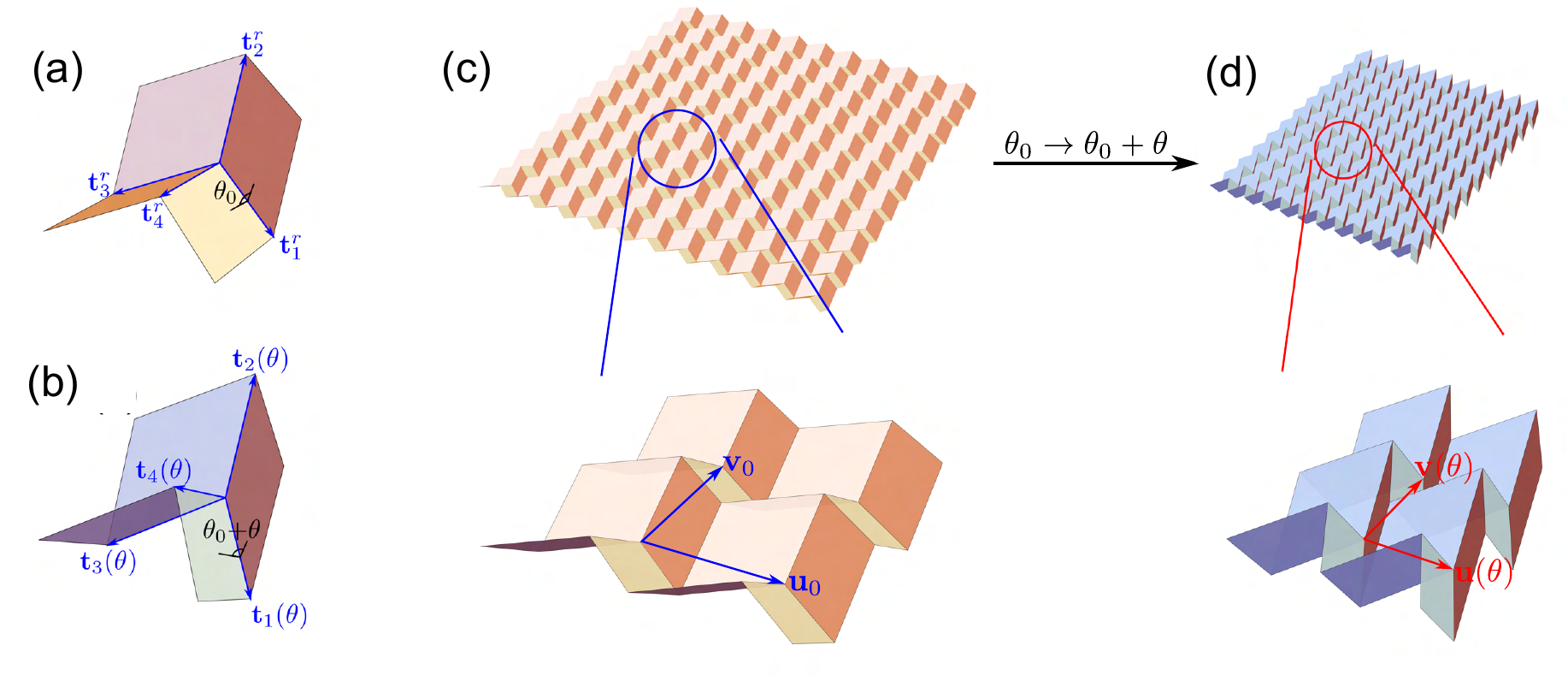} 
\caption{Parallelogram origami design and mechanism kinematics. (a) A  parallelogram origami unit cell and the labeling of its creases. (b) Parameterization of its mechanism deformations. (c) The overall pattern is obtained by tessellating a unit cell. (d) The mechanism of the pattern repeats that of the unit  cell.}
\label{Fig:1stFig}
\end{figure}

A signature property  of parallelogram origami is that each such pattern possesses a single degree-of-freedom (DOF) mechanism motion.  Let us focus on the single cell in Fig.\;\ref{Fig:1stFig}(a) with $\theta_0$ denoting the dihedral angle at the $\mathbf{t}^r_1$-crease. We fold this cell as origami by folding the crease from $\theta_0$ to $\theta_0 + \theta$ while keeping the panels rigid. This actuation  yields the deformed tangents in Fig.\;\ref{Fig:1stFig}(b)
\begin{equation}
\begin{aligned}
\mathbf{t}^r_i \to \mathbf{t}_i(\theta), \quad \theta \in (\theta^{-}, \theta^+) ,
\end{aligned}
\end{equation}
which satisfy $|\mathbf{t}_i(\theta)| = |\mathbf{t}_i^r|$ for each $i =1,\ldots,4$, and  vary smoothly in $\theta$, thus describing a  mechanism motion of the cell.\footnote{Explicit parameterizations of this actuation are obtained by solving kinematic compatibility conditions for the panels to rigidly rotate about the folds  of a four-fold vertex.  Such parameterizations are given in \cite{feng2020designs,tachi2009generalization} for rigidly and flat foldable origami and in \cite{foschi2022explicit} for the general Euclidean and non-Euclidean origami.}  The interval $(\theta^{-}, \theta^+)$ depends only on $\mathbf{t}_1^r, \ldots, \mathbf{t}_4^r$, and is characterized as the largest open interval  containing zero that satisfies  
 \begin{equation}
 \begin{aligned}\label{eq:noChangeMV}
\big[ \mathbf{t}^r_i \cdot (\mathbf{t}^r_j \times \mathbf{t}_k^r) \big]\big[ \mathbf{t}_i(\theta) \cdot (\mathbf{t}_j(\theta) \times \mathbf{t}_k(\theta)) \big] > 0  \quad \text{ for all } ijk \in \{123,234,341,412\}.
 \end{aligned}
 \end{equation}
 Such an interval exists for any parallelogram origami design satisfying (\ref{eq:MVConstraints}).\footnote{A proof of this fact is given in \cite{xu2024derivation}; it uses an argument based on  the implicit function theorem.} The overall pattern in Fig.\;\ref{Fig:1stFig}(c), in turn, has a mechanism motion that repeats the motion of the cell (Fig.\;\ref{Fig:1stFig}(d)). It is obtained by tessellating a deformed cell along the Bravais lattice vectors 
 \begin{equation}
 \begin{aligned}
 \mathbf{u}(\theta) := \mathbf{t}_1(\theta) - \mathbf{t}_3(\theta), \quad \mathbf{v}(\theta) := \mathbf{t}_2(\theta) - \mathbf{t}_4(\theta), \quad \theta \in (\theta^{-}, \theta^+).
 \end{aligned}
 \end{equation} 
As in the reference state, we assume without loss of generality that 
  \begin{equation}
 \begin{aligned}\label{eq:planarMotionDef}
 \mathbf{u}(\theta) \cdot \mathbf{e}_3 = \mathbf{v}(\theta) \cdot \mathbf{e}_3 = 0, \quad \mathbf{e}_3 \cdot (\mathbf{u}(\theta) \times \mathbf{v}(\theta)) > 0
 \end{aligned}
 \end{equation}
to express the basic fact that the mechanism motion is effectively planar.   Furthermore, we define the \textit{shape tensor} as the unique linear transformation $\mathbf{A}(\theta) \in \mathbb{R}^{3\times2}$ such that
\begin{equation}
\begin{aligned}\label{eq:shapeTensor}
\mathbf{A}(\theta) \tilde{\mathbf{u}}_0 = \mathbf{u}(\theta), \quad \mathbf{A}(\theta) \tilde{\mathbf{v}}_0 = \mathbf{v}(\theta).
\end{aligned}
\end{equation}
Note that $\mathbf{e}_3 \cdot \mathbf{A}(\theta) = \mathbf{0}$ per (\ref{eq:planarMotionDef}).

\subsection{Coarse-graining rules and an effective description of origami soft modes}

 \begin{figure}[t!]
\centering
\includegraphics[width=.93\textwidth]{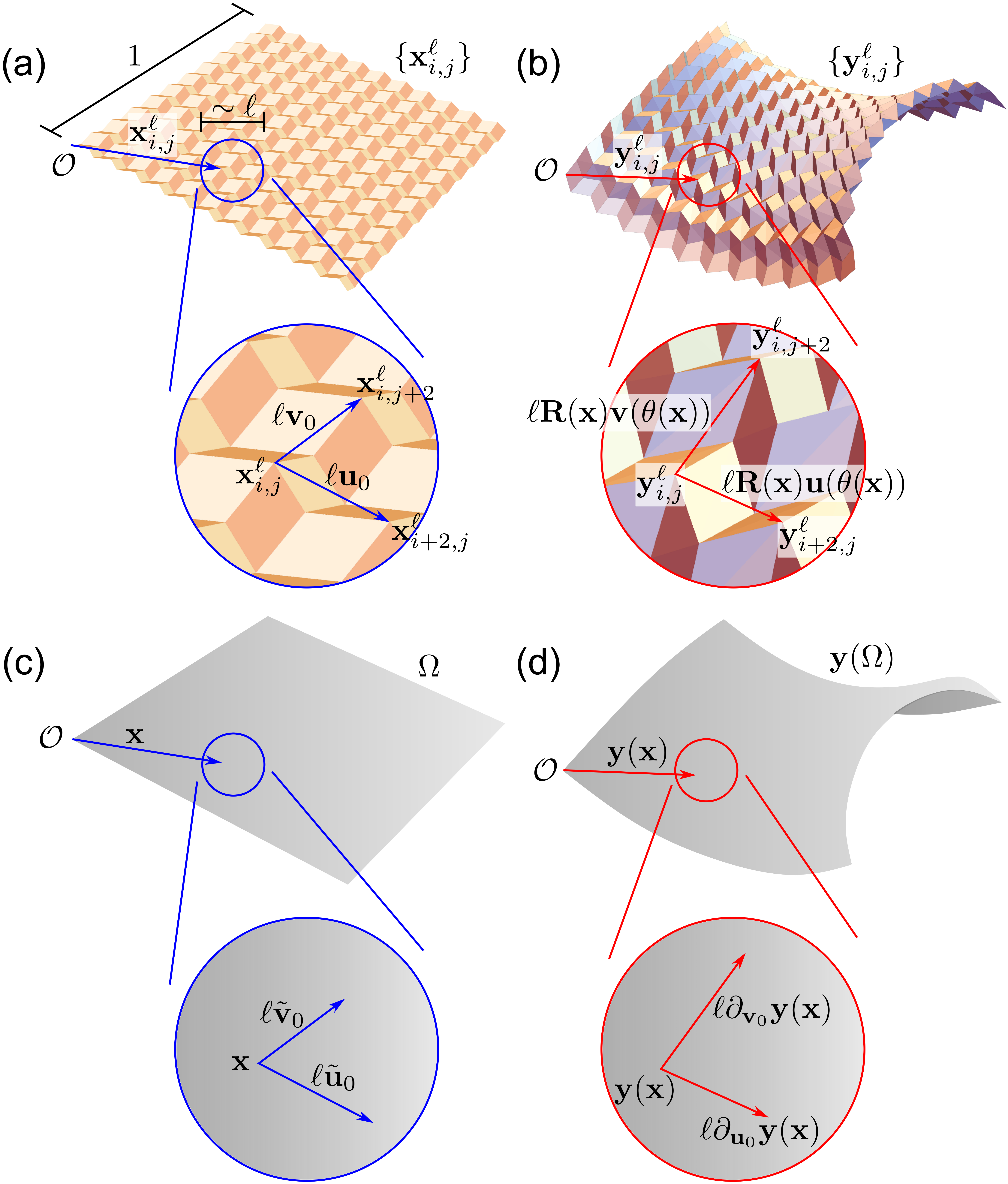} 
\caption{Illustration of the locally mechanistic coarse-graining rule. (a) A finely patterned parallelogram origami, labeled $\{ \mathbf{x}_{i,j}^{\ell}\}$, has  cells of size $\sim \ell$ on an overall domain of size $\sim 1$. (b) A soft mode of this pattern deforms the vertices $\{ \mathbf{x}_{i,j}^{\ell} \} \mapsto \{ \mathbf{y}_{i,j}^{\ell}\}$ consistent with a spatially varying local mechanism. (c) The effective midsurface of the origami, denoted $\Omega$, deforms smoothly to  $\mathbf{y}(\Omega)$ under a soft mode. (d) The tangent vectors induced by the effective deformation match the locally mechanistic response in (b).}
\label{Fig:RefDef}
\end{figure}

Beyond the pure mechanisms, parallelogram origami can bend, twist and actuate in complex  nonuniform ways at very little elastic energy by localizing most of the deformation at the folds. We call such deformations soft modes.  As anyone who has folded a Miura origami knows, soft modes are the generic response of these patterns to loads. Here, we describe  a  theory that captures the effective behavior of all possible soft modes in parallelogram origami.  

Let's begin with some notation. A typical parallelogram origami metamaterial is a pattern of $M \times N$ repeating unit cells of length $|\mathbf{u}_0|$ and width $|\mathbf{v}_0|$ as in Fig.\;\ref{Fig:1stFig}(c). Our  starting point for coarse graining these metamaterials is shown in Fig.\;\ref{Fig:RefDef}(a-b).  We non-dimensionalize  by fixing the width of the reference configuration of the origami  pattern as $1$,  thus  introducing a characteristic cell size $\ell =1/N$.   The vertices of this reference configuration are labeled by $\{ \mathbf{x}_{i,j}^{\ell} \} \subset \mathbb{R}^3$, which indexes them by $i$  along the length and $j$ along the width. A corresponding soft mode of deformation takes each point  $\mathbf{x}_{i,j}^{\ell}$ to a point  $\mathbf{y}_{i,j}^{\ell}$ in 3D space, and is  denoted by $\{ \mathbf{y}_{i,j}^{\ell}\} \subset \mathbb{R}^3$.  In \cite{xu2024derivation}, we showed that soft modes  admit a  continuum description quantified by the cell size $\ell$. Fig.\;\ref{Fig:RefDef}(c-d) illustrates the main idea. At a basic level, we seek to replace the pattern  $\{ \mathbf{x}_{i,j}^{\ell}\}$ by a planer  continuum domain $\Omega \subset \mathbb{R}^2$, reflecting its ``midsurface", with the goal of approximating  the gross shape change of its soft modes $\{ \mathbf{y}_{i,j}^{\ell}\}$  by continuum deformations of this surface $\mathbf{y} \colon \Omega \rightarrow \mathbb{R}^3$. We call such deformations  \textit{effective deformations} from hereon.   As explained below,  this continuum approach is made meaningful through two coarse-graining rules that link the design and kinematics of the origami's cells explicitly to its  effective  deformations.

The first coarse-graining rule originates from our earlier work on planar kirigami \cite{zheng2022continuum, zheng2023modelling}.  Soft modes are \textit{locally mechanistic}. At the scale of each unit cell they look like a mechanism. However, their features vary slowly from cell to cell. This  rule reveals a fundamental relationship between the first fundamental form of the effective deformation  of a soft mode
\begin{equation}
\begin{aligned}\label{eq:1stFund}
\mathbf{I}(\mathbf{y}) := (\nabla \mathbf{y})^T \nabla \mathbf{y}
\end{aligned}
\end{equation}
and the cell-by-cell  actuation $\theta = \theta(\mathbf{x})$, which we now think of as a spatially varying field depending on $\mathbf{x} = (x_1, x_2) \in \Omega$. Here and throughout, the gradient operator $\nabla$ denotes differentiation with respect to the planar domain $\Omega$; thus $\nabla \mathbf{y} \in \mathbb{R}^{3\times2}$ and $\mathbf{I}(\mathbf{y}) \in \mathbb{R}^{2\times2}_{\text{sym}}$. 

To explain the coarse-graining rule, the vertex labeled $\mathbf{x}_{i,j}^{\ell}$ in Fig.\;\ref{Fig:RefDef}(a) corresponds to a material point $\mathbf{x} \in \Omega$ in Fig.\;\ref{Fig:RefDef}(c). A soft mode deforms this origami vertex to $\mathbf{y}_{i,j}^{\ell}$, as shown in Fig.\;\ref{Fig:RefDef}(b),  which  corresponds to a spatial point $\mathbf{y}(\mathbf{x})$ on the effective surface in Fig.\;\ref{Fig:RefDef}(d).   Notice that, while the origami surface oscillates due to the underlying microstructure of the unit cell, the deformation of every other vertex $\mathbf{x}_{i+2m,j+2n}^{\ell} \mapsto  \mathbf{y}_{i+2m,j+2n}^{\ell}$ is consistent with a smooth underlying surface.  As the corrugated features  are also consistent locally with that of a perfect mechanism, we can quantify the tangent planes of this surface by a rotation field $\mathbf{R} \colon \Omega \rightarrow SO(3)$ and actuation field $\theta \colon \Omega \rightarrow (\theta^{-}, \theta^+)$  that satisfy 
\begin{equation}
\begin{aligned}\label{eq:keyCoarseGrain10}
& \mathbf{y}_{i+2,j}^{\ell} - \mathbf{y}_{i,j}^{\ell}  \approx  \ell \mathbf{R}(\mathbf{x}) \mathbf{u}(\theta(\mathbf{x})), \qquad   \mathbf{y}_{i,j+2}^{\ell} - \mathbf{y}_{i,j}^{\ell}  \approx  \ell \mathbf{R}(\mathbf{x}) \mathbf{v}(\theta(\mathbf{x}))
\end{aligned}
\end{equation}
(compare  Fig.\;\ref{Fig:1stFig}(c-d)  to Fig.\;\ref{Fig:RefDef}(a-b)).  On the other hand, as a matter of standard continuum mechanics, we have that $\mathbf{y}$ in Fig.\;\ref{Fig:RefDef}(c-d) satisfies  
\begin{equation}
\begin{aligned}\label{eq:keyCoarseGrain11}
& \mathbf{y}(\mathbf{x} + \ell \tilde{\mathbf{u}}_0) - \mathbf{y}(\mathbf{x})   \approx  \ell \partial_{\mathbf{u}_0} \mathbf{y}(\mathbf{x}), \qquad   \mathbf{y}( \mathbf{x} + \ell \tilde{\mathbf{v}}_0) - \mathbf{y}(\mathbf{x}) \approx  \ell \partial_{\mathbf{v}_0} \mathbf{y}(\mathbf{x}).
\end{aligned}
\end{equation}
Our first coarse-graining rule is to couple the cell-by-cell kinematics to the effective deformation via the constraints
\begin{equation}
\begin{aligned}\label{eq:firstFundConstraint0}
\partial_{\mathbf{u}_0} \mathbf{y} = \mathbf{R} \mathbf{u}(\theta),  \qquad  \partial_{\mathbf{v}_0} \mathbf{y} = \mathbf{R} \mathbf{v}(\theta),
\end{aligned} 
\end{equation}
which match the righthand sides of (\ref{eq:keyCoarseGrain10}) and (\ref{eq:keyCoarseGrain11}). We can also eliminate the rotation field $\mathbf{R}$ and  instead couple $\theta$ to $\mathbf{I}(\mathbf{y})$ by writing  (\ref{eq:firstFundConstraint0}) equivalently  as  a metric constraint
\begin{equation}
\begin{aligned}\label{eq:firstFundConstraint}
\mathbf{I}(\mathbf{y})  = \mathbf{A}^T(\theta) \mathbf{A}(\theta) 
\end{aligned} 
\end{equation}
for the the shape tensor  in (\ref{eq:shapeTensor}).

We now turn to our second coarse-graining rule. While most of the deformation of a soft mode is localized at the folds, the panels must bend slightly to accommodate a non-uniform actuation from cell to cell. This interplay constrains the effective theory  by how neighboring unit cells with slightly bent panels fit together, leading to a fundamental relationship between the actuation $\theta$ and the second fundamental form of the effective deformation 
\begin{equation}
\begin{aligned}\label{eq:secFundDef}
\mathbf{II}(\mathbf{y}) := \begin{pmatrix} \partial_1 \partial_1 \mathbf{y} \cdot \mathbf{n}(\mathbf{y}) & \partial_1 \partial_2 \mathbf{y} \cdot \mathbf{n}(\mathbf{y})  \\   \partial_1 \partial_2 \mathbf{y} \cdot \mathbf{n}(\mathbf{y}) &  \partial_2 \partial_2 \mathbf{y} \cdot \mathbf{n}(\mathbf{y})\end{pmatrix}  \quad \text{ for } \quad  \mathbf{n}(\mathbf{y}) := \frac{ \partial_1 \mathbf{y} \times \partial_2 \mathbf{y}}{| \partial_1 \mathbf{y} \times \partial_2 \mathbf{y}|}.
\end{aligned}
\end{equation}   

Fig.\;\ref{Fig:BendIdea} illustrates the idea. Zoom into the local neighborhood of cells in a soft mode and identify the  rotation of the base cell by $\mathbf{R}(\mathbf{x})$, as indicated, to be consistent with (\ref{eq:keyCoarseGrain10}) and Fig.\;\ref{Fig:RefDef}.  Given the slightly bent nature of all the panels in this description, the cells rotate  relative to each other by a small amount. We quantify this kinematics at leading order by noticing that the rotation of the cell to the right and above the base cell satisfy
\begin{equation}
\begin{aligned}
&\mathbf{R}(\mathbf{x} + \ell \tilde{\mathbf{u}}_0)  \approx   \mathbf{R}(\mathbf{x}) \big(\mathbf{I}+ \ell (\boldsymbol{\omega}_{\mathbf{u}_0}(\mathbf{x}) \times ) \big),  \qquad \mathbf{R}(\mathbf{x} + \ell \tilde{\mathbf{v}}_0)  \approx   \mathbf{R}(\mathbf{x}) \big(\mathbf{I}+ \ell (\boldsymbol{\omega}_{\mathbf{v}_0}(\mathbf{x}) \times ) \big),  
\end{aligned}
\end{equation}
respectively,  for some vectors fields  $\boldsymbol{\omega}_{\mathbf{u}_0}, \boldsymbol{\omega}_{\mathbf{v}_0} \colon \Omega \rightarrow \mathbb{R}^3$ that measure bending in the pattern.\footnote{Note,  $(\mathbf{v} \times)$ denotes  a tensor on  $\mathbb{R}^{3\times3}$ that  satisfies $(\mathbf{v} \times) \mathbf{w} = \mathbf{v} \times \mathbf{w}$ for all $\mathbf{w} \in \mathbb{R}^3$ for any $\mathbf{v} \in \mathbb{R}^3$; in other words, $(\mathbf{v} \times)$ is a generic parameterization of a skew tensor.} Since the cells must fit together at their boundaries,  the fields $\boldsymbol{\omega}_{\mathbf{u}_0}$, $\boldsymbol{\omega}_{\mathbf{v}_0}$ and $\theta$ are not independent of each other. It turns out that they satisfy (see \cite{xu2024derivation})
\begin{equation}
\begin{aligned}
\boldsymbol{\omega}_{\mathbf{u}_0}\cdot \mathbf{v}'(\theta) = \boldsymbol{\omega}_{\mathbf{v}_0} \cdot \mathbf{u}'(\theta).
\end{aligned}
\end{equation}
As a final bit of manipulation, this constraint can be rewritten using the implied identity $\partial_{\mathbf{u}_0}\big[ \mathbf{R} \mathbf{v}(\theta)\big] = \partial_{\mathbf{v}_0}\big[ \mathbf{R}\mathbf{u}(\theta)\big]$ from (\ref{eq:firstFundConstraint0}). Doing so  eliminates the bending measures $\boldsymbol{\omega}_{\mathbf{u}_0}$, $\boldsymbol{\omega}_{\mathbf{v}_0}$ in favor of  the second fundamental form $\mathbf{II}(\mathbf{y})$, giving an equivalent constraint  in terms of the fields $\mathbf{y}$ and $\theta$:
\begin{equation}
\begin{aligned}\label{eq:secFundConstraint}
\big[ \mathbf{v}(\theta) \cdot \mathbf{v}'(\theta) \big] \big[ \tilde{\mathbf{u}}_0 \cdot \mathbf{II}(\mathbf{y}) \tilde{\mathbf{u}}_0 \big]  + \big[ \mathbf{u}(\theta) \cdot \mathbf{u}'(\theta) \big] \big[ \tilde{\mathbf{v}}_0 \cdot \mathbf{II}(\mathbf{y}) \tilde{\mathbf{v}}_0 \big]  = 0.
\end{aligned}
\end{equation}

 \begin{figure}[t!]
\centering
\includegraphics[width=.75\textwidth]{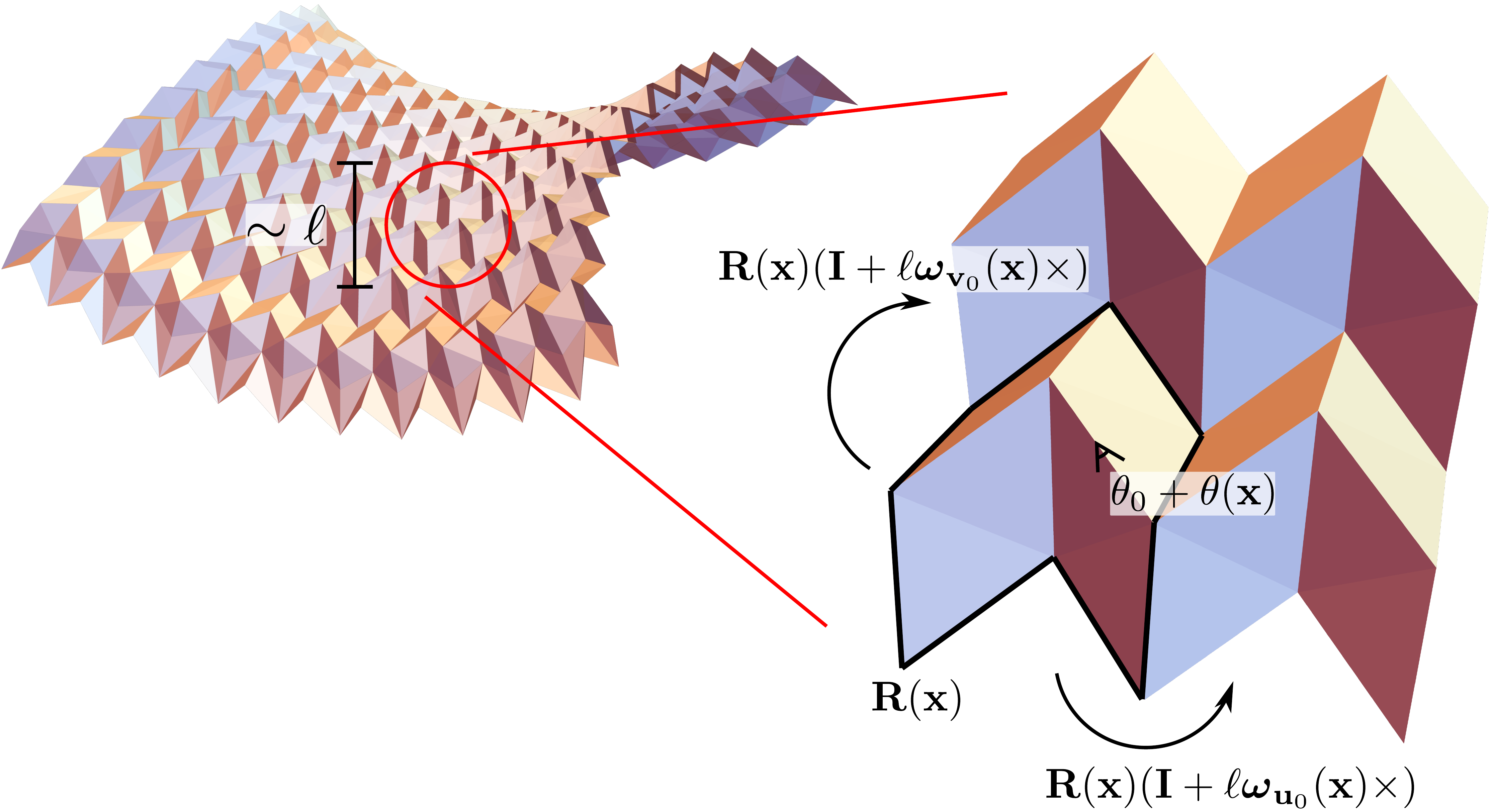} 
\caption{Local fitting problem for slightly bent unit cells. Neighboring cells rotate relative to each other due to panel bending. These rotations  are coupled through the compatibility conditions at the cell boundaries.}
\label{Fig:BendIdea}
\end{figure}

\subsection{Asymptotic analysis and an effective plate theory}
The bulk of our effort in \cite{xu2024derivation} goes into justifying the constraints in (\ref{eq:firstFundConstraint}) and (\ref{eq:secFundConstraint}). We showed that every sufficiently smooth solution
to these effective compatibility equations is the limit of a sequence of soft modes --- origami deformations with negligible panel strain that approximate the effective deformation.  The precise result is as follows: Let $\{ \mathbf{x}_{i,j}^{\ell}\} \subset \mathbb{R}^3$ and $\Omega \subset \mathbb{R}^2$ be the origami pattern and its effective midsurface as  above, so that the cell size $\ell$ is $\ll 1$ and the overall pattern has dimensions $|\Omega|\sim 1$.  Then, for any sufficiently smooth $\theta \colon \Omega \rightarrow (\theta^{-}, \theta^+)$ and $\mathbf{y} \colon \Omega \rightarrow \mathbb{R}^3$ that solve  (\ref{eq:firstFundConstraint}) and (\ref{eq:secFundConstraint}),  there is a corresponding deformation of the  origami vertices  $\{ \mathbf{x}_{i,j}^{\ell} \} \mapsto  \{ \mathbf{y}_{i,j}^{\ell} \}$  such that
\begin{equation}
\begin{aligned}\label{eq:mainResult}
&\frac{| \mathbf{y}_{i,j}^{\ell} -  \mathbf{y}_{i+1,j}^{\ell} | - | \mathbf{x}_{i,j}^{\ell} -  \mathbf{x}_{i+1,j}^{\ell} |}{| \mathbf{x}_{i,j}^{\ell} -  \mathbf{x}_{i+1,j}^{\ell} |} + \frac{| \mathbf{y}_{i,j}^{\ell} -  \mathbf{y}_{i,j+1}^{\ell} | - | \mathbf{x}_{i,j}^{\ell} -  \mathbf{x}_{i,j+1}^{\ell} |}{| \mathbf{x}_{i,j}^{\ell} -  \mathbf{x}_{i,j+1}^{\ell} |}      \\
& \qquad + \frac{| \mathbf{y}_{i,j}^{\ell} -  \mathbf{y}_{i+1,j+1}^{\ell} | - | \mathbf{x}_{i,j}^{\ell} -  \mathbf{x}_{i+1,j+1}^{\ell} |}{| \mathbf{x}_{i,j}^{\ell} -  \mathbf{x}_{i+1,j+1}^{\ell} |}  + \frac{| \mathbf{y}_{i+1,j}^{\ell} -  \mathbf{y}_{i,j+1}^{\ell} | - | \mathbf{x}_{i+1,j}^{\ell} -  \mathbf{x}_{i,j+1}^{\ell} |}{| \mathbf{x}_{i+1,j}^{\ell} -  \mathbf{x}_{i,j+1}^{\ell} |}  = O(\ell^2), \\ 
&|\mathbf{y} (\mathbf{x}) - \mathbf{y}_{i,j}^{\ell}|   = O(\ell)
\end{aligned}
\end{equation}
 for all  $(i,j)$, where $\mathbf{x} \in \Omega$ and $\mathbf{x}_{i,j}^{\ell}$ are corresponding points (see Fig.\;\ref{Fig:RefDef}(a) and (c)). The first statement in (\ref{eq:mainResult}) says that the strain of every panel is at most $\sim \ell^2$. The second says that the origami deformation approximates the effective deformation. 
 
In fact, we can say more.  As another key result of  \cite{xu2024derivation}, we derived an effective plate theory for parallelogram origami as the asymptotic limit of a large family of prototypical bar-and-hinge models of the origami  \cite{filipov2017bar,liu2017nonlinear}. Section 2.2 of \cite{xu2024derivation} explains the setup of our derivation in full detail. In brief, such models are composed of a  sum of stretching, bending and folding elastic energies, denoted  $E_{\text{tot}}^{\ell}(\{ \mathbf{y}_{i,j}^{\ell}\} ) = E^{\ell}_{\text{str}}(\{ \mathbf{y}_{i,j}^{\ell}\} )  + E^{\ell}_{\text{bend}}(\{ \mathbf{y}_{i,j}^{\ell}\} )  + E^{\ell}_{\text{fold}}(\{ \mathbf{y}_{i,j}^{\ell}\} )$.  The stretching elasticity is an  energy that scales with the magnitude of the four  panel strains in the first equation in (\ref{eq:mainResult}). The folding elasticity penalizes the deviations of each fold's dihedral angle from one set by the origami's reference configuration. Finally, the bending elasticity is modeled by adding ``stiff" folds along the panel diagonals and measuring the deviation of these folds from flat under a deformation. 

After scaling the moduli of the bar-and-hinge model by powers of $\ell$ such that  ``stretching stiffness" $\gg$ ``bending stiffness" $\gg$ ``folding stiffness", we showed that the sequence of origami constructions in (\ref{eq:mainResult}) admits the following limit energy:
 \begin{equation}
\begin{aligned}\label{eq:effectivePlateTheory}
\lim_{\ell \rightarrow 0} E^{\ell}_{\text{tot}}( \{ \mathbf{y}_{i,j}^{\ell}\}) =  \begin{cases}
\int_{\Omega} \begin{pmatrix} \mathbf{II}(\mathbf{y}) \\ (\nabla \theta)^T \end{pmatrix} \colon \mathbb{K}(\theta) \colon  \begin{pmatrix} \mathbf{II}(\mathbf{y}) \\ (\nabla \theta)^T \end{pmatrix} \dif x & \text{ if } ( \theta, \mathbf{y})   \in A_{\text{eff}} \\
+ \infty & \text{ otherwise,}
\end{cases}
\end{aligned}
\end{equation}
where $\mathbb{K}(\theta) \in \mathbb{R}^{3 \times 2 \times 3 \times 2}$ is an explicit fourth order tensor with major symmetry $[ \mathbb{K}(\theta)]_{i\alpha j \beta} = [\mathbb{K}(\theta)]_{j \beta i \alpha}$ that depends kinematically only on the actuation field $\theta$ (see \cite{xu2024derivation} for the full definition)\footnote{In bar-and-hinge elasticity, there is always a question of which panel diagonal to choose for the stiff folds in the bending energy. We considered all possible choices in  \cite{xu2024derivation} and showed that the  limiting plate theory is of the form in (\ref{eq:effectivePlateTheory}) for every case. However, the moduli $\mathbb{K}(\theta)$ depends on our choices for the panel diagonals. All else being equal, $\mathbb{K}(\theta)$ is stiffer if the longer panel diagonals are chosen and softer for the shorter ones. This result is consistent with the heuristic that shorter panel diagonals are the more physically sensible modeling choice.}
and $( \mathbf{II}(\mathbf{y}) , \nabla \theta)^T  \in \mathbb{R}^{3\times2}$ denote the ``bending" strains of the effective plate.\footnote{As a heuristic for why  $\nabla \theta$ is a bending strain, consider the bowtie example in Fig.\;\ref{Fig:IntroFig}(a). This origami deformation is effectively flat, which means the corresponding effective deformation in our theory satisfies $\mathbf{II}(\mathbf{y}) = \mathbf{0}$. However, the panels are bending to achieve this configuration and thus storing some elastic energy. This energy is quantified by $\nabla \theta$  in our plate theory in   (\ref{eq:effectivePlateTheory}). }
The  set $A_{\text{eff}}$ denotes the class of  admissible effective deformations and actuation fields describing a soft mode
\begin{equation}
\begin{aligned}
A_{\text{eff}} := \big\{ &(\mathbf{y} , \theta) \in H^2(\Omega, \mathbb{R}^3) \times H^1(\Omega, (\theta^{-}, \theta^+ )) \text{ subject to  (\ref{eq:firstFundConstraint}) and (\ref{eq:secFundConstraint}) a.e.\;in $\Omega$} \big\},
\end{aligned}
\end{equation}
where $H^k$ is the space of square integrable functions whose $1^{\text{st}}, \ldots,$ and $k^{\text{th}}$-order partial derivatives are also square integrable.

To derive this plate theory, we focused  on smooth effective fields  within $A_{\text{eff}}$. For any such fields, we showed  that the total energy of the corresponding origami construction $\{ \mathbf{y}_{i,j}^{\ell}\}$ is dominated by panel bending $E_{\text{tot}}^{\ell}(\{ \mathbf{y}_{i,j}^{\ell}\} ) \approx E_{\text{bend}}^{\ell}(\{ \mathbf{y}_{i,j}^{\ell}\} ) \sim 1$, yielding the finite limit in the first part of the statement in (\ref{eq:effectivePlateTheory}). However, when the effective fields depart  this set, the panel strain of the constructions is $O(\ell)$ or larger, as opposed to the  $O(\ell^2)$ (see (\ref{eq:mainResult})). The total energy is then dominated by panel stretching $E_{\text{tot}}^{\ell}(\{ \mathbf{y}_{i,j}^{\ell}\} ) \approx E_{\text{str}}^{\ell}(\{ \mathbf{y}_{i,j}^{\ell}\} ) \gg 1$, which is unbounded in the limit $\ell \rightarrow 0$ as indicated by the ``$+ \infty$" in (\ref{eq:effectivePlateTheory}). The main point is that soft modes in parallelogram origami are characterized by  effective fields that belong to $A_{\text{eff}}$ and possess a quadratic energy density in $\mathbf{II}(\mathbf{y})$ and $\nabla \theta$.

\section{Modeling framework}\label{sec:modelForSimulations}

Building on the theory, we now develop a constitutive model well-adapted to simulate the effective behavior of soft modes in parallelogram origami. A basic challenge is that it is hard to numerically implement exact geometrical constraints, like the ones in $A_{\text{eff}}$. This is especially true under prototypical displacement and traction boundary conditions.  We seek instead a practical continuum constitutive model that is faithful to the theory but also efficient to simulate and predictive.

\subsection{Constitutive model}\label{ssec:modelIntro}

We take as our effective constitutive model the sum of  bulk and regularizing terms 
\begin{equation}
\begin{aligned}\label{eq:Eint}
E_{\text{int}}(\mathbf{y}, \theta) := E_{\text{bulk}}(\mathbf{y}, \theta) +  E_{\text{reg}}(\mathbf{y}, \theta) ,
\end{aligned}
\end{equation} 
 measuring the stored elastic energy due to the effective deformation and actuation of a parallelogram origami pattern on an underlying reference domain $\Omega \subset \mathbb{R}^2$.
 
We choose the bulk energy in this formulation  to relax the geometric constraints in $A_{\text{eff}}$ via
\begin{equation}
\begin{aligned}\label{eq:bulkTerm}
E_{\text{bulk}}(\mathbf{y}, \theta) &:= \int_{\Omega} \Big\{ W_{1}(  \theta, \mathbf{I}(\mathbf{y}))   + W_{2}( \theta, \mathbf{II}(\mathbf{y})) \Big\} \dif x.
\end{aligned}
\end{equation}
In particular, the two energy densities $W_{1,2} \colon  (\theta^{-}, \theta^+)  \times  \mathbb{R}_{\text{sym}}^{2\times2}\rightarrow \mathbb{R}$  are penalizations of the constraints in (\ref{eq:firstFundConstraint}) and (\ref{eq:secFundConstraint}) defined by 
\begin{equation}
\begin{aligned}\label{eq:W12Def}
&W_{1}( \theta,  \mathbf{G}) := \frac{c_1}{\sqrt{\det \mathbf{G}}} \Big|\mathbf{G} - \mathbf{A}^T(\theta) \mathbf{A}(\theta)\Big|^2 , \\
&W_{2}(  \theta, \mathbf{K}) :=   \frac{c_2L_{\Omega}^2}{|\tilde{\mathbf{u}}_0|^4 |\tilde{\mathbf{v}}_0|^4} \Big(\big[ \mathbf{v}(\theta) \cdot \mathbf{v}'(\theta)\big]  \big[\tilde{\mathbf{u}}_0 \cdot  \mathbf{K} \tilde{\mathbf{u}}_0\big]|  + \big[\mathbf{u}(\theta) \cdot \mathbf{u}'(\theta)\big] \big[ \tilde{\mathbf{v}}_0 \cdot  \mathbf{K} \tilde{\mathbf{v}}_0 \big]  \Big)^2 .
\end{aligned}
\end{equation}
The division by $\sqrt{\det \mathbf{I}(\mathbf{y})}$ in the $W_{1}(\theta, \mathbf{I}(\mathbf{y}))$  seeks to prevent interpenetration of matter in the model, while the division by $ |\tilde{\mathbf{u}}_0|^4 |\tilde{\mathbf{v}}_0|^4$ in $W_2(\theta, \mathbf{II}(\mathbf{y}))$ scales this energy density so that it is invariant under conformal transformations of the unit cell given by  $(\tilde{\mathbf{u}}_0, \tilde{\mathbf{v}}_0) \rightarrow \lambda  ( \tilde{\mathbf{u}}_0, \tilde{\mathbf{v}}_0)$ and $(\mathbf{u}(\theta), \mathbf{v}(\theta) ) \rightarrow \lambda ( \mathbf{u}(\theta), \mathbf{v}(\theta))$ for all $\lambda > 0$.   Finally, $L_{\Omega}>0$ denotes a characteristic lengthscale of the domain $\Omega$. It is introduced in $W_{2}$ so that $c_2$ has  dimensions of energy per unit area, just like $c_1$. 

Having captured the key nonlinearities in the bulk term, we invoke the adage that ``simplicity is best" for the regularizing term by defining
 \begin{equation}
\begin{aligned}\label{eq:regTerm}
E_{\text{reg}}(\mathbf{y},\theta) := \int_{\Omega} \Big\{ d_1 L_{\Omega}^2 |\nabla \nabla \mathbf{y}|^2 + d_2 L_{\Omega}^2 |\nabla \theta|^2 + d_3  \theta^2 \Big\} \dif x.
\end{aligned}
\end{equation}
The first and second energy densities in (\ref{eq:regTerm}) are simple expressions that serve to both approximate the plate energy density in (\ref{eq:effectivePlateTheory}) and regularize the bulk energy in (\ref{eq:bulkTerm}).  They involve higher order gradients and thus, like $W_2$,   are scaled by   $L_{\Omega}^2$  to set the dimensions of $d_1,d_2$ as energy per unit area.  The last term is the simplest effective energy density that accounts for  the  higher order but ubiquitous  elastic energy to fold origami.

As this effective model is built off of the coarse-grained theory in Section \ref{sec:Theory}, the moduli $c_1, \ldots, d_3$ can be  linked to the elastic and geometric properties of the panels and folds, along with the cell size $\ell$, which we now briefly explain.
Set $W_i = c_i \hat{W}_i$ for $i = 1,2$ in (\ref{eq:W12Def}), so that each $\hat{W}_i$ is dimensionless. Remember these energy densities reflect the geometric constraints in (\ref{eq:firstFundConstraint}) and (\ref{eq:secFundConstraint}) for small panel strain in the  origami. In fact, our analysis in \cite{xu2024derivation} shows that  when $\hat{W}_1 \sim 1$, the panel strains of the origami  must be $\sim 1$.  The energy density in this setting is thus $c_1 \sim c_1 \hat{W}_1  \sim \text{``elastic energy density of $O(1)$ panel strains"} \sim \mu$, where $\mu$ is the shear modulus of the panels.  So $c_1 \sim \mu$.   Next, if $\hat{W}_1  = 0$ and $\hat{W}_2 \sim1$, then the panel strains of the corresponding origami are $\sim \ell$.  This setting gives an  energy density  $c_2 \sim c_2 \hat{W}_2  \sim  \text{``elastic energy density of $O(\ell)$ panel strains"} \sim \mu \times (\text{``panel strains"})^2 \sim \mu \ell^2$ provided $\ell \ll1$. We conclude that $c_2 \sim \mu \ell^2$. Next, suppose  the normalized energy densities satisfy  $\hat{W}_{1,2} =0$ and $L_{\Omega}^2 |\nabla \nabla \mathbf{y}|^2,  L_{\Omega}^2 |\nabla \theta|^2 \sim 1$. This setting is generally consistent with a soft mode of the origami for which both energy densities $d_1 L_{\Omega}^2 |\nabla \nabla \mathbf{y}|^2$ and $d_2 L_{\Omega}^2 |\nabla \theta|^2$ are comparable to the bending energy density stored in the panels $\sim \mu h^2 \kappa^2$, where $\kappa \sim 1/L_{\Omega}$ is the typical panel curvature and $h$ is its thickness.  We deduce from this case that $d_1 \sim d_2 \sim \mu \frac{h^2}{L_{\Omega}^2}$. Finally, the folding energy density of origami is quantified by $\sim \big(\frac{\text{``fold area per cell"}}{\text{``cell area"}}\big) \times (\text{``elastic energy density of the folded material"})$.   As the strain inside the folds is $\sim 1$ when $\theta \sim 1$, we conclude for large folding that $d_3 \sim d_3 |\theta|^2 \sim \frac{w L_{\Omega}\ell}{L_{\Omega}^2 \ell^2 } \mu_{\text{f}} \sim \frac{w}{L_{\Omega}} \ell^{-1} \mu_{\text{f}}$, where $\mu_{\text{f}}$ is the shear modulus of the folded material and $w$ is its width. All told, we obtain the following scaling laws for these moduli
\begin{equation}
\begin{aligned}
c_1 \sim \mu, \quad c_2 \sim \mu \ell^2, \quad d_1 \sim d_2 \sim \mu \Big(\frac{h}{L_{\Omega}}\Big)^2 , \quad d_3 \sim \mu_{\text{f}}  \frac{w}{L_{\Omega}} \ell^{-1} .
\end{aligned}
\end{equation} 

In practice, we anticipate fitting $c_1, c_2, d_1,d_2, d_3$  to experiments, rather than being too beholden to these scaling laws. Nevertheless, some general features are worth noting. Since an origami metamaterial satisfies $h/L_{\Omega} \ll \ell \ll 1$ in the typical case, these scaling laws suggest that these moduli  should obey the relations
\begin{equation}
\begin{aligned}\label{eq:modConsistency}
c_1 \gg c_2 \gg d_1 \sim d_2.
\end{aligned}
\end{equation} 
Notably absent in these relations is $d_3$. As a practical matter, we studied the regime of very soft folds in \cite{xu2024derivation} (corresponding to $d_3 \ll c_1, \ldots, d_2$)  to derive the plate theory in (\ref{eq:effectivePlateTheory}) by asymptotic analysis. However, the physically appropriate choice for this  modulus depends delicately on the properties of the panels and folds. In general, we expect $\mu_{\text{f}} \ll \mu$ because the folds are often made by perforating and/or plastically yielding the creases or are simply made of a more compliant material than the panels. In addition, $w$ typically satisfies  $w/L_{\Omega} \sim h/L_{\Omega} \ll \ell$.  So we can conclude that $d_{3} \ll c_1$; however, its size relative to $c_2, d_1, d_2$ is  not universal.  In Section \ref{sec:Examples}, for instance, we choose $d_3$ to be larger than $d_{1,2}$ and smaller than $c_2$ in most examples, as  this choice produces simulations that conform well to the paper models of the origami in Fig.\;\ref{Fig:IntroFig}. 

As a final point on the constitutive model, we take the ambient space  on which to define and study this energy to be 
\begin{equation}
\begin{aligned}\label{eq:VDef}
V^{\varepsilon} := \big\{ (\mathbf{y}, \theta) \in H^2(\Omega, \mathbb{R}^3) \times H^1(\Omega, \mathbb{R}) \colon |\partial_1 \mathbf{y} \times \partial_2\mathbf{y} | \geq \varepsilon \text{ and } \theta \in [ \theta^{-} + \varepsilon, \theta^+ - \varepsilon ] \text{ a.e.\;in } \Omega \big\}
\end{aligned}
\end{equation}
for a small positive parameter $\varepsilon > 0$. Since $\nabla \nabla \mathbf{y}$ and $\nabla \theta$ must be square-integrable for $E_{\text{int}}(\mathbf{y}, \theta)$ to be well-defined and bounded, restricting our attention to deformations and actuation fields in $H^2 \times H^1$ is well-justified. The $\varepsilon$-inequalities are not strictly needed to obtain a bounded energy, but are physically motivated and help to simplify some technical issues in proving the well-posedness of the model.  For instance, the function $\mathbf{u}(\theta)$ and $\mathbf{v}(\theta)$ are well-defined and smooth on the interval $(\theta^{-}, \theta^+)$ but can become poorly behaved as $\theta \rightarrow \theta^{\pm}$. In fact, there are additional physics at play in these limits, corresponding to a change in mountain-valley assignment, that we do not account for in our modeling. The restriction $ \theta \in [ \theta^{-} + \varepsilon, \theta^+ - \varepsilon]$ rules out this behavior. The other inequality $|\partial_1 \mathbf{y} \times \partial_2\mathbf{y} | \geq \varepsilon$ ensures that the  surface normal $\mathbf{n}(\mathbf{y})$  and $1/\sqrt{\det \mathbf{I}(\mathbf{y})}$ are  always well-defined and, in particular, that the admissible class of effective surfaces are bounded away from ones that contain singularities. While both restrictions are  crucial to our proof of the well-posedness of the model, they have no impact on how we numerically implement the model in practice.

\subsection{Boundary conditions and applied forces} 

We investigate the effective behavior of parallelogram origami under Dirichlet boundary conditions, corresponding to prescribed deformation and slope, as well as applied forces, generalized moments, and distributed loads.   Let $\Gamma_\text{d}, \Gamma_{\text{n}} \subset \partial \Omega$ denote relatively open subsets of the boundary of $\Omega$. We prescribe the deformation and slope, respectively, via
\begin{equation}
\begin{aligned}
\mathbf{y} = \overline{\mathbf{y}} \text{ on } \Gamma_{\text{d}} \quad \text{ and } \quad (\nabla \mathbf{y}) \mathbf{n} = \overline{\mathbf{s}} \text{ on } \Gamma_{\text{n}}
\end{aligned}
\end{equation} 
for some $\overline{\mathbf{y}} \colon \Gamma_{\text{d}} \rightarrow \mathbb{R}^3$ and $\overline{\mathbf{s}} \colon \Gamma_{\text{n}} \rightarrow \mathbb{R}^3$. 

To incorporate applied forces, we  invoke the principle of minimum potential energy.  The total  potential energy is the sum of the stored elastic energy in (\ref{eq:Eint}) and the contribution associated to applied forces,  written here as 
\begin{equation}
\begin{aligned}
E(\mathbf{y}, \theta) := E_{\text{int}}( \mathbf{y}, \theta) + E_{\text{ext}}(\mathbf{y}) .
\end{aligned}
\end{equation} 
 Let $\Gamma_{\text{t}} := \partial \Omega \setminus \Gamma_{\text{d}}$ and $\Gamma_{\text{m}} := \partial \Omega \setminus \Gamma_{\text{n}}$ denote the boundary sets on which we have not prescribed the deformation or slope.  We allow for applied forces on these sets via the external energy potential 
\begin{equation}
\begin{aligned}\label{eq:Eext}
E_{\text{ext}}(\mathbf{y})  := - \int_{\Omega} \overline{\mathbf{b}} \cdot \mathbf{y} \dif x - \int_{\Gamma_{\text{t}}} \overline{\mathbf{t}} \cdot \mathbf{y} \dif s  - \int_{\Gamma_{\text{m}}} \overline{\mathbf{m}}  \cdot ( \nabla \mathbf{y} ) \mathbf{n} \dif s,
\end{aligned}
\end{equation}
where $\overline{\mathbf{b}} \colon \Omega \rightarrow \mathbb{R}^3$ denotes the prescribed referential distributed loads, $\overline{\mathbf{t}} \colon \Gamma_{\text{t}} \rightarrow \mathbb{R}^3$ the prescribed referential traction, and $\overline{\mathbf{m}} \colon \Gamma_{\text{m}} \rightarrow \mathbb{R}^{3\times 2}$ the prescribed referential generalized moment. 

Note that  the boundary conditions  correspond only to $\mathbf{y}$, $\nabla \mathbf{y}$ and their pair-conjugate generalized forces.  In principle, one could also  include a Dirichlet  boundary conditions on $\theta$ and its pair-conjugate torque.  We avoid this inclusion for simplicity and because  such boundary conditions seem hard to realize experimentally in parallelogram origami.\footnote{These boundary conditions are, however, relevant to planar kirigami metamaterials. A nice example comes from  \cite{deng2020characterization}, where they used ``slit actuation" boundary conditions to enable a domain wall in experiments on these systems.}

\subsection{Energy minimization and an existence theorem}

Having defined the model, boundary conditions, and forces, we now establish the existence of minimizers of the energy. Let 
\begin{equation}
\begin{aligned}\label{eq:EstarDef}
E^{\star} :=  \inf \big\{ E(\mathbf{y} , \theta) \colon  (\mathbf{y} ,\theta )  \in V^{\varepsilon}_{\Gamma} \big\} 
\end{aligned}
\end{equation}
for  $V^{\varepsilon}_{\Gamma}$ the subset of effective deformations and actuation fields in $V^{\varepsilon}$ that enforces the desired Dirichlet boundary conditions
\begin{equation}
\begin{aligned}\label{eq:VepsGamma}
V^{\varepsilon}_{\Gamma} := \big\{ (\mathbf{y}, \theta) \in V^{\varepsilon} \colon \mathbf{y} = \overline{\mathbf{y}} \text{ a.e.\;on } \Gamma_{\text{d}}  \text{ and  } (\nabla \mathbf{y}) \mathbf{n} = \overline{\mathbf{s}} \text{ a.e.\;on } \Gamma_{\text{n}} \big\}. 
\end{aligned}
\end{equation}
The equalities at the boundaries  $\Gamma_{\text{d}}$ and $\Gamma_{\text{n}}$ are to be understood in the sense of traces.\footnote{In particular, we assume that  $\overline{\mathbf{y}} \in H^\frac12(\Gamma_\text{d}, \mathbb{R}^3)$ and $\bar{\mathbf{s}} \in H^\frac12(\Gamma_{\text{n}},\mathbb{R}^3)$, as these  are the appropriate trace Sobolev spaces for such boundary conditions (see Evans \cite{evans2022partial}, Chapter 5, and Brezis \cite{brezis}, Chapter 9, for more details).}  We have the following result.
\begin{theorem}\label{ExistenceTheorem}
Assume that $\Gamma_{\emph{n}} \subset \Gamma_{\emph{d}}$ are measurable,  $\Gamma_{\emph{d}}$ has nonzero measure, and that $V^{\varepsilon}_{\Gamma}$ is non-empty. Let $\overline{\mathbf{b}} \in L^2(\Omega, \mathbb{R}^3)$, $\overline{\mathbf{t}} \in L^2(\Gamma_{\emph{t}}, \mathbb{R}^3)$, $\overline{\mathbf{m}} \in L^2(\Gamma_{\emph{m}}, \mathbb{R}^{3\times3})$. Then, there is a $(\mathbf{y}, \theta)  \in V^{\varepsilon}_{\Gamma}$ such that $E^{\star} = E(\mathbf{y}, \theta) $. 
\end{theorem}
\noindent A proof of this theorem is provided in Appendix \ref{sec:proof}. The key point  is that the energy $E(\mathbf{y}, \theta)$, $(\mathbf{y}, \theta) \in V_{\Gamma}^{\varepsilon}$, is  weakly lower semicontintinuous, a result that follows from a general classification of lower semicontinuous functionals in Dacorogna \cite{dacorogna2007direct}, Theorem 3.23. There is some subtlety to this point. The hypotheses needed to apply the Dacorogna result are not directly compatible with the $\varepsilon$-inequalities of $V_{\Gamma}^{\varepsilon}$. However, we are able to  modify  the energy in order to apply the theorem to our setting and then argue that the modification is, in fact, benign. The rest of the proof is more-or-less standard fare in the calculus of variation. We provide the details in the appendix for the benefit of the mechanics reader who is less familiar with the techniques. 


\subsection{Weak formulation and equilibrium equations}\label{ssec:Weakform}
We now derive the equilibrium equations for our  model by taking the first variation of the energy. For this purpose, we dispense with the $\varepsilon$-dependence of the space of admissible fields in (\ref{eq:VepsGamma}), as it is not needed to develop this result, and instead consider the space 
\begin{equation}
\begin{aligned}\label{eq:VGamma}
V_{\Gamma} := \big\{ (\mathbf{y}, \theta) \in H^2(\Omega, \mathbb{R}^3) \times H^1(\Omega, (\theta^{-}, \theta^+)) \colon \mathbf{y} = \overline{\mathbf{y}} \text{ a.e.\;on } \Gamma_{\text{d}}  \text{ and  } (\nabla \mathbf{y}) \mathbf{n} = \overline{\mathbf{s}} \text{ a.e.\;on } \Gamma_{\text{n}} \big\} .
\end{aligned}
\end{equation}
We also denote by $V_0$ the space of functions  in (\ref{eq:VGamma}) such that $\overline{\mathbf{y}} = \mathbf{0}$ and $\overline{\mathbf{s}} = \mathbf{0}$. 

Let $(\mathbf{y}, \theta) \in V_{\Gamma}$ be a local minimizer to the energy $E$. The first variation of this energy satisfies
\begin{equation}
\begin{aligned}\label{eq:firstVariation}
0 &= \frac{\dif}{\dif t}  E(\mathbf{y} + t \mathbf{w},  \theta + t \eta) \big|_{t=0}  \\
&= \int_{\Omega} \Big\{  \mathbf{P}(\mathbf{y}, \theta) \colon  \nabla \mathbf{w}   +  \big \langle \mathcal{H}(\mathbf{y}, \theta) ,  \nabla \nabla \mathbf{w} \big \rangle  + q(\mathbf{y}, \theta)  \eta  + \mathbf{j}(\theta) \cdot \nabla \eta  \Big\} \dif x \\
&\qquad  -  \int_{\Omega} \overline{\mathbf{b}} \cdot \mathbf{w}  \dif x - \int_{\Gamma_{\text{t}}} \overline{\mathbf{t}} \cdot \mathbf{w}  \dif s -   \int_{\Gamma_{\text{m}}} \overline{\mathbf{m}} \cdot (\nabla \mathbf{w} ) \mathbf{n}   \dif s
\end{aligned}
\end{equation} 
for all $(\mathbf{w}, \eta) \in V_0$, where the inner product $\langle  \cdot , \cdot \rangle$ is  $\langle \mathcal{S}, \mathcal{T} \rangle = [ \mathcal{S}]_{i\alpha \beta} [\mathcal{T}]_{i\alpha \beta}$ and $\mathbf{P}$, $\mathcal{H}$, $q$, and $\mathbf{j}$ are generalized stresses that depend on the deformation and actuation field. The first of these stresses is a traditional Piola-Kirchhoff  stress of the form
\begin{equation}
\begin{aligned}\label{eq:PiolaStress}
\mathbf{P}(\mathbf{y}, \theta) := \mathbf{P}_1(\mathbf{y}, \theta) + \mathbf{P}_2(\mathbf{y}, \theta)
\end{aligned}
\end{equation} 
for second-order tensor fields $\mathbf{P}_1$ and $\mathbf{P}_2$ on $\mathbb{R}^{3\times2}$ obtained by differentiating the energy densities $W_{1}$ and $W_2$  with respect to the deformation gradient $\nabla \mathbf{y}$.  This differentiation furnishes 
\begin{equation}
\begin{aligned}\label{eq:P1P2}
&\mathbf{P}_1(\mathbf{y}, \theta) := \frac{4c_1}{\sqrt{\det \mathbf{I}(\mathbf{y})}}  (\nabla \mathbf{y})\Big[ \big( \mathbf{I}(\mathbf{y}) - \mathbf{A}^T(\theta) \mathbf{A}(\theta) \big) - \frac{1}{4} |\mathbf{I}(\mathbf{y}) - \mathbf{A}^T(\theta) \mathbf{A}(\theta)|^2 \big(\mathbf{I}(\mathbf{y})\big)^{-1} \Big]   \\
&\mathbf{P}_2(\mathbf{y}, \theta) :=  \frac{2c_2 L_{\Omega}^2 k(\mathbf{y},\theta)}{|\tilde{\mathbf{u}}_0|^4|\tilde{\mathbf{v}}_0|^4|\partial_1 \mathbf{y} \times \partial_2 \mathbf{y}|} \Big[ ( \partial_2 \mathbf{y} \times  \mathbf{N}(\mathbf{y})L(\theta) \mathbf{y} )    \otimes \mathbf{e}_1 - ( \partial_1 \mathbf{y} \times  \mathbf{N}(\mathbf{y}) L(\theta) \mathbf{y} )    \otimes \mathbf{e}_2\Big],
\end{aligned}
\end{equation}
after some algebraic manipulation.
The latter formula includes $k(\mathbf{y}, \theta) := \big[ \mathbf{v}(\theta) \cdot \mathbf{v}'(\theta) \big] \big[ \tilde{\mathbf{u}}_0 \cdot \mathbf{II}(\mathbf{y}) \tilde{\mathbf{u}}_0 \big]  + \big[ \mathbf{u}(\theta) \cdot \mathbf{u}'(\theta) \big] \big[ \tilde{\mathbf{v}}_0 \cdot \mathbf{II}(\mathbf{y}) \tilde{\mathbf{v}}_0 \big]$, which relaxes the left side of the kinematic constraint in (\ref{eq:secFundConstraint}). It also includes  the projection tensor $\mathbf{N}(\mathbf{y}) := \mathbf{I} - \mathbf{n}(\mathbf{y}) \otimes \mathbf{n}(\mathbf{y})$, and the linear operator $L(\theta) := [(\mathbf{v}(\theta) \cdot \mathbf{v}'(\theta)) \partial_{\mathbf{u}_0} \partial_{\mathbf{u}_0} + (\mathbf{u}(\theta) \cdot \mathbf{u}'(\theta)) \partial_{\mathbf{v}_0} \partial_{\mathbf{v}_0}]$  applied to $\mathbf{y}$. The second generalized stress $\mathcal{H}$ is a third-order tensor field on $\mathbb{R}^{3\times2 \times2}$  given by 
\begin{equation}
\begin{aligned}
\mathcal{H}(\mathbf{y}, \theta) := \mathcal{H}_2(\mathbf{y}, \theta) + 2d_1 L_{\Omega}^2 \nabla \nabla \mathbf{y} ,
\end{aligned}
\end{equation}
where $\mathcal{H}_2$  is obtained by differentiating $W_2$ with respect to the Hessian $\nabla \nabla \mathbf{y}$. This differentiation gives 
\begin{equation}
\begin{aligned}
\mathcal{H}_2(\mathbf{y}, \theta) := \frac{2 c_2 L_{\Omega}^2 k(\mathbf{y}, \theta)}{|\tilde{\mathbf{u}}_0|^4|\tilde{\mathbf{v}}_0|^4} \Big( [\mathbf{v}(\theta) \cdot \mathbf{v}'(\theta) ]\big\{ \mathbf{n}(\mathbf{y}) \otimes \tilde{\mathbf{u}}_0 \otimes \tilde{\mathbf{u}}_0 \big\} + [\mathbf{u}(\theta) \cdot \mathbf{u}'(\theta) ]\big\{ \mathbf{n}(\mathbf{y}) \otimes \tilde{\mathbf{v}}_0 \otimes \tilde{\mathbf{v}}_0 \big\} \Big).
\end{aligned}
\end{equation}
The third generalized stress $q$ is a scalar field of the form
\begin{equation}
\begin{aligned}\label{eq:q1q2}
q(\mathbf{y}, \theta) := q_1(\mathbf{y}, \theta) + q_2(\mathbf{y}, \theta) +  2 d_3 \theta
\end{aligned}
\end{equation}
with terms $q_{1,2}$ obtained as derivatives of $W_{1,2}$ with respect to $\theta$ via
\begin{equation}
\begin{aligned}
&q_1(\mathbf{y}, \theta) := - \frac{2c_1}{\sqrt{\det \mathbf{I}(\mathbf{y})}}   \Big( \mathbf{I}(\mathbf{y}) - \mathbf{A}^T(\theta) \mathbf{A}(\theta) \Big) \colon \big( \mathbf{A}^T(\theta) \mathbf{A}(\theta)\big)' ,  \\
&q_2(\mathbf{y}, \theta) := \frac{2 c_2 L_{\Omega}^2 k(\mathbf{y}, \theta)}{|\tilde{\mathbf{u}}_0|^4|\tilde{\mathbf{v}}_0|^4}\Big(  \big[ \mathbf{v}(\theta) \cdot \mathbf{v}'(\theta) \big]' \big[ \tilde{\mathbf{u}}_0 \cdot \mathbf{II}(\mathbf{y}) \tilde{\mathbf{u}}_0 \big]  + \big[ \mathbf{u}(\theta) \cdot \mathbf{u}'(\theta) \big]' \big[ \tilde{\mathbf{v}}_0 \cdot \mathbf{II}(\mathbf{y}) \tilde{\mathbf{v}}_0 \big]\Big).
\end{aligned}
\end{equation}
The last generalized stress $\mathbf{j}(\theta)$ in (\ref{eq:firstVariation}) is
\begin{equation}
\begin{aligned}\label{eq:finalStress}
\mathbf{j}(\theta) := 2 d_2 L_{\Omega}^2 \nabla \theta.
\end{aligned}
\end{equation}
This completes our description of the weak formulation of the  equilibrium equations. In summary, if $(\mathbf{y}, \theta)  \in V_{\Gamma}$ is a local minimizer to the energy $E$, then   (\ref{eq:firstVariation}) holds for all $(\mathbf{v}, \eta) \in V_0$ for the generalized stresses $\mathbf{P}, \mathcal{H}, q, \mathbf{j}$ defined by (\ref{eq:PiolaStress}-\ref{eq:finalStress}).

We now derive the strong form of the equilibrium equations and the corresponding natural boundary conditions under the assumption that the minimizers $(\mathbf{y}, \theta)$ are sufficiently smooth.  A key step in this result is an identity (\ref{eq:importForStrongForm}) on the first integral in (\ref{eq:firstVariation}) derived in Appendix \ref{sec:DeriveStrongForm}. We manipulate the weak form  using this identity to read
\begin{equation}
\begin{aligned}\label{eq:manipWeakForm}
0 &= \int_{\Omega}   \big(\Div \big[ \Div \mathcal{H}(\mathbf{y}, \theta)  - \mathbf{P} (\mathbf{y}, \theta)   \big] - \overline{\mathbf{b}} \big)\cdot \mathbf{w}  \dif x  + \int_{\Gamma_{\text{t}}}  \big( \big[\mathbf{P} (\mathbf{y}, \theta)  - \nabla  \mathcal{H}(\mathbf{y}, \theta) \colon (\mathbf{I} + \mathbf{n}^{\perp} \otimes \mathbf{n}^{\perp}) \big]  \mathbf{n} - \overline{\mathbf{t}} \big) \cdot \mathbf{w} \dif s \\
&  \qquad  + \int_{\Gamma_{\text{m}}} \big[\mathcal{H}(\mathbf{y}, \theta) \colon (\mathbf{n} \otimes  \mathbf{n})   - \overline{\mathbf{m}} \big] \cdot (\nabla \mathbf{w}) \mathbf{n}    \dif s   + \int_{\Omega} \big(  q(\mathbf{y}, \theta)  - \nabla \cdot \mathbf{j}(\theta) \big) \eta  \dif x +  \int_{\partial \Omega} \big( \mathbf{j} (\theta)\cdot \mathbf{n} \big)  \eta \dif s
\end{aligned}
\end{equation}
for all $(\mathbf{w}, \eta) \in V_0$,  after employing that $\mathbf{w} = \mathbf{0}$ a.e.\;on $\Gamma_\text{d}$ and $(\nabla \mathbf{w})\mathbf{n} = \mathbf{0}$  a.e.\;on $\Gamma_{\text{n}}$ to simplify the boundary terms.  For clarity on this formula,  $\mathbf{n}$ denotes the outward normal to $\partial \Omega$, with $\mathbf{n}^{\perp} = \mathbf{R}(\pi/2) \mathbf{n}$ the corresponding unit tangent vector.  In addition, $[\Div \mathcal{H}]_{i\alpha} = [\mathcal{H}]_{i\alpha \beta,\beta}$, $[\mathcal{H}  \colon  ( \mathbf{n} \otimes \mathbf{n})]_i = [\mathcal{H}]_{i \alpha \beta} n_{\alpha} n_{\beta}$,  and $[\nabla \mathcal{H} \colon (\mathbf{I} - \mathbf{n}^{\perp} \otimes \mathbf{n}^{\perp})]_{i \alpha} = [\mathcal{H}]_{i\alpha \beta,\gamma} ( \delta_{\beta \gamma} - n^{\perp}_{\beta} n^{\perp}_{\gamma})$.  To complete the derivation, we test (\ref{eq:manipWeakForm}) against  various subsets of $(\mathbf{w}, \eta) \in V_0$ to isolate different terms in the integral and apply standard localization arguments  case-by-case to deduce the full set of  equilibrium equations  and natural boundary conditions. Skipping the details,  the governing equations in our  constitutive model for parallelogram origami  are
\begin{equation}
\begin{aligned}
\begin{cases}\label{eq:governingEquations}
\Div \big[ \Div \mathcal{H} (\mathbf{y}, \theta) - \mathbf{P} (\mathbf{y}, \theta)  \big]= \overline{\mathbf{b}}  & \text{ in $\Omega$} \\
\nabla \cdot \mathbf{j}(\theta) - q(\mathbf{y}, \theta) = 0  & \text{ in } \Omega \\ 
\mathbf{y}  = \overline{\mathbf{y}} & \text{ on } \Gamma_{\text{d}}  \\
 \big[\mathbf{P} (\mathbf{y}, \theta) - \nabla \mathcal{H}(\mathbf{y}, \theta) \colon ( \mathbf{I} + \mathbf{n}^{\perp} \otimes \mathbf{n}^{\perp})  \big]  \mathbf{n}    = \overline{\mathbf{t}}  &  \text{ on } \Gamma_{\text{t}} \\ 
(\nabla \mathbf{y}) \mathbf{n} = \overline{\mathbf{s}} & \text{ on } \Gamma_{\text{n}} \\
\mathcal{H}(\mathbf{y}, \theta) \colon( \mathbf{n} \otimes \mathbf{n} ) = \overline{\mathbf{m}} & \text{ on } \Gamma_{\text{m}} \\ 
\nabla \theta \cdot \mathbf{n} = 0 & \text{ on } \partial \Omega.
\end{cases} 
\end{aligned}
\end{equation}

The complexity of these equations underlies the fact that a continuum description of metamaterials  requires not only a characterization of the cell-averaged effective deformation $\mathbf{y}$ but also auxiliary field(s) (in this case an angle field $\theta$) that track the mechanical behavior of the cell-by-cell microstructure. Such generalized elastic continuum are known to be quite rich; see for instance Eringen's work on microcontinuum theories \cite{eringen2012microcontinuum}. The first equilibrium equation is the standard one in non-linear elasticity where $\mathbf{P} (\mathbf{y}, \theta) - \Div \mathcal{H} (\mathbf{y}, \theta)$ plays the role of a generalized stress. Indeed, by freezing $\theta$ (so it is no longer an elastic field variable) and dropping the terms in the model depending on $\nabla \nabla \mathbf{y}$, the governing equations in (\ref{eq:governingEquations}) reduce to the familiar ones: $\Div \mathbf{P}(\mathbf{y}) + \overline{\mathbf{b}}= \mathbf{0}$  in $\Omega$ subject to mixed traction and Dirichlet boundary conditions $\mathbf{P}(\mathbf{y}) \mathbf{n} = \overline{\mathbf{t}}$ on $\Gamma_{\text{t}}$ and $\mathbf{y} = \overline{\mathbf{y}}$ on $\Gamma_{\text{d}}$. Of course, $\theta$ is  not frozen but instead an elastic field variable reflecting the actuation of the folds, which helps to reduce the stored elastic energy and stress of the origami under loads. This necessitates another equilibrium equation $\nabla \cdot \mathbf{j}(\theta) - q(\mathbf{y}, \theta) = 0$ in $\Omega$ and boundary condition $\nabla \theta \cdot \mathbf{n}= 0$ on $\partial \Omega$. The higher order gradients of the deformation in the model  account for origami bending. They introduce interesting couplings in the equilibrium equations through the terms $\mathcal{H}(\mathbf{y}, \theta)$, $\mathbf{P}(\mathbf{y}, \theta)$ and $q(\mathbf{y}, \theta)$ and lead to  somewhat peculiar looking traction and generalized moment boundary conditions.  Such boundary conditions, while not an everyday occurrence in modern continuum mechanics research,  are completely consistent with those studied by Toupin \cite{toupin1962elastic}, back when the development of generalized continuum and strain gradient models  was a  topic of the times.

\section{Finite element formulation}\label{sec:FEM}

In this section, we develop a finite element formulation of the governing equations in (\ref{eq:governingEquations}). The main technical challenge is to address  the second gradients  of  deformation, present in the stored energy of the model, in a computationally reasonable way. Traditional finite element discretizations of the deformation, involving  Lagrange polynomials, possess jumps in their derivatives  at internal edges and thus are non-conforming in $H^2(\Omega, \mathbb{R}^3)$. This  means that such discretizations do not accurately capture the $\nabla \nabla \mathbf{y}$ terms present in our model in a straightforward manner.  While $H^2$-conforming finite elements have been known for a long time \cite{argyris,bell}, ensuring the desired smoothness across the edges leads, in practice, to elements of order at least $\geq 4^{\text{th}}$, which become computationally expensive.
Writing an interpolator for these elements is also difficult, meaning that one often ends up imposing the Dirichlet boundary conditions weakly. Another approach is to use subdivision elements \cite{cirak2000subdivision}, which achieve $H^2$-conforming  deformations through meshing strategies that originate from computational geometry \cite{loop1987smooth}. This approach has proven successful at simulating the large deformation response of plates and shells, but is difficult to implement in standard finite element packages. 
We follow instead the interior penalty method in \cite{engel2002continuous}. Specifically, we  use traditional Lagrange polynomials for the discretization, allowing the normal derivatives to jump across edges, but then modify the weak formulation of the governing equations to handle these jumps.

\subsection{Preliminaries}\label{sec:NotationFEM}

Here we briefly collect  some notation  helpful for concisely describing the finite element formulation of (\ref{eq:governingEquations}) using the interior penalty method, which we develop in the next section.

We always consider a  discretization of the domain $\Omega$ via  a shape-regular conforming triangulation $\mathcal{T}^h$  with $h_T :=  \text{diam} (T)$ and $h := \max_{T \in \mathcal{T}_h} h_T$ (see \cite{ciarlet2002finite} for details). We denote by $\mathcal{E}^h$ the set of all edges of $\mathcal{T}^h$, $\mathcal{E}_{\Gamma}^h$ the set of all boundary edges $e \in \mathcal{E}^h$ such that $e \subset \partial \Omega$, and $\mathcal{E}_{\text{int}}^h := \mathcal{E}^h \setminus \mathcal{E}_{\Gamma}^h$ the set of all internal edges. We also consider the  subsets $\mathcal{E}_{\Gamma}^h$ denoted $\mathcal{E}_{\Gamma_{\text{d}}}^h, \mathcal{E}_{\Gamma_{\text{t}}}^h, \mathcal{E}_{\Gamma_{\text{n}}}^h, \mathcal{E}_{\Gamma_{\text{m}}}^h$, which collect the edges on $\Gamma_{\text{d}}$, $\Gamma_{\text{t}},$ $\Gamma_{\text{n}}$ and $\Gamma_{\text{m}}$, respectively.  
For any interior edge $e \in \mathcal{E}_{\text{int}}^h$ shared by two triangles, denoted  $T^{+}$ and $T^{-}$, we define $\mathbf{n}_e$ as the unit normal to $e$ that points from $T^{+}$ to $T^{-}$. As field variables defined on the triangulation $\mathcal{T}^h$ can jump across edges, we define the jump and average of a field variable $(\cdot)$ (any scalar, vector, or $n^{\text{th}}$-order tensor field) across the edge $e$ by  $\llbracket (\cdot) \rrbracket := (\cdot)^{+} - (\cdot)^{-}$  and $\ldblbrace (\cdot) \rdblbrace := \frac{1}{2} ( (\cdot)^{+} + (\cdot)^{-})$, respectively, where $(\cdot)^{\pm}$ denotes the limit of the restriction of the field variable $(\cdot)\vert_{T^{\pm}}$  to the edge $e$.
For an exterior edge $e \in \mathcal{E}_{\Gamma}^{h}$, we take $\mathbf{n}_e \equiv \mathbf{n}$ to be the outward unit normal of $T$ on $e$. Finally, the length of any edge $e \in \mathcal{E}^h$ is given by $|e|$. 

We define the approximation space for the deformation and angle field on our triangulation $\mathcal{T}^h$ as 
\begin{equation}
\begin{aligned}
\label{eq:fem space}
V^h :=  \big\{ (\mathbf{y}_h, \theta_h)  \in C^0(\Omega, \mathbb{R}^3) \times C^0(\Omega,\mathbb{R}) \colon   \mathbf{y}_h\vert_{T} \in P^2(T)\;\;\text{and} \;\;  \theta_h\vert_{T} \in P^1(T)\;\;  \forall \;\; T \in \mathcal{T}^h \big\} 
\end{aligned}
\end{equation} 
where $P^k(T)$ denotes the space of Lagrange polynomials of degree $\leq k$ on the triangle $T$, and $C^0$ denotes the space of continuous functions.
Note that $\mathbf{y}_h$ need not be in $H^2$, which makes the proposed method nonconforming.
We write the classical Scott--Zhang interpolator \cite{ern_guermond} for $V^h$ as $\mathcal{I}_h$ and take into account Dirichlet boundary conditions on $\mathbf{y}_h$ by defining
\begin{equation}
\begin{aligned}\label{eq:discreteSpace}
V^h_{\Gamma_{\text{d}}}:= \{ (\mathbf{y}_h, \theta_h)  \in V^h \colon   \mathbf{y}_h = \mathcal{I}_h \overline{\mathbf{y}} \text{ on } \Gamma_{\text{d}} \}.
\end{aligned}
\end{equation}
The Dirichlet boundary condition on the slope, $\nabla \mathbf{y}_h \mathbf{n} = \overline{\mathbf{s}}$ on $\Gamma_{\text{n}}$, are imposed weakly.
The discrete solution space is then $V^h_{\Gamma_{\text{d}}}$. The associated homogeneous space is defined as $V_{0}^h$, and is given by (\ref{eq:discreteSpace}) but with $\overline{\mathbf{y}} = \mathbf{0}$. A generic deformation $\mathbf{y}_h$ in $V^h$ can have a discontinuous gradient across $e \in \mathcal{E}^h_{\text{int}}$, so we cannot directly define a Hessian (though we can define a gradient in the usual sense as $\mathbf{y}_h$ is continuous across $e$).
We define instead a broken Hessian $\nabla_h \nabla \mathbf{y}_h := \sum_{T \in \mathcal{T}_h} \nabla (\nabla \mathbf{y}_h)_T \chi_T$, where $\chi_T$ is the indicator function of $T \in \mathcal{T}^h$, which consists of computing the Hessian elementwise (see \cite{di2011mathematical} for more details). Finally, for any $(\mathbf{y}_h, \theta_h) \in V^h$, we write the generalized stresses via a slight abuse of notation as $\mathbf{P}(\mathbf{y}_h, \theta_h)$, $\mathcal{H}(\mathbf{y}_h, \theta_h)$, $q(\mathbf{y}_h, \theta_h)$, where terms that depend on the Hessian of $\mathbf{y}_h$ (which is not well-defined) are replaced by the broken Hessian $\nabla_h \nabla \mathbf{y}_h$.

\subsection{Weak formulation using the interior penalty method}\label{sec:WeakFormFEM}

The problem of solving for an  equilibrium solution of the governing equations in (\ref{eq:governingEquations}) can be written succinctly  in weak form as follows: Find a $(\mathbf{y}, \theta) \in V_{\Gamma}$  such that 
\begin{equation}
\begin{aligned}\label{eq:weakFormConcise}
   \mathscr{A}_0((\mathbf{y}, \theta); (\mathbf{w}, \eta)) = \mathscr{B}_0(\mathbf{w}) \quad \text{ for all }  (\mathbf{w}, \theta) \in  V_0 
\end{aligned}
\end{equation}
where, following (\ref{eq:firstVariation}), $\mathscr{A}_0((\mathbf{y}, \theta); (\mathbf{w}, \eta)) := \int_{\Omega} \{  \mathbf{P}(\mathbf{y}, \theta) \colon  \nabla \mathbf{w}   +  \big \langle \mathcal{H}(\mathbf{y}, \theta) ,  \nabla \nabla \mathbf{w} \big \rangle  + q(\mathbf{y}, \theta)  \eta  + \mathbf{j}(\theta) \cdot \nabla \eta  \} \dif x$  and $\mathscr{B}_0(\mathbf{w}) :=     \int_{\Omega} \overline{\mathbf{b}} \cdot \mathbf{w}  \dif x + \int_{\Gamma_{\text{t}}} \overline{\mathbf{t}} \cdot \mathbf{w}  \dif s +  \int_{\Gamma_{\text{m}}} \overline{\mathbf{m}} \cdot (\nabla \mathbf{w} ) \mathbf{n}   \dif s$ denote the parts of the weak form associated to the stored elastic energy and external loads, respectively. Of course, solving for the equilibrium solution numerically necessitates some approximations. 

We now formulate conditions for an approximate equilibrium solution using $C^0$ finite elements  and the interior penalty method. In short, after discretizing the domain, deformation and angle field via the description in  the previous section, we modify the weak formulation in (\ref{eq:weakFormConcise}) to read: Find a $(\mathbf{y}_h, \theta_h) \in V^h_{\Gamma_\text{d}}$ such that 
\begin{equation}
\label{eq:discrete problem}
    \begin{aligned}
        \mathscr{A}^h((\mathbf{y}_h, \theta_h); (\mathbf{w}_h, \eta_h)) = \mathscr{B}_0(\mathbf{w}_h) \quad \text{ for all }  (\mathbf{w}_h, \theta_h) \in  V^h_0
    \end{aligned}
\end{equation}
The left-hand side of this equation is the sum of four terms  
\begin{equation}
\begin{aligned}\label{eq:A0h}
    \mathscr{A}^h := \mathscr{A}^h_0 + \mathscr{A}^h_{\text{con}} + \mathscr{A}^h_{\text{sta}} + \mathscr{A}^h_{\text{bnd}},
\end{aligned}
\end{equation}
 each evaluated on $(\mathbf{y}_h, \theta_h) \in V^h_{\Gamma_{\text{d}}}$ and $(\mathbf{w}_h, \eta_h) \in V^h_{0}$. The first term is a discrete analog of the internal energy term $\mathscr{A}_0((\mathbf{y}, \theta); (\mathbf{w}, \eta))$  in  (\ref{eq:weakFormConcise}). It  is given by
 \begin{equation}
     \begin{aligned}
\mathscr{A}^h_0 ( (\mathbf{y}_h, \theta_h ); (\mathbf{w}_h, \eta_h))  := \sum_{T \in \mathcal{T}^h} \int_{T} \Big\{  \mathbf{P}(\mathbf{y}_h, \theta_h) \colon  \nabla \mathbf{w}_h   +  \big \langle \mathcal{H}(\mathbf{y}_h, \theta_h) ,  \nabla_h \nabla \mathbf{w}_h \big \rangle  + q(\mathbf{y}_h, \theta_h)  \eta_h  + \mathbf{j}(\theta_h) \cdot \nabla \eta_h  \Big\} \dif x.
     \end{aligned}
 \end{equation}
The second term  $\mathscr{A}_{\text{con}}^h$ is introduced for consistency, i.e., to ensure that any equilibrium solution to the original problem in  (\ref{eq:governingEquations})  is also a solution to the weak form in our numerical   method. It is of the form 
 \begin{equation}
 \begin{aligned}
 \label{eq:consistency}
\mathscr{A}^h_{\text{con}}( (\mathbf{y}_h, \theta_h ); (\mathbf{w}_h, \eta_h)) := - \sum_{e \in \mathcal{E}_{\text{int}}^h} \int_{e}    \big(\ldblbrace \mathcal{H}(\mathbf{y}_h, \theta_h)\rdblbrace \cdot \mathbf{n}_e \big) \colon \llbracket \nabla \mathbf{w}_h \rrbracket  \dif s -  \int_{\Gamma_{\text{n}}}  \big( \mathcal{H}(\mathbf{y}_h, \theta_h) \cdot \mathbf{n} \big) \colon \nabla \mathbf{w}_h  \dif s,
 \end{aligned}
 \end{equation}
where the integrals over each $e \in \mathcal{E}_{\text{int}}^h$ arise due to the jumps in both $\nabla \mathbf{y}_h$ and $\nabla \mathbf{w}_h$ along edges in the mesh, and the one over $\Gamma_{\text{n}}$ results from imposing the slope boundary conditions weakly in this formulation. Appendix \ref{sec:appendix}  provides a derivation of (\ref{eq:consistency}) (see also  \cite{di2011mathematical,engel2002continuous} for further examples of such terms and their justification). 
The third term  $\mathscr{A}_{\text{sta}}^h$ seeks to stabilize/minimize the aforementioned jumps in the deformation gradient via
\begin{equation}
\begin{aligned}
     \mathscr{A}^h_{\text{sta}} ( (\mathbf{y}_h, \theta_h ); (\mathbf{w}_h, \eta_h)):=  \sum_{e \in \mathcal{E}^h_\text{int}} \frac{\alpha}{|e|} \int_{e} \llbracket \nabla \mathbf{y}_h \rrbracket : \llbracket \nabla \mathbf{w}_h \rrbracket  \dif s.
     \end{aligned}
 \end{equation}
It corresponds to a jump energy of the form $\tfrac{1}{2}\sum_{e \in \mathcal{E}^h_\text{int}} \frac{\alpha}{ |e|}  \int_{e}  \llbracket \nabla \mathbf{y}_h \rrbracket^2 \dif s$ and  provides  coercivity of the overall discrete energy for $\alpha > 0$ large enough,  ensuring the existence of a solution to (\ref{eq:discrete problem}).
The final term $\mathscr{A}_{\text{bnd}}^h$ seeks to impose the slope boundary condition weakly via 
\begin{equation}
\label{eq:bilinear bnd}
    \mathscr{A}^h_\text{bnd}( (\mathbf{y}_h, \theta_h) ; (\mathbf{w}_h, \eta_h)) :=  \sum_{e \in \mathcal{E}_{\Gamma_{\text{n}}}^h} \frac{\alpha}{|e|} \int_{e}  \big((\nabla \mathbf{y}_h)\mathbf{n} - \bar{\mathbf{s}} \big) \cdot (\nabla \mathbf{w}_h)\mathbf{n} \dif s.
\end{equation}
 It  arises from an energy  of the form $\tfrac{1}{2}  \sum_{e  \in \mathcal{E}_{\Gamma_{\text{n}}}^h } \tfrac{\alpha}{|e|}\int_{e} |(\nabla \mathbf{y}_h) \mathbf{n} - \bar{\mathbf{s}} |^2 \dif s$, and thus furnishes $(\nabla \mathbf{y}_h) \mathbf{n} \approx \overline{\mathbf{s}}$ by way of a least square fit.

 \subsection{Implementation}

The weak form in \eqref{eq:discrete problem} is a nonlinear algebraic system of equations. 
We implement it in \texttt{Firedrake} \cite{FiredrakeUserManual}, which is a finite element library with a \texttt{Python} interface that  allows one to choose a finite element space \eqref{eq:fem space} and then write the discrete weak formulation \eqref{eq:discrete problem}.
We use Firedrake's integrated Newton solver to find a solution to  \eqref{eq:discrete problem}.
As the convergence of Newton solvers are sensitive to the initial guess, we provide  a method to compute an initial guess for $(\mathbf{y}_h,\theta_h)$ in Appendix \ref{sec:InitialGuess}.  Also, we typically employ a loading protocol that monotonically increases the load in small increments until finally solving the desired large deformation boundary value problem. We use the method in Appendix \ref{sec:InitialGuess} to initialize the simulation in the first step, while  all  subsequent steps  are initialized using the converged solution of the prior increment.  Finally, it is well-known that penalty methods, like the one proposed here, are sensitive to the choice of $\alpha$. By trial and error, we have found that $\alpha = 0.1$ works well when  $c_1 \in [1,5]$ and the other moduli $c_2, d_1,\ldots, d_3$ are consistent with the scalings discussed in Section \ref{ssec:modelIntro}. We use this value for $\alpha$ in all forthcoming simulations in Section \ref{sec:Examples}.  

\section{Examples}\label{sec:Examples}

\begin{figure}[t!]
\hspace*{-1.4cm}\centering
\includegraphics[width=.75\textwidth]{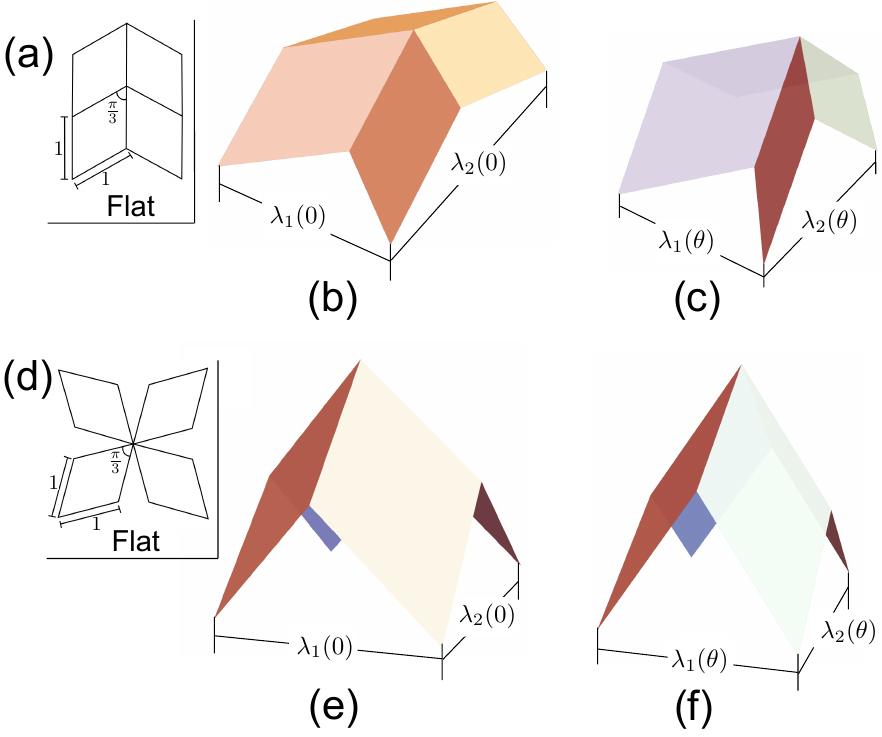} 
\caption{The Miura and Eggbox origami unit cells used in our catalog of examples.  (a,d) Flat states. Each unit cell is built using four rhombi panels with side lengths 1 and acute angle $\pi/3$. (b,e) Partially folded reference states.   $\lambda_{1}(0)$ and $\lambda_2(0)$ denote the side lengths from the formulas in (\ref{eq:lambdaMO}) and (\ref{eq:lambdaEO}), evaluated at $\theta =0$. (c,f) Actuated states. The side lengths change from $\lambda_i(0)$ to $\lambda_i(\theta)$, $i = 1,2$, under the actuation, which is  parameterized by $\theta$.}
\label{Fig:CellsRefSec5}
\end{figure}

This section showcases a variety of examples of parallelogram origami under loads, modeled using the constitutive framework in Section \ref{sec:modelForSimulations} along with the numerical method in Section \ref{sec:FEM}.   We  explore the effective mechanical behavior of the two canonical examples of parallelogram origami: Miura and Eggbox origami. Miura is an auxetic metamaterial with a negative Poisson's ratio; Eggbox is not auxetic and thus has
 a positive Poisson's ratio.  Both patterns posses unit cells built from the single parallelogram panel and have orthogonal Bravais lattice vectors $\mathbf{u}(\theta) = \lambda_1(\theta) \mathbf{e}_1$ and $\mathbf{v}(\theta) = \lambda_2(\theta) \mathbf{e}_2$ with simple trigonometric expressions, making them convenient for analytical and numerical investigation. We focus on examples where the unit cell is composed of four identical rhombi  panels of side length 1 and acute angle $\pi/3$ (see Fig.\;\ref{Fig:CellsRefSec5}(a) and (d)). We also take the reference configuration of all the examples in this section to consist of the same partially folded unit cells (Fig.\;\ref{Fig:CellsRefSec5}(b) for Miura origami, Fig.\;\ref{Fig:CellsRefSec5}(e) for Eggbox). After some tedious algebra (which we do not detail), the actuation of the Bravais lattice vectors for Miura origami in this setting is given by stretches $\lambda_1(\theta) \equiv  \lambda_1^{(\text{MO})}(\theta)$ and $\lambda_2(\theta) \equiv \lambda_2^{(\text{MO})}(\theta)$ that satisfy  
\begin{equation}
\begin{aligned}\label{eq:lambdaMO}
    \lambda_{1}^{(\text{MO})}(\theta) = \sqrt{3} \cos\left( \frac{\theta + \pi/6}{2} \right) , \quad \lambda_{2}^{(\text{MO})}(\theta) = 2 \sqrt{2} \Big(5- 3\cos(\theta+ \pi/6)   \Big)^{-1/2}, \quad \theta \in \left(-\frac{\pi}{6}, \frac{5\pi}{6} \right) .  
\end{aligned}
\end{equation}
For Eggbox origami, the analogous stretches  $\lambda_1(\theta) \equiv \lambda_1^{(\text{EO})}(\theta)$ and $\lambda_2(\theta) \equiv \lambda_2^{\text{(EO)}}(\theta)$  are 
\begin{equation}\label{eq:lambdaEO}
    \begin{aligned}
      \lambda_{1}^{(\text{EO})}(\theta) = 2 \sin\Big[ \frac{1}{2} \Big(  \arccos\big(1 - \cos \theta\big) - \theta   \Big) \Big], \quad \lambda_{2}^{(\text{EO})}(\theta) =\lambda_1^{(\text{EO})}(-\theta) , \quad \theta \in \big(-\frac{\pi}{3} , \frac{\pi}{3} \big) . 
    \end{aligned}
\end{equation}
Fig.\;\ref{Fig:CellsRefSec5}(c) and (f) illustrate the actuation. In all examples, whether Miura or Eggbox, we take the effective reference domain $\Omega$ to be 
\begin{equation}
\begin{aligned}\label{eq:DomainForExamples}
\Omega = (0, \lambda_1(0)) \times (0, \lambda_2(0)) \approx  \begin{cases}
(0, 1.67) \times ( 0, 1.82)  & \text{ for Miura} \\
(0, 1.41) \times (0, 1.41) & \text{ for Eggbox}.
\end{cases}
\end{aligned}
\end{equation}
This choice, while seemingly more cumbersome than fixing a unit square domain, ensures that the aspect ratio of $\Omega$ is the same as that of a unit cell of the origami. As such, we can compare analytical and simulation results of the effective model directly to origami deformations with a \textit{discrete} $M\times M$ number of cells. We do such a comparison for all examples in sequel. The origami deformations are constructed to match that of the effective model by  the method outlined in Appendix \ref{sec:oriDefs}.

Note  that the design space of all parallelogram origami is large (in fact, 7 dimensional\footnote{A parallelogram origami unit cell has 4 lengths  and 4 sector angles corresponding to the sides of the parallelograms and  their shape.  However, one of the lengths can be fixed (set to 1, for instance) without loss of generality since the elastic model  does not change under a dilation of the pattern.}), yet we are focusing on just two designs here. This is purposeful. Our primary interest in this work is to introduce and validate the model  and numerical method. Due to their familiarity, the Miura and Eggbox patterns make excellent templates for this validation. Also, as we will show, these patterns turn out to be quite rich in their own right, with many distinguishing features worthy of a thorough investigation and comparison. 
A more detailed examination of the behavior of broad classes  of parallelogram origami under loads  will be the topic of future work.

\subsection{Mechanism deformations}\label{ssec:PureMech}

We first study the effective mechanical response of parallelogram origami patterns under a pure mechanism actuation. This setting corresponds to Dirichlet boundary conditions of the form $\overline{\mathbf{y}}(\mathbf{x}) = \mathbf{A}(\overline{\theta}) \mathbf{x}$ on $\partial \Omega$, where $\mathbf{A}(\theta)$ is the shape tensor in (\ref{eq:shapeTensor}) and $\overline{\theta} \in (\theta^{-}, \theta^+)$ sets the level of actuation. We also assume that $\overline{\mathbf{b}} =  \overline{\mathbf{m}} = \mathbf{0}$, so there are no applied forces. Thus,  in the absence of a mechanical resistance to folding ($d_3 = 0$), the energy minimizing  (and zero energy) deformation and angle field in the model are that of the purely mechanistic response
\begin{equation}
\begin{aligned}
\mathbf{y}(\mathbf{x}) = \mathbf{A}(\overline{\theta}) \mathbf{x}, \quad \theta(\mathbf{x}) = \overline{\theta} \quad \text{ in } \Omega.
\end{aligned}
\end{equation}
The situation is more complicated when $d_3 >0$. The folds prefer not to be actuated $(\theta  = 0)$,  leading to stresses that depend non-linearly on $\overline{\theta}$.

\begin{figure}
\centering
\hspace*{-.6cm}
\includegraphics[width=0.80\textwidth]{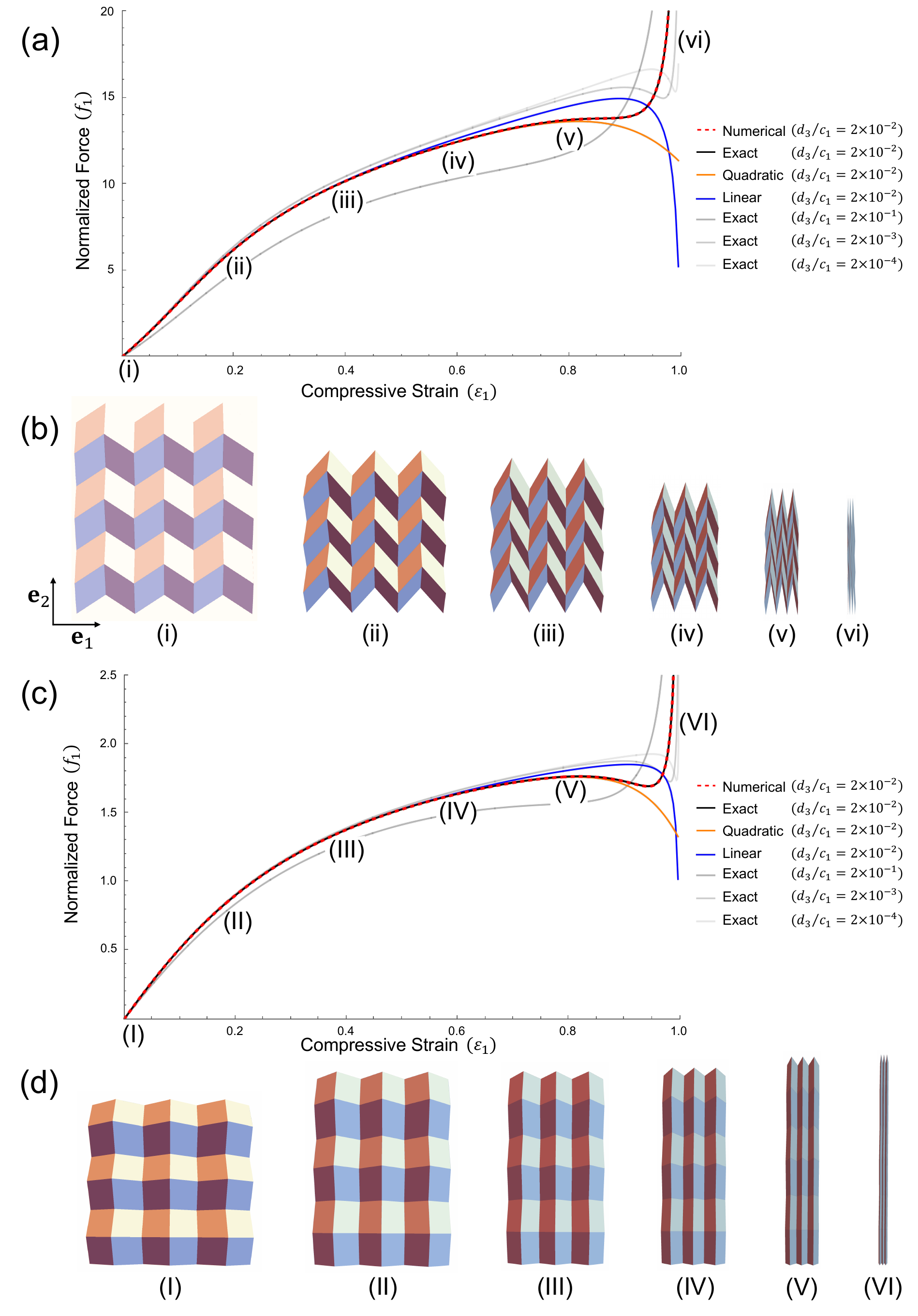} 
\caption{Normalized force versus compressive strain for the pure mechanism in Miura and Eggbox origami. (a) The force plot for Miura origami, and (b) its representative actuation for six points (i)-(vi) along the curve. (c) The force plot for Eggbox origami, and (d) its representative actuation for six points (I)-(VI) along the curve. The normalized force and compressive strain are defined in (\ref{eq:normForce}).}
\label{Fig:ForcePlotVert}
\end{figure}

We obtain an analytical expression for the stresses as follows. First, observe that the ansatz $(\mathbf{y}(\mathbf{x}), \theta(\mathbf{x})) = (\mathbf{A}(\overline{\theta}) \mathbf{x}, \overline{\theta} + \delta \theta)$  resolves the boundary conditions and the first equilibrium equation in (\ref{eq:governingEquations}) automatically. The second equilibrium equation $\nabla \cdot \mathbf{j}(\theta) - q(\mathbf{y}, \theta) = 0$ reduces to a purely algebraic equation of the form  
\begin{equation}
    \begin{aligned}\label{eq:algebraic}
      q_1( \mathbf{A}(\overline{\theta}) \mathbf{x}, \overline{\theta} + \delta \theta) + 2 d_3( \overline{\theta} + \delta \theta) = 0 
       \end{aligned}
\end{equation}
for $q_1(\mathbf{y}, \theta)$ in (\ref{eq:q1q2}). Since $q_1(\mathbf{A}(\overline{\theta}) \mathbf{x}, \bar{\theta}) = 0$ and $\partial_{\theta}q_{1} (\mathbf{A}(\overline{\theta}) \mathbf{x}, \bar{\theta})  \neq 0$, this equation has a solution $\delta \theta$ provided $d_3/c_1> 0$ is sufficiently small, which is the physically relevant regime. To approximate the solution, note that (\ref{eq:algebraic}) satisfies 
\begin{equation}
    \begin{aligned}
       \frac{3}{2} g'(\overline{\theta}) \delta \theta^2 + 2  \big( g(\overline{\theta})   + 
 j(\overline{\theta}) \tfrac{d_3}{c_1} \big) \delta \theta + 2 \frac{d_3}{c_1} j(\overline{\theta}) \overline{\theta} =  O( \delta \theta^3) 
    \end{aligned}
\end{equation}
for $g(\overline{\theta}) := |\big(\mathbf{A}^T(\overline{\theta}) \mathbf{A}(\overline{\theta})\big)'|^2$ and $j(\overline{\theta}) := \sqrt{\det \big(\mathbf{A}^T(\overline{\theta}) \mathbf{A}(\overline{\theta})\big)} $. Hence, neglecting the terms of order $O(\delta \theta^3)$  leads to
\begin{equation}
\begin{aligned}\label{eq:approxMech1}
    \delta \theta \approx - 2\frac{\tfrac{d_3}{c_1} j(\overline{\theta}) \overline{\theta}}{\big( g(\overline{\theta}) + j(\overline{\theta}) \tfrac{d_3}{c_1}\big)}   \left(1 + \sqrt{1- 3\frac{\tfrac{d_3}{c_1}g'(\overline{\theta})j(\overline{\theta}) \overline{\theta}}{\big( g(\overline{\theta}) + j(\overline{\theta}) \tfrac{d_3}{c_1}\big)^2}  }\right)^{-1}
\end{aligned}
\end{equation}
using the quadratic formula. The Piola-Kirchhoff stress at quadratic order is then   
\begin{equation}\label{eq:approxMech2}
    \begin{aligned}
        \mathbf{P}\big( \mathbf{A}(\overline{\theta})\mathbf{x}, \overline{\theta} + \delta \theta\big) \approx   -4 \frac{c_1}{j(\overline{\theta})} \mathbf{A}(\overline{\theta})\Big[ \delta \theta  \big(\mathbf{A}^T(\overline{\theta}) \mathbf{A}(\overline{\theta})\big)' + \tfrac{\delta \theta^2}{2} \big(\mathbf{A}^T(\overline{\theta}) \mathbf{A}(\overline{\theta})\big)'' + \tfrac{\delta \theta^2}{4} g(\overline{\theta}) \big( \mathbf{A}^T(\overline{\theta}) \mathbf{A}(\overline{\theta})\big)^{-1}\Big].
    \end{aligned}
\end{equation}

As these formulas highlight, even the simplest example of a pure mechanism has a rich nonlinear relationship between mechanical deformation and load.   Fig\;\ref{Fig:ForcePlotVert}(a-b) illustrate the mechanical response concretely in the context of Miura origami; Fig\;\ref{Fig:ForcePlotVert}(c-d) do likewise for Eggbox origami.   Since  these examples possess orthogonal Bravais lattices vectors (in (\ref{eq:lambdaMO}) and (\ref{eq:lambdaEO})),  we plot the normalized compressive force in the $\mathbf{e}_1$ direction versus the  compressive strain in that direction using the formulas 
\begin{equation}
\begin{aligned}\label{eq:normForce}
&f_1(\overline{\theta})  := - \lambda_2(0) \frac{[\mathbf{P}\big( \mathbf{A}(\overline{\theta})\mathbf{x}, \overline{\theta} + \delta \theta\big)]_{11}}{d_3} \quad \text{ and } \quad  
&\varepsilon_1(\overline{\theta}) := - \frac{\lambda_1(\overline{\theta}) - \lambda_1(0)}{\lambda_1(0)}. 
\end{aligned}
\end{equation}
Notice from the expressions in (\ref{eq:approxMech1}) and  (\ref{eq:approxMech2}) that the folding modulus $d_3$, not the bulk modulus $c_1$, sets the scale for the force. Furthermore, the expression for the normalized force depends on the ratio $d_3/c_1$. We therefore consider four decades of $d_3/c_1$ in the plots in Fig.\;\ref{Fig:ForcePlotVert}. We also compare the exact solution for the value $d_3/c_1 = 2 \times 10^{-2}$ to the one obtained using our numerical method, and to linear and quadratic approximations of the curves.\footnote{The linear curve is obtained by replacing the exact $\delta \theta$ solving (\ref{eq:algebraic}) with $\delta \theta_{\text{lin}} = - 2 \tfrac{d_3}{c_1} j(\overline{\theta}) \overline{\theta}/\big(g(\overline{\theta}) + j(\overline{\theta}) \tfrac{d_3}{c_1}\big)$ and $\mathbf{P}(\mathbf{A}(\overline{\theta}) \mathbf{x}, \overline{\theta} + \delta \theta)$ in the formula for $f_1(\overline{\theta})$ with $-4 \big(c_1 \delta \theta_{\text{lin}}/ j(\overline{\theta}) \big) \mathbf{A}(\overline{\theta}) \big(\mathbf{A}^T(\overline{\theta}) \mathbf{A}(\overline{\theta})\big)'$. The quadratic curve is obtained by using the approximations in (\ref{eq:approxMech1}) and (\ref{eq:approxMech2}).} The agreement between the exact and numerical solutions are essentially perfect in both examples throughout the entire actuation process. However, the two approximating curves deviate markedly at a large values of $\overline{\theta}$, presumably  because   both   $g(\overline{\theta})$ and $j(\overline{\theta})$ approach  $0$ as $\overline{\theta}$ goes to the fully folded states in Miura and Eggbox origami, making the Taylor expansions unreliable at large actuation.

 Perhaps the most interesting aspect of these plots is the stress-softening behavior. This type of behavior was first theoretically reported in \cite{wei2013geometric} and subsequently reproduced in a bar and hinge model in \cite{liu2017nonlinear}.  We see from our curves that this loss of monotonicity occurs when $d_3/c_1$ is sufficiently small. The curves also quickly spring back up, a behavior that arises because  the bulk energy $W_{1}( \theta, \mathbf{I}(\mathbf{y}))$ in (\ref{eq:bulkTerm}-\ref{eq:W12Def}) is scaled by $1/\sqrt{\det \mathbf{I}(\mathbf{y})}$, which elastically penalizes deformations where the panels are close to overlapping one another. In particular, when the actuation approaches the  fully folded states (see Fig.\;\ref{Fig:ForcePlotVert}(vi) and (VI)), the second term in the stress $\mathbf{P}_1(\mathbf{y}, \theta)$ in (\ref{eq:P1P2}) dominates the first, leading the force curves to veer upwards.  

It is also interesting to compare and contrast the curves for Miura and Eggbox. We see that both sets of curves are qualitatively  similar --- the normalized force increases sub-linearly, before dipping a bit (depending on the ratio $d_3/c_1)$ and then springing back up. However, most notably,  the Miura pattern experiences significantly  larger normalized force compared to the Eggbox one, even though the comparison is for the same underlying moduli parameters $c_1$ and $d_3$. This feature highlights the dominant role of geometry in this calculation. To get more granular, observe that the normalized stiffness of each pattern in its undeformed configuration is given by 
\begin{equation}
\begin{aligned}
\frac{f_1'(0)}{\varepsilon_1'(0)} = \frac{8}{g(0) + \tfrac{d_3}{c_1}}, 
\end{aligned}
\end{equation}
where $g(0)$ depends only on the initial geometry of the pattern. 
It happens that   $g^{(\text{MO})}(0) \approx 0.461$  for Miura and $g^{(\text{EO})}(0) = 2$ for Eggbox. Thus,  the initial stiffness for the Miura is more than four times larger than that of the Eggbox pattern solely due to geometry.  This feature more-or-less persists throughout the entirety of actuation.

\subsection{Pure bending and twisting deformations}

Having examined the pure mechanism deformations both analytically and numerically, we now explore pure bending and twisting modes. These deformation modes arise as  simple but inhomogeneous solutions to the purely geometric constraints on the first and second fundamental forms in (\ref{eq:firstFundConstraint}) and (\ref{eq:secFundConstraint}), i.e., they come from pairs  $(\mathbf{y}, \theta) \colon \Omega \rightarrow \mathbb{R}^3 \times (\theta^{-}, \theta^+)$ that satisfy  
\begin{equation}
    \begin{aligned}\label{eq:purelyGeom}
    \mathbf{I}(\mathbf{y})  = \mathbf{A}^T(\theta) \mathbf{A}(\theta), \quad \big[ \mathbf{v}(\theta) \cdot \mathbf{v}'(\theta) \big] \big[ \tilde{\mathbf{u}}_0 \cdot \mathbf{II}(\mathbf{y}) \tilde{\mathbf{u}}_0 \big]  + \big[ \mathbf{u}(\theta) \cdot \mathbf{u}'(\theta) \big] \big[ \tilde{\mathbf{v}}_0 \cdot \mathbf{II}(\mathbf{y}) \tilde{\mathbf{v}}_0 \big]  = 0.
    \end{aligned}
\end{equation}
  In our prior work \cite{xu2024derivation}, we developed  a large class of such solutions. We briefly reproduce some of that analysis to help  explain the origins and nature  of these deformation modes (see  Section 5.1 and Appendix C in \cite{xu2024derivation} for more details). 
  
  \begin{figure}[t!]
\centering
\includegraphics[width=0.87\textwidth]{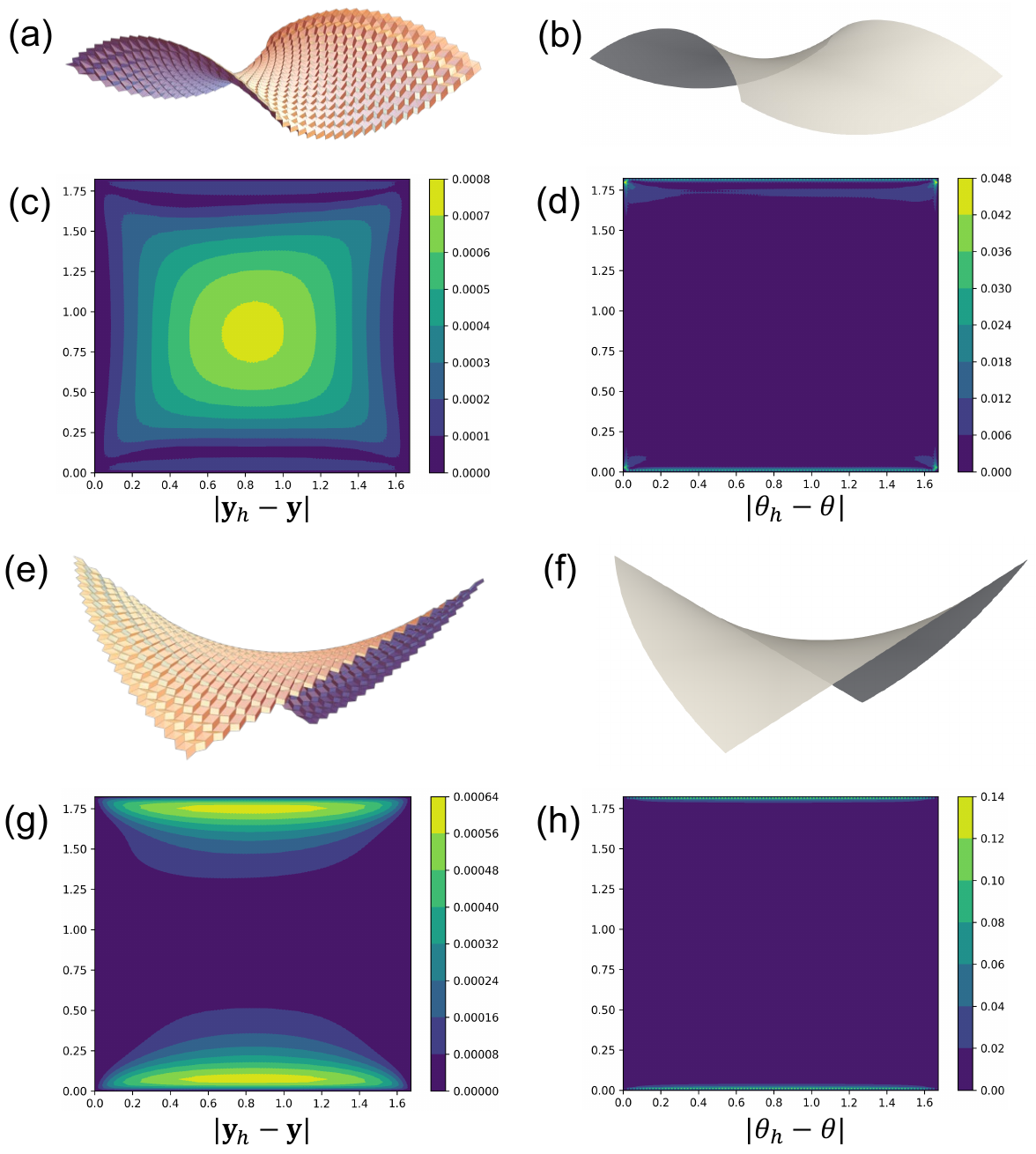} 
\caption{Pure bending and twisting modes of Miura origami. (a) Origami deformation for an analytical solution $(\mathbf{y}, \theta)$ of a pure bending mode. (b) Corresponding finite element simulation $(\mathbf{y}_h, \theta_h)$ under Dirichlet boundary conditions specified by the analytical solution. The error in (c) the effective displacement and (d) the angle field  between the analytical and simulated solutions, plotted on the reference domain. For reference, the characteristic out-of-plane displacement in bending is $\sim 0.2$ and $\max |\theta_h| =0.6254$.  (e-h) Analogous plots for pure twisting.  The characteristic out-of-plane displacement is $\sim 0.4$ and $\max |\theta_h| =0.4964$.}
\label{Fig:MiuraPureBendTwist}
\end{figure}
  
  \begin{figure}[t!]
\centering
\includegraphics[width=0.87\textwidth]{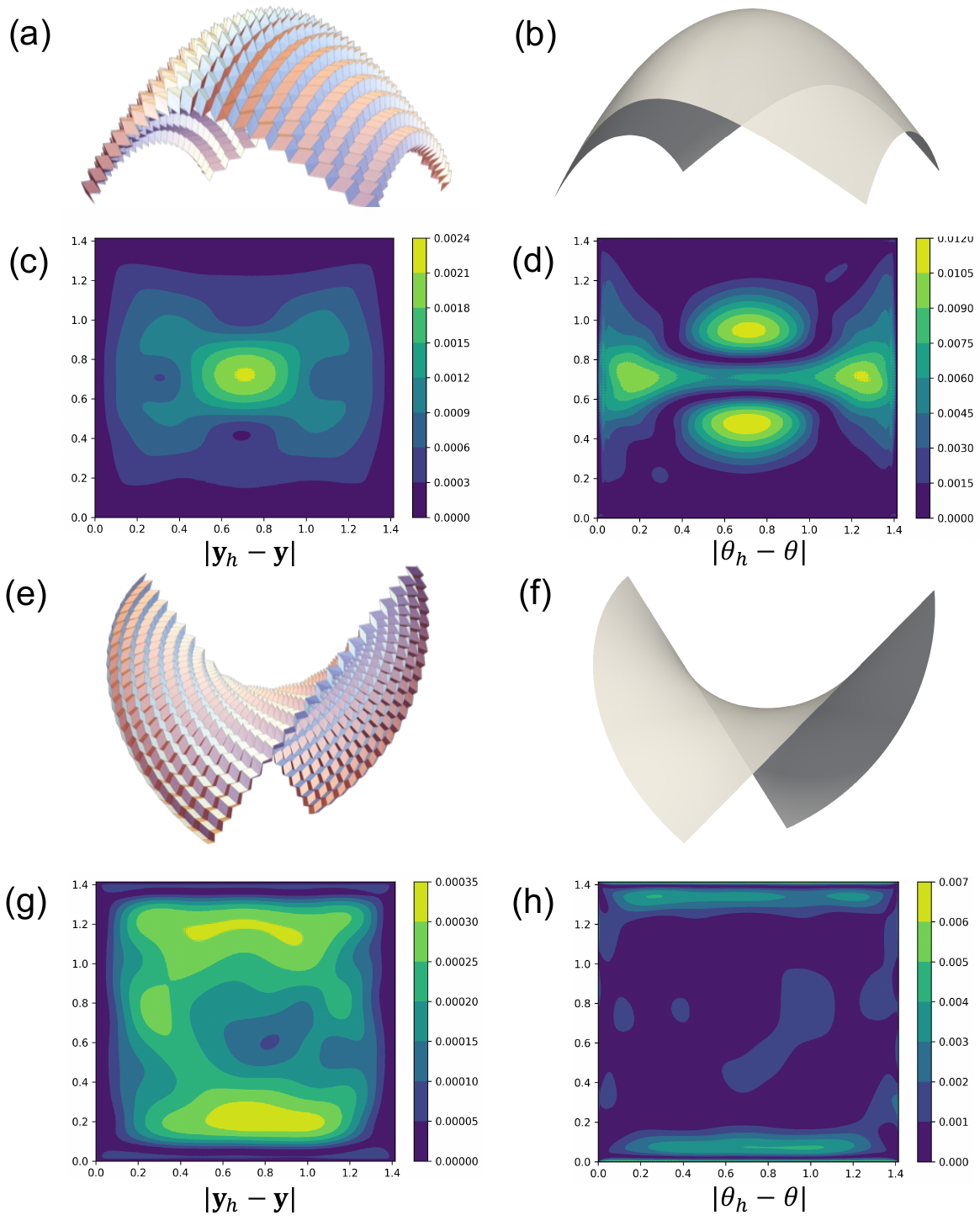} 
\caption{Pure bending and twisting modes of Eggbox origami. (a) Origami deformation for an analytical solution $(\mathbf{y}, \theta)$ of a pure bending mode. (b) Corresponding finite element simulation $(\mathbf{y}_h, \theta_h)$ under Dirichlet boundary conditions specified by the analytical solution. The error in (c) the effective displacement and (d) the angle field  between the analytical and simulated solutions, plotted on the reference domain. For reference, the characteristic out-of-plane displacement in bending is $\sim 0.3$ and $\max |\theta_h| = 0.4217$.  (e-h) Analogous plots for pure twist.  The characteristic out-of-plane displacement is $\sim 0.5$ and $\max |\theta_h|=0.5035$.}
\label{Fig:EggboxPureBendTwist}
\end{figure}

Let's focus on the setting of orthogonal Bravais lattice vectors $\mathbf{u}(\theta) = \lambda_1(\theta) \mathbf{e}_1$ and $\mathbf{v}(\theta) = \lambda_2(\theta) \mathbf{e}_2$, consistent with the Miura and Eggbox examples. In such cases, the constraints in  (\ref{eq:purelyGeom}) are equivalent to parameterizing the first and second fundamental forms as 
\begin{equation}
\begin{aligned}\label{eq:getParam}
\mathbf{I}(\mathbf{y}) = \begin{pmatrix}\tfrac{\lambda^2_1(\theta)}{\lambda_1^2(0)} & 0 \\ 0 & \tfrac{\lambda^2_2(\theta)}{\lambda_2^2(0)} \end{pmatrix}, \quad \mathbf{II}(\mathbf{y}) =  \begin{pmatrix} -\tfrac{\lambda_1'(\theta)}{\lambda_1(\theta)}  \kappa  & \tau \\ \tau & \tfrac{\lambda_2'(\theta)}{\lambda_2(\theta)}  \kappa \end{pmatrix}
\end{aligned}
\end{equation}
for compatible curvature fields $\kappa, \tau \colon \Omega \rightarrow \mathbb{R}$. Here $\kappa$ describes a bending curvature, while  $\tau$ represents  twist. As always, $\theta \colon \Omega \rightarrow (\theta^{-}, \theta^+)$ is the cell-by-cell actuation.  Of course, these fields  cannot be arbitrary. Much like the statement that a deformation gradient must be curl-free to admit a deformation, the metric and curvature tensors in (\ref{eq:getParam}) must solve the Gauss and Codazzi-Mainardi equations from differential geometry \cite{do2016differential}.   This produces a system of three nonlinear PDEs in the three unknowns $\kappa, \tau$, and $\theta$, fully characterizing the existence of a deformation $\mathbf{y}$ solving (\ref{eq:getParam}).

We do not write out the full PDE system here because it is cumbersome and we do not know how to solve it in general. Instead we produce some illustrative examples based on a 1D ansatz. By setting 
\begin{equation}
\begin{aligned}\label{eq:1DAnsatz}
    \theta(\mathbf{x}) = \theta_{\text{1D}}\big( \mathbf{x} \cdot \lambda_2^{-1}(0) \tilde{\mathbf{e}}_2  \big) , \quad \kappa(\mathbf{x}) = \kappa_{\text{1D}} \big( \mathbf{x} \cdot \lambda_2^{-1}(0) \tilde{\mathbf{e}}_2\big), \quad \tau(\mathbf{x}) = \tau_{\text{1D}} \big( \mathbf{x} \cdot \lambda_2^{-1}(0) \tilde{\mathbf{e}}_2\big),
\end{aligned}
\end{equation}
the PDEs becomes a  coupled set of ODEs in the scalar variable $s := \mathbf{x} \cdot \lambda_2^{-1}(0) \tilde{\mathbf{e}}_2$ for the fields $(\theta_{\text{1D}}(s), \kappa_{\text{1D}}(s), \tau_{\text{1D}}(s))$, which we solved in Appendix C \cite{xu2024derivation}. The general solutions are
\begin{equation}
\begin{aligned}\label{eq:kappaTauParam}
\kappa_{\text{1D}} = \frac{c_{\kappa}}{\lambda_1'(\theta_{\text{1D}} ) \lambda_2(\theta_{\text{1D}}) }, \quad \tau_{\text{1D}} = \frac{c_\tau}{\lambda^2_1(\theta_{\text{1D}})}
\end{aligned}
\end{equation}
for arbitrary constants $c_{\kappa}$ and $c_{\tau}$, and for $\theta_{\text{1D}}$ solving the nonlinear ODE 
\begin{equation}
\begin{aligned}\label{eq:nLinearODE}
    \frac{d}{d s}\Big[\frac{\lambda_1'(\theta_{\text{1D}})}{\lambda_2(\theta_{\text{1D}})} \frac{d}{ds} \theta_{\text{1D}} \Big] = c_{\tau}^2 \frac{\lambda_2(\theta_{\text{1D}})}{\lambda_1^3(\theta_{\text{1D}})} + c_{\kappa}^2 \frac{\lambda_2'(\theta_{\text{1D}})}{\lambda_2^2(\theta_{\text{1D}}) \lambda'_1(\theta_{\text{1D}})} 
\end{aligned}
\end{equation}
subject to the initial conditions $\theta_{\text{1D}}(0) = \overline{\theta}$ and $\theta_{\text{1D}}'(0) =  \overline{\zeta}$. Note that pure bending modes in the parameterization are obtained by solving (\ref{eq:kappaTauParam}) and (\ref{eq:nLinearODE}) with $c_\tau = 0$; pure twisting modes are obtained likewise by solving these equations with $c_{\kappa} = 0$.

 Fig.\;\ref{Fig:MiuraPureBendTwist} and \ref{Fig:EggboxPureBendTwist}   compare analytical solutions  of pure bending and twisting modes for both Miura and Eggbox to analogous  simulations using our model and numerical method. In all cases,  the $\theta$-dependence is  specified by  (\ref{eq:lambdaMO}) for Miura and (\ref{eq:lambdaEO}) for Eggbox and $\Omega$ is as in (\ref{eq:DomainForExamples}).  Concerning the analytics on Miura, the inputs to the ODE in (\ref{eq:nLinearODE}) are $(c_{\kappa}, c_{\tau}, \theta_{\text{1D}}(0), \theta_{\text{1D}}'(0)) = (1.518,0,-\pi/6 + 0.2, 16)$ for pure bending and $=(0, 1.5205,-\pi/6 + 0.2, 11)$ for pure twisting; the analogous inputs for Eggbox are $= (1.43, 0, 0.4, -1.5)$ and $= (0, 1.9, -0.5, 4)$ for bend and twist, respectively.  These inputs are chosen so that the solution $\theta_{\text{1D}}(s)$ is approximately symmetric about $s = 1/2$, corresponding to the midline of the domain $\Omega$ in the $\mathbf{v}_0$ direction. After obtaining  the solutions, we substitute $(\theta_{\text{1D}}, \kappa_{\text{1D}}, \tau_{\text{1D}}) \equiv (\theta, \kappa, \tau)$ into (\ref{eq:getParam}). The deformation $\mathbf{y}$ is then constructed from these compatible fields using a procedure outlined at the end of  Appendix C of \cite{xu2024derivation}. Finally, Figs.\;\ref{Fig:MiuraPureBendTwist}, \ref{Fig:EggboxPureBendTwist}(a) and (e)     show plots of the corresponding origami deformations built from these solutions using the ansatz  in Appendix \ref{sec:oriDefs}. Notice that the Miura takes the shape of a saddle for pure bending $(c_{\tau} = 0)$, while Eggbox takes that of the opposite bending profile in  a cap. In contrast, the pure twist deformations of these pattern $(c_{\kappa}=0)$ are qualitatively similar.
 
 Our goal with the numerics is to approximate the above analytical solutions $(\mathbf{y}, \theta)$ via  corresponding finite element simulations $(\mathbf{y}_h, \theta_h)$ as a means to validate the nonlinear parts of our model and numerical implementation. To do this, we impose full Dirichlet  boundary conditions $\mathbf{y}_h = \mathbf{y}$ and $(\nabla \mathbf{y}_h) \mathbf{n} = (\nabla \mathbf{y}) \mathbf{n}$ on $\partial \Omega$. We also set the moduli parameters of the model for all numerical simulations as $c_1 = 1, c_2L_{\Omega}^2 = 0.5, d_1L_{\Omega}^2 = 10^{-3}, d_2 = d_3 = 0$. By these choices, the stored elastic energy for $(\mathbf{y}_h, \theta_h)$ is dominated mostly by deviations from the purely geometric constraints in (\ref{eq:purelyGeom}).  In particular,   the regularizing terms penalizing  the actuation field and its gradient are dropped, while only a small penalty on the second gradient of the deformation is added to stabilize the simulations. Thus, all of the numerical simulations should be driven towards the analytical solutions.  That is exactly what we observe.

 Figs.\;\ref{Fig:MiuraPureBendTwist}, \ref{Fig:EggboxPureBendTwist}(b) and (f) plot the effective deformations obtained from the simulations. Notice that the responses  are  genuinely large deformation in each case, with the  characteristic out-of-plane displacements on the order of $\sim 0.2$-$0.5$ for domains of characteristic length $\sim 1.4$-$1.8$. Yet the errors in the deformation are negligible by comparison. Figs.\;\ref{Fig:MiuraPureBendTwist}, \ref{Fig:EggboxPureBendTwist}(c) and (g) plot normed differences between the analytical and simulated effective deformations. The errors  are on the order of $0.1$-$1\%$, depending on the simulation. The twist case performs especially well in this comparison. Figs.\;\ref{Fig:MiuraPureBendTwist}, \ref{Fig:EggboxPureBendTwist}(d) and (h) also compare the   angle fields. For reference, this field  varies  in magnitude from $0$ to $0.5$ in the typical simulation.  From the figures, we see that its  errors are on the order of $5$-$20$\% in places, especially towards the boundaries of the pattern in the Miura cases. However, in most of the domain, the errors are far below $1\%$. The deformations are perhaps better approximated  because of our choice of boundary conditions for the simulations. In any case,  the simulations and analytical solutions are  mostly in excellent agreement.

\subsection{Pinching deformations}\label{ssec:PinchingSimulations}
Having validated the numerical method in the previous two sections, we now highlight its versatility by exploring inhomogeneous boundary conditions for which analytical solutions are difficult to come by.  

We begin with 2D pinching simulations. Fig.\;\ref{Fig:Miura_Bowtie} shows the case of Miura, while Fig.\;\ref{Fig:Eggbox_Bowtie} shows that of Eggbox origami. In the simulations, we subject the effective origami  domains $\Omega$ in (\ref{eq:DomainForExamples}) to Dirichlet boundary conditions $\mathbf{y}_h = \mathbf{A}(\overline{\theta}) \mathbf{x}$ on $\Gamma_{L}$ and $\Gamma_R$, where $\Gamma_{L,R}$ are as illustrated in Fig.\;\ref{Fig:Miura_Bowtie} and \ref{Fig:Eggbox_Bowtie}(c); they correspond to exactly  $20\%$ of the left and right boundaries of the sample. We also monotonically increase $\overline{\theta}$, so that both samples compress along the horizontal direction, and enforce the planarity condition $\mathbf{y}_h \cdot \mathbf{e}_3 = 0$.  Finally, the moduli parameters  are chosen consistently with (\ref{eq:modConsistency}) as\footnote{$c_2$ is not needed for these simulations due to the planarity condition. However,  we report it since we will use the same parameters to study the origami under transverse bending loads in the next section.} $c_1 = 5, c_2 =1, d_1 = d_2 = 10^{-2}$ and $d_3 = 0.1$; also $L_{\Omega} = 1.7 $ for Miura and $= \sqrt{2}$ for Eggbox.

Plots of  the normalized force versus the pinching strain are shown in Fig.\;\ref{Fig:Miura_Bowtie}(a) for the Miura and Fig.\;\ref{Fig:Eggbox_Bowtie}(a) for the Eggbox. The force is calculated numerically using a variational method detailed in Appendix \ref{sec:forces}; it is scaled by the folding modulus $d_3$ to give the normalized force in the plots. We also reproduce from Section \ref{ssec:PureMech} the pure mechanism results on these plots, as well as the case where the boundary condition $\mathbf{y}_h = \mathbf{A}(\overline{\theta}) \mathbf{x}$ is prescribed on the full left and right parts of the boundary of $\partial \Omega$. Figs.\;\ref{Fig:Miura_Bowtie} and \ref{Fig:Eggbox_Bowtie}(b) display the variation of the angle field versus the pinching strain  in each simulation and for each boundary condition, while Figs.\;\ref{Fig:Miura_Bowtie} and \ref{Fig:Eggbox_Bowtie}(c-e)  highlight representative examples of the pinching simulations at large strains.  Finally, Table \ref{tab:Table1} compares the different energy contributions for the simulations at  $20\%$ strain intervals.

\setlength{\tabcolsep}{12pt}
\renewcommand{\arraystretch}{1.4}

\begin{table}[h!]
\begin{center}
\begin{tabular}{||c | c |  c | c |  c||} 
 \hline 
Miura Pinching Simulations   &  20\% &   40 \%    &  60\% & 80\%    \\ [.6ex] 
 \hline\hline 
$\int_{\Omega} W_1(\theta, \mathbf{I}(\mathbf{y}))dx$  &   $1.81\times10^{-2}$  & $3.34\times10^{-2}$ & $5.31 \times 10^{-2}$  & $6.77\times10^{-2}$ \\  
 \hline
$\int_{\Omega} d_1L_{\Omega}^2  |\nabla \nabla \mathbf{y} |^2 dx $   &  $6.10 \times 10^{-3}$  & $1.46 \times 10^{-2}$ & $3.19 \times 10^{-2}$ & $6.80 \times 10^{-2}$ \\
 \hline
  $\int_{\Omega} d_2L_{\Omega}^2  |\nabla \theta |^2 dx $  &  $7.34 \times 10^{-3}$ & $1.51 \times 10^{-2}$  & $2.83 \times 10^{-2}$  & $5.04 \times 10^{-2}$ \\  
 \hline
  $\int_{\Omega} d_3 | \theta |^2 dx $  &  $1.58 \times 10^{-1}$ & $4.49 \times 10^{-1}$ & $7.86 \times 10^{-1}$ & $1.09$  \\  
 \hline \hline 
Eggbox Pinching Simulations   & 20\% &   40 \%    &  60\% & 80\%    \\ [.6ex] 
 \hline\hline 
 $\int_{\Omega} W_1(\theta, \mathbf{I}(\mathbf{y}))dx$  &   $4.16 \times 10^{-4}$ & $1.05 \times 10^{-3}$ & $1.89 \times 10^{-3}$ & $3.29 \times 10^{-3}$ \\  
 \hline
$\int_{\Omega} d_1L_{\Omega}^2  |\nabla \nabla \mathbf{y} |^2 dx $   &  $1.29 \times 10^{-3}$ & $3.42 \times 10^{-3}$ & $6.25 \times 10^{-3}$  & $1.15 \times 10^{-2}$ \\
 \hline
  $\int_{\Omega} d_2L_{\Omega}^2  |\nabla \theta |^2 dx $  &  $1.32 \times 10^{-4}$ & $2.46 \times 10^{-4}$  & $3.16 \times 10^{-4}$  & $3.75 \times 10^{-4}$ \\  
 \hline
  $\int_{\Omega} d_3 | \theta |^2 dx $  &  $3.75 \times 10^{-3}$ & $1.13 \times 10^{-2}$ & $2.02 \times 10^{-2}$ & $2.99 \times 10^{-2}$   \\  
 \hline 
\end{tabular}
\caption{Comparison of the different energy contributions during the pinching simulations for $c_1 = 5, d_1 = d_2 = 10^{-2}, d_3 = 0.1$ and $L_{\Omega} = 1.7 $ for Miura and $= \sqrt{2}$ for Eggbox. Snapshots of each energy are taken at $20\%$ strain intervals. }
\label{tab:Table1}
\end{center}
\end{table}

\begin{figure}[t!]
\centering
\includegraphics[width=1\textwidth]{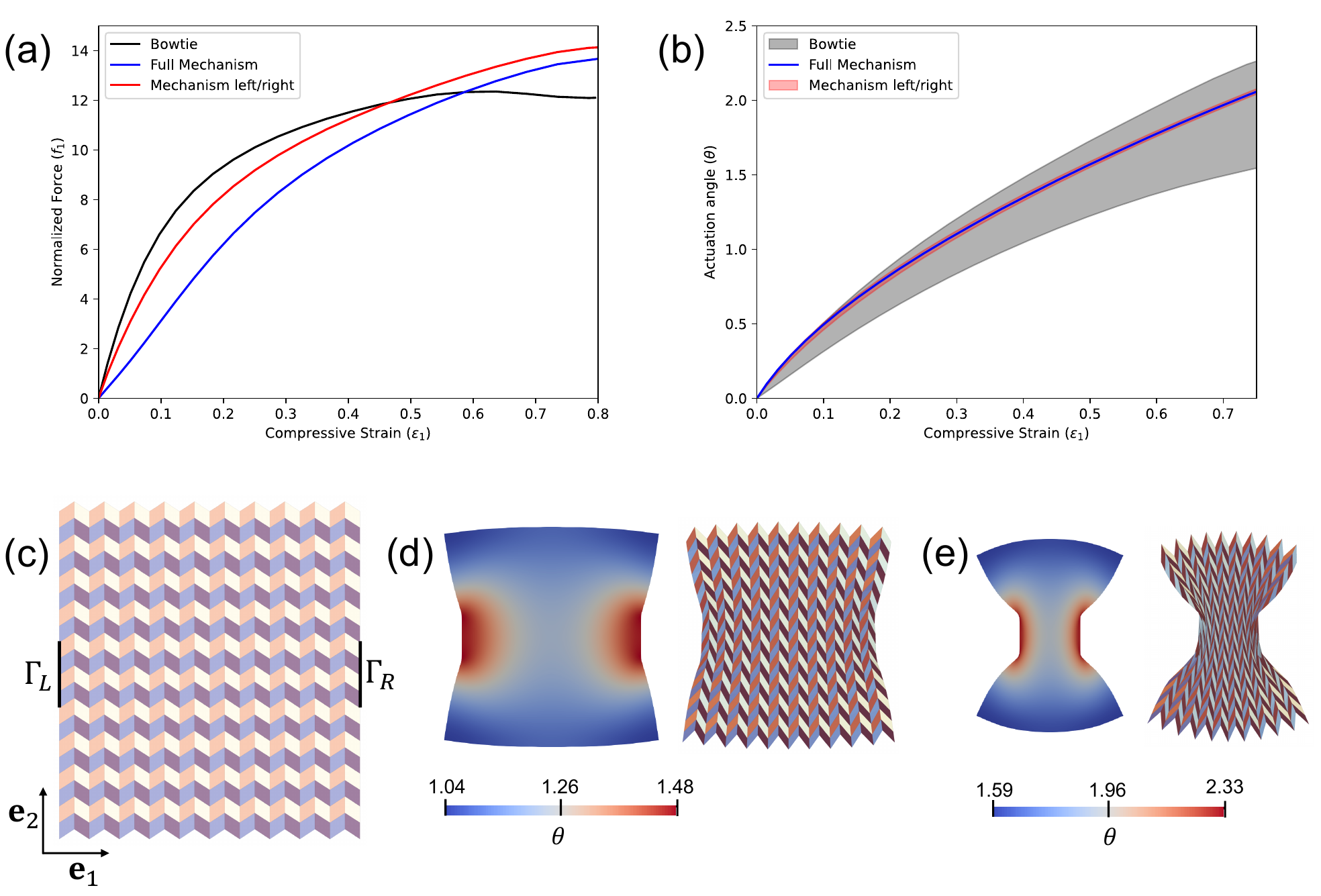} 
\caption{Pinching simulations for Miura origami. (a) The normalized force versus compressive strain for  the pinching simulations (black) compared to that of the pure mechanism (red) and the case where the full left and right boundaries are displaced. (b) The angle envelop (minimum to maximum) at each value of compressive strain. The angle is constant for the pure mechanism. (c) Reference configuration prior to pinching; (d) 40\% compressive strain; (e) 80\% compressive strain. } 
\label{Fig:Miura_Bowtie}
\end{figure}

As the boundaries are pinched, the folds must actuate in regions around $\Gamma_{\text{d}} = \Gamma_{L} \cup \Gamma_R$ to accommodate the underlying compression.  This actuation, being energetically costly, is then gradually relaxed in a manner regularized by the $\nabla \theta_h$ and $\nabla \nabla \mathbf{y}_h$ terms in the energy,  yielding the bowtie shapes shown. Evidently, all of this is done without incurring too much bulk energy: notably, the bulk energy modulus $c_1$ is $50$ times larger that the folding modulus $d_3$,  yet the bulk energy contribution is consistently one order of magnitude lower than  the dominant folding energy (see Table \ref{tab:Table1}).  Consequently, much like the pure mechanism, the folding modulus $d_3$ sets the  force for the pinching simulations. Intriguingly,  the total force to pinch the sample at a given engineering strain is comparable to  the pure mechanism,  and in fact larger for most strains.  This is likely  due to the basic nature of the  boundary forces in our generalized continuum. They consist of  two terms --- a traditional Piola-Kirchhoff stress $\mathbf{P}(\mathbf{y}_h, \theta_h)$ and a generalized stress proportional to $\nabla \mathcal{H}(\mathbf{y}_h, \theta_h)$ (see Appendix \ref{sec:forces}). The former is comparable for both cases; the stress is slightly higher for the pinching simulations, but covers only $20\%$ of the boundary. The latter, however, is zero for the pure mechanism since the deformation and actuation in this case is uniform. It is alternatively quite high in the pinching simulations since the higher-order gradients contributing to this term are magnified near the compressed boundaries.

\begin{figure}[t!]
\centering
\includegraphics[width=1\textwidth]{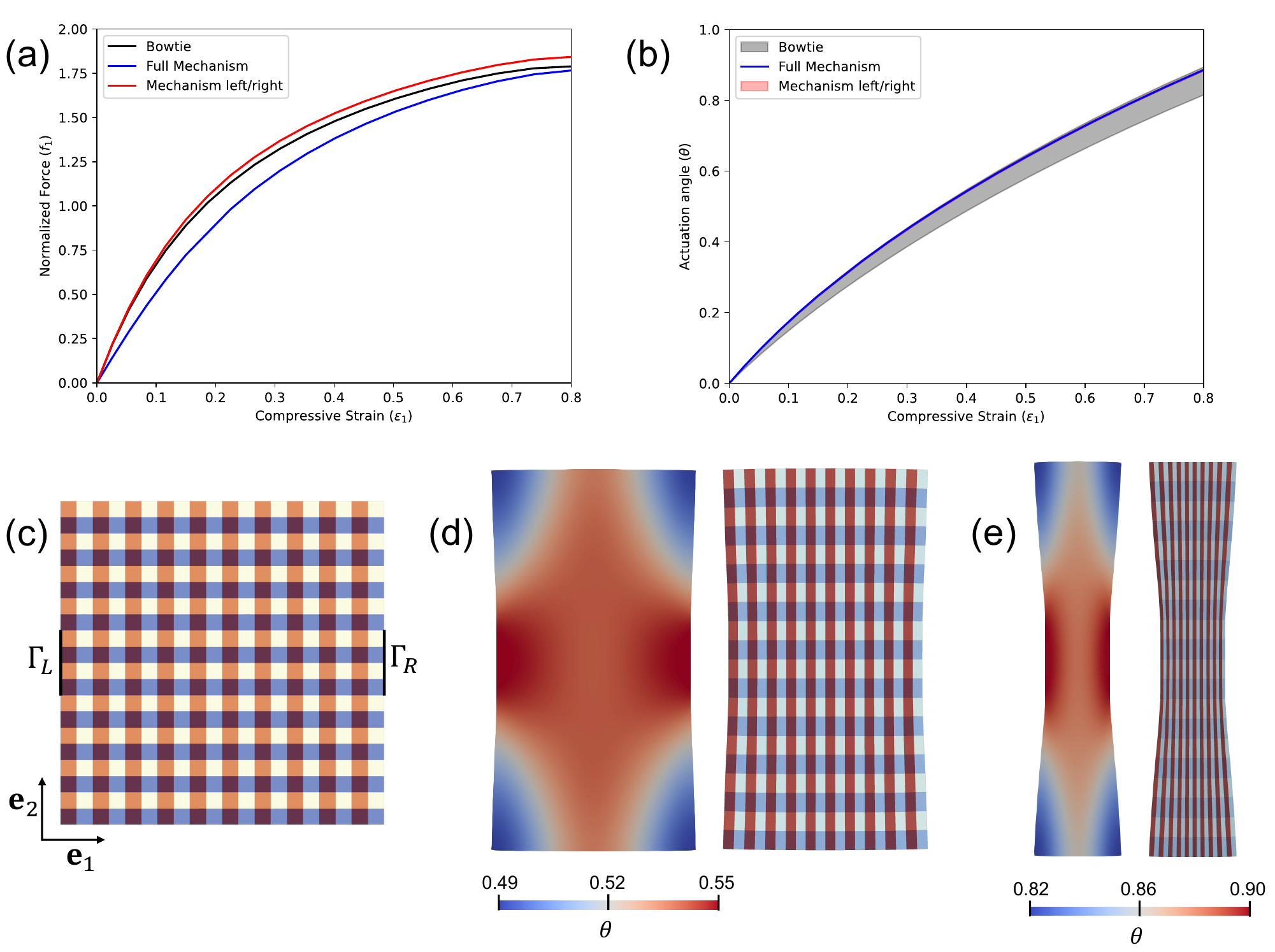} 
\caption{Pinching simulations for Eggbox origami. (a) The normalized force versus compressive strain for  the pinching simulations (black) compared to that of the pure mechanism (red) and the case where the full left and right boundaries are displaced. (b) The angle envelop (minimum to maximum) at each value of compressive strain. The angle is constant for the pure mechanism. (c) Reference configuration prior to pinching; (d) 40\% compressive strain; (e) 80\% compressive strain. } 
\label{Fig:Eggbox_Bowtie}
\end{figure}

Beyond these general observations, which hold similarly for the Miura and Eggbox cases, there are distinct differences in the two pinching simulations. Notice that the angle variation at each value of compressive strain is much higher for the Miura case (Fig.\;\ref{Fig:Miura_Bowtie}(b)) compared to the Eggbox (Fig.\;\ref{Fig:Eggbox_Bowtie}(b)). The contours illustrating this variation are also qualitatively quite different. For the Miura in Fig.\;\ref{Fig:Miura_Bowtie}(c-e), they form ellipses around the compressive loading  on $\Gamma_{L}$ and $\Gamma_{R}$. In Eggbox, there is a distinct diamond shape of nearly constant actuation in the bulk of the sample in Fig.\;\ref{Fig:Eggbox_Bowtie}(c-e).     Accordingly, the shape of the Miura is a  pronounced bowtie, while that of the Eggbox more closely  resembles the pure mechanism.    

 An explanation for these  differences lies in the auxeticity of the pattern.  Indeed, the simulated fields satisfy $\mathbf{y}_h \approx \mathbf{y}_{\text{g}}$ and $\theta_h \approx \theta_{\text{g}}$ for some 2D effective deformation and angle field geometrically constrained by\footnote{$\mathbf{II}(\mathbf{y}_{\text{g}})= \mathbf{0}$ because we are dealing with 2D loading conditions for which $\mathbf{y}_{\text{g}}\cdot \mathbf{e}_3 = 0$.}
 \begin{equation}
 \begin{aligned}
 \mathbf{I}(\mathbf{y}_{\text{g}}) = \mathbf{A}^T(\theta_{\text{g}}) \mathbf{A}(\theta_{\text{g}}) = \begin{pmatrix}  \tfrac{\lambda^2_1(\theta_{\text{g}})}{\lambda^2_1(0)} & 0 \\ 0 &  \tfrac{\lambda^2_2(\theta_{\text{g}})}{\lambda^2_2(0)}  \end{pmatrix} \quad \text{ and } \quad \mathbf{II}(\mathbf{y}_{\text{g}}) = \mathbf{0},
 \end{aligned}
 \end{equation}
since the bulk energy in each simulation is negligible compared to its modulus. 
 An application of Gauss's theorem egregium to these formulas then furnishes a PDE governing $\theta_\text{g}$ of the form
 \begin{equation}
 \begin{aligned}\label{eq:EllipHyperbolic}
 \partial_1 \Big( \frac{\lambda_1(0) \lambda_2'(\theta_{\text{g}})}{\lambda_2(0)\lambda_1(\theta_{\text{g}})} \partial_1 \theta_{\text{g}} \Big) +  \partial_2 \Big( \frac{\lambda_2(0) \lambda_1'(\theta_{\text{g}})}{\lambda_1(0)\lambda_2(\theta_{\text{g}})} \partial_2 \theta_{\text{g}}\Big)  = 0. 
 \end{aligned}
 \end{equation}
 This second-order PDE is elliptic if $\lambda_1'(\theta) \lambda_2'(\theta) > 0$ for all $\theta$ and hyperbolic if $\lambda_1'(\theta) \lambda_2'(\theta) < 0$ for all $\theta$. In other words, it is elliptic for patterns like Miura origami that are auxetic and hyperbolic for those like Eggbox that are not auxetic. 
 
 Ellipticity versus hyperbolicity is a standard classification in PDE theory, with clear and general ramifications. Elliptic PDEs describe a decay in the field variables away from loads, highlighted by the fact that the maximum and minimum of the field variables always occurs on the boundary and never on the interior unless they are constant.\footnote{This is called the strong maximum principle of Elliptic PDEs.} The Miura in Fig.\;\ref{Fig:Miura_Bowtie}(c-e) perfectly reflects these properties --- the angle field $\theta_h$ is largest at the loaded boundary, decays away along elliptic contours from this boundary,  and is at its minimum at the four corners.   Hyperbolic PDEs are the diametric opposite. They describe wave-like phenomena, where the field variable instead propagates from the loaded boundary to the interior of the domain along characteristic curves. Characteristics also cannot intersect, making the solutions to hyperbolic PDEs generally more rigid to Dirichlet boundary data. The Eggbox in \ref{Fig:Eggbox_Bowtie}(c-e) reflects these properties --- a large diamond region of nearly constant actuation on the interior of the domain  recedes to the corners  through ``fan-like" characteristics; the pattern also displays far less heterogeneity, thus less ``flexibility" in a sense, than its Miura counterpart.

An astute observer may notice that the Eggbox case does not completely fit the description of a hyperbolic PDE without a little embellishment --- the actuation is still largest at the loaded boundaries and smallest at the corners, like in the Miura case. To explain, we must remember that $\theta_h$ only approximates $\theta_{\text{g}}$. In fact, there is a universal elliptic part to our generalized continuum model. The governing equations  in (\ref{eq:governingEquations}) include the equilibrium condition $\nabla \cdot \mathbf{j}(\theta_h) = q(\mathbf{y}_h, \theta_h)$, which is elliptic in $\theta_h$ since  $\nabla \cdot \mathbf{j}(\theta_h) \sim \nabla^2 \theta_h$.  This complicates the qualitative behavior of the Eggbox. In particular,   $\theta_h$ must   approximate  an angle field $\theta_{\text{g}}$  associated to the hyperbolic equation in (\ref{eq:EllipHyperbolic})  while simultaneously solving  an elliptic equilibrium equation.

\subsection{Bending deformations under transverse loads }

For our final set of examples, we examine  the behavior of Miura and Eggbox origami under transverse loads. The basic setup is shown in Fig.\;\ref{Fig:TransverseBending}(a). We apply a uniform transverse load $\overline{\mathbf{b}} = b \mathbf{e}_3$ (for a constant $b$) to the effective domain $\Omega$ of each pattern and pin the boundaries via the Dirichlet conditions $\mathbf{y}_h = \mathbf{x}$ on $\Gamma_{L}$ and $\Gamma_R$, which are centered on the sample. We study this boundary value problem in the cases that $\Gamma_{L,R}$ take up $5\%$, $10\%$, and $20\%$ of the left and right boundaries, respectively. By partially pinning these boundaries, as opposed to say fully pinning them,  each  pattern can accommodate the load via  doubly-curved shapes with minimal geometric frustration.  The choice of parameters in these simulations is the same as the pinching case:  $c_1 = 5, c_2 =1, d_1 = d_2 = 10^{-2}$, $d_3 = 0.1$ and $L_{\Omega} = 1.7 $ for Miura and $= \sqrt{2}$ for Eggbox. 

Fig.\;\ref{Fig:TransverseBending}(b) plots the average out-of-plane displacement $\bar{w} = \frac{1}{|\Omega|} \int_{\Omega} \mathbf{y} \cdot \mathbf{e}_3 \dif x$  versus the load $b$ for each simulation of each pattern under each boundary condition.  Fig.\;\ref{Fig:TransverseBending}(c) shows three representative deformed configurations for Miura in the $10\%$ pinned case; Fig.\;\ref{Fig:TransverseBending}(d) shows analogous such shapes for the Eggbox, again in the $10\%$ case.

\begin{figure}[t!]
\centering
\includegraphics[width=.97\textwidth]{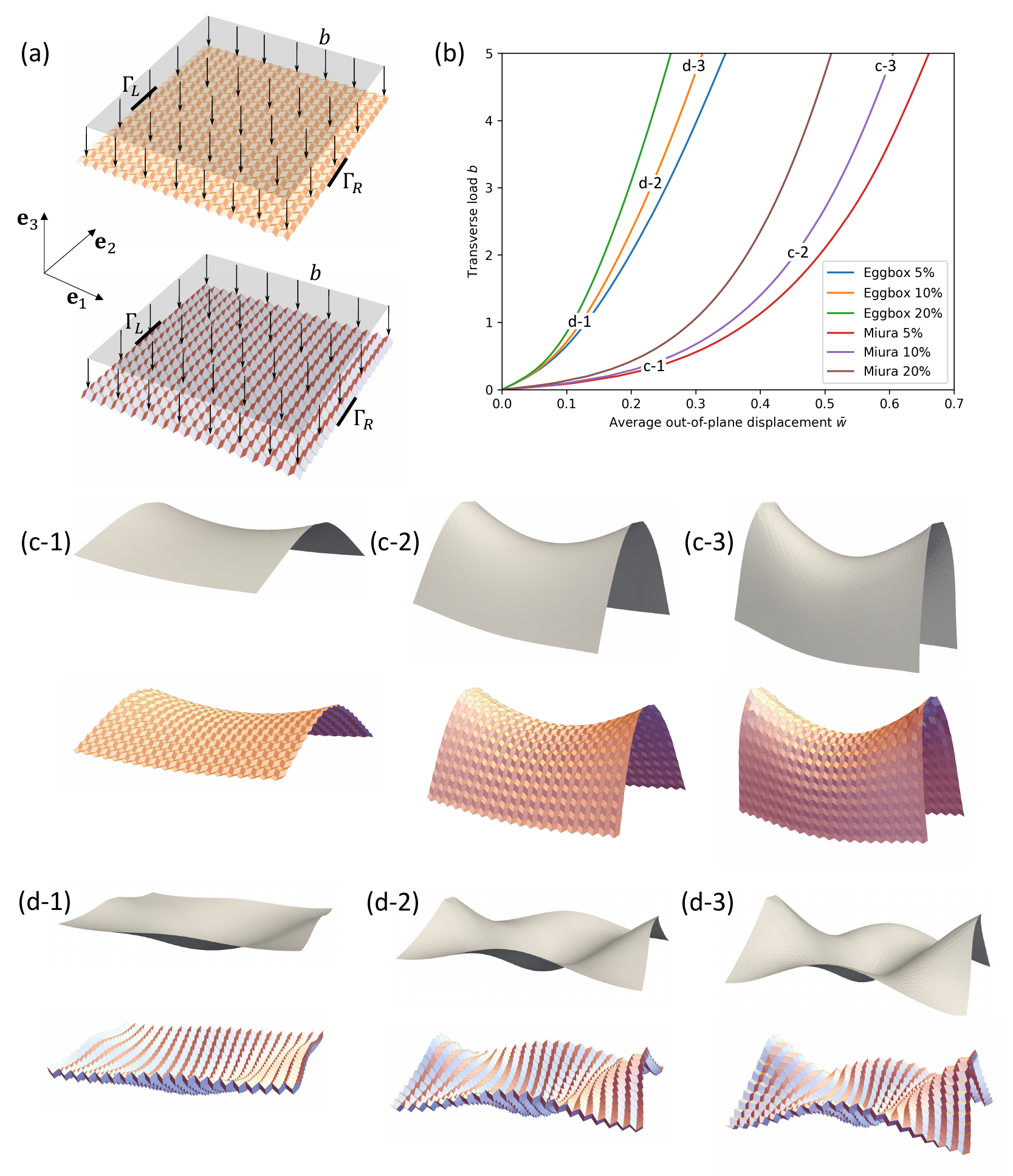} 
\caption{Bending simulations under transverse load. (a) The pattern is pinned at $\Gamma_L$ and $\Gamma_R$ and subject to the uniform load $b$. (b) Plot of the average displacement $\bar{w} = \frac{1}{|\Omega|} \int_{\Omega} \mathbf{y} \cdot \mathbf{e}_3 \dif x$ versus the load. (c) Deformation profiles for Miura and (d) for Eggbox at representative points on the loading curves, labeled in (b). Each of these profiles is for the case where  $10\%$ of the left and right boundaries is pinned.  } 
\label{Fig:TransverseBending}
\end{figure}

At the early stages of loading, both Miura and Eggbox deform quite softly, with the stored energy is  dominated by the bending term $\int_{\Omega} d_1 |\nabla \nabla \mathbf{y}|^2 \dif x$. However,   geometric frustration sets in as the loading progresses.  At the tail end of the simulations, the two bulk energy terms $\int_{\Omega} W_1(\theta, \mathbf{II}(\mathbf{y})) \dif x$ and $\int_{\Omega} W_2(\theta, \mathbf{II}(\mathbf{y})) \dif x$ become comparable to the bending term.  This transition --- from  negligible to prominent bulk stored energy --- coincides with the universal stiffening response shown in the plots, and is likely its cause. Another notable, if unsurprising, feature of the force plot is that the patterns are stiffer for the $20\%$ pinned boundary compared to the $10\%$ and to the $5\%$ cases. This is easily explained by noticing that  each pattern needs to flatten out in the $x_2$-direction toward the $\Gamma_{L,R}$ boundaries, and thus become approximately singly curved in  small  subregions of $\Omega$ that contain these boundaries. A significant bulk energy  contribution $W_2(\theta, \mathbf{II}(\mathbf{y}))$ necessarily arises in these regions because $W_2$ biases the patterns to be doubly-curved.  The last and most surprising observation from these plots is that the bending stiffness of Eggbox is much larger than that of  Miura. This  is true even though the in-plane stiffness of the Miura greatly exceeds the Eggbox for both the pure mechanism and the pinching simulations, illustrated previously. 

We now spend some time analyzing and discussing the differences in bending stiffness between the two examples. As a basic heuristic, we believe this difference  is essentially explained by the fact that the saddle-like shapes preferred by the Miura can do more work $\int_{\Omega} b \mathbf{y} \cdot \mathbf{e}_3 \dif x$ at each value of $b$ than the analogous shapes for the Eggbox.  Notice that the  the cap-like shapes common to Eggbox are not observed in the simulations. Instead, the pattern forms a cap in its center but then reverses its curvature towards the boundary to induce more vertical displacement and thus more work. This geometric frustration underlies the Eggbox's  much stiffer response. 

We can quantitatively investigate this heuristic, at least in a linearized regime $|b| \ll 1$ where the analysis is more tractable. The idea is to minimize the total potential energy at leading order  for bending-type deformations of the form
\begin{equation}
\begin{aligned}
\mathbf{y}_b(\mathbf{x}) = \begin{pmatrix} \mathbf{x} \\ b w(\mathbf{x})  \end{pmatrix} + O(b^2), \quad \theta_b(\mathbf{x}) = O(b^2), 
\end{aligned}
\end{equation}
with $b$ playing the role of the small parameter in the asymptotics.  Notice that the out-of-plane displacement  $b w$  is taken to be much larger than the planar components,  as one would expect under small transverse loading.\footnote{At first glance, it's not obvious that the actuation should scale as $b^2$ for small $b$. However,  after expanding $\theta_b$ in powers of $b$ via $\theta_b =  \delta \theta_0 + b \delta \theta_1 + b^2 \delta \theta_2 + \ldots$, one will find upon energy minimization that $\delta \theta_0 = \delta \theta_1 =0$.}
Substituting these deformations into the total energy leads to 
\begin{equation}
\begin{aligned}
E(\mathbf{y}_b, \theta_b) =  b^2 \Big( c_2 L_{\Omega}^2 \int_{\Omega} \Big( \frac{\lambda_1'(0)}{\lambda_1(0)}   \partial_2 \partial_2 w +  \frac{\lambda_2'(0)}{\lambda_2(0)}  \partial_1 \partial_1 w \Big)^2  \dif x +   d_1 L_{\Omega}^2  \int_{\Omega} |\nabla \nabla w|^2 \dif x  - \int_{\Omega} w \dif x \Big) + o(b^2). 
\end{aligned}
\end{equation}
As  the moduli satisfy $c_2 \gg d_1$ in the typical case, it is reasonable to impose the PDE constraint
\begin{equation}
\begin{aligned}\label{eq:PDEOutOfPlane}
 \frac{\lambda_1'(0)}{\lambda_1(0)}   \partial_2 \partial_2 w +  \frac{\lambda_2'(0)}{\lambda_2(0)}  \partial_1 \partial_1 w = 0
\end{aligned}
\end{equation}
upon minimization. This  requirement  results in a small correction to the overall minimum energy that can be neglected for our purposes. It follows that
\begin{equation}\label{eq:OutOfPlaneMin}
\begin{aligned}
\min_{\substack{\mathbf{y}_b \colon \Omega \rightarrow \mathbb{R}^3 \\ \theta_b \colon \Omega \rightarrow  (\theta^{-},\theta^+) \\ \mathbf{y}_b = (\mathbf{x},  0)  \text{ on } \Gamma_{L,R}}} E(\mathbf{y}_b, \theta_b) \approx  \min_{\substack{ w \colon \Omega \rightarrow \mathbb{R}^3 \\  w \text{ solves } (\ref{eq:PDEOutOfPlane}) \\  w = 0 \text{ on } \Gamma_{L,R} }}  b^2  \int_{\Omega}  \Big(   d_1 L_{\Omega}^2 |\nabla \nabla w|^2 - w \Big)   \dif x.
\end{aligned}
\end{equation}
for all sufficiently small $b$.    Note that $b$ is the load per unit area and $\frac{1}{|\Omega|}  \int_{\Omega} b w \dif x$ is the average displacement under this load. Thus, if $w^{\star}$ is a minimizer to the latter, then the normalized bending stiffness relevant to the plots in Fig.\;\ref{Fig:TransverseBending}(b) is  the ratio 
\begin{equation}
\begin{aligned}\label{eq:normalizedStiffness}
K_{\text{bend}} = \frac{1}{\frac{1}{|\Omega|} \int_{\Omega} w^{\star} \dif x }.
\end{aligned}
\end{equation}

Instead of trying to perform the PDE-constrained optimization in (\ref{eq:OutOfPlaneMin}) --- which is challenging and likely overkill for the task of explaining the  differences in stiffness between the two examples --- we pursue a minimization within an ansatz. First, to simplify the analysis, we replace the reference configuration in (\ref{eq:DomainForExamples}) with one centered at the origin $\Omega = (-\lambda_1(0)/2,\lambda_1(0)/2)\times (-\lambda_2(0)/2, \lambda_2(0)/2)$ so that the minimizer $w^{\star}$ to (\ref{eq:OutOfPlaneMin}) is even in  $x_1$ and $x_2$. The most general fourth order-polynomial that has this symmetry and solves the PDE in (\ref{eq:PDEOutOfPlane}) is 
\begin{equation}
\begin{aligned}
w_{\text{poly}}(x_1, x_2) =  w_0 + \kappa_0 \Big(   \frac{\lambda_1'(0)}{\lambda_1(0)}  x_1^2 - \frac{\lambda_2'(0)}{\lambda_2(0)} x_2^2  \Big) +  \tau_0 \Big( \frac{\lambda'_1(0) \lambda_2(0)}{6\lambda_2'(0) \lambda_1(0)} x_1^4 - x_1^2 x_2^2 +  \frac{\lambda'_2(0) \lambda_1(0)}{6\lambda_1'(0) \lambda_2(0)} x_2^4 \Big)
\end{aligned}
\end{equation}
for constants $w_0,\kappa_0$ and $\tau_0$. 
We supply the constant $w_0 = w_0(\kappa_0, \tau_0)$ in this formula by imposing the condition $w_{\text{poly}}(\lambda_1(0)/2,0) = 0$, which is a proxy for the boundary data $w = 0$ on $\Gamma_{L,R}$ in the minimization in (\ref{eq:OutOfPlaneMin}). Hence,  a ``back-of-the-envelope" prediction of the normalized bending stiffness  is 
\begin{equation}
\begin{aligned}
K_{\text{bend}} \approx \frac{1}{\frac{1}{|\Omega|} \int_{\Omega} w_{\text{poly}}^{\star} \dif x },
\end{aligned}
\end{equation}
where $w_{\text{poly}}^{\star}$  solves the simple minimization  problem
\begin{equation}
\begin{aligned}
w^{\star}_{\text{poly}} = \argmin_{\kappa_0, \tau_0}\int_{\Omega}  \Big(   d_1 L_{\Omega}^2 |\nabla \nabla w_{\text{poly}}|^2 - w_{\text{poly}} \Big)   \dif x.
\end{aligned}
\end{equation}

As is typical with these calculations, the formulas for the minimizing  coefficients and the approximate stiffness  are  lengthy expressions of $\lambda_{1,2}(0)$ and $\lambda_{1,2}'(0)$, which we do not report. After substituting  for their explicit values using (\ref{eq:lambdaMO}) and (\ref{eq:lambdaEO}), the stiffness is 
\begin{equation}
\begin{aligned}\label{eq:approxNStiffness}
\frac{1}{\frac{1}{|\Omega|} \int_{\Omega} w_{\text{poly}}^{\star} \dif x } = d_1 L_{\Omega}^2 \begin{cases}
31.9 &\text{ for Miura} \\ 
322.2 &\text{ for Eggbox}. 
\end{cases}
\end{aligned}
\end{equation}
Then, using the values of $d_1$ and $L_{\Omega}$  from the simulations, we calculate  from  (\ref{eq:approxNStiffness})  a normalized stiffness of $ 0.92$ for Miura and $6.44$ for Eggbox. The corresponding simulated values at the origin are $0.72$ and $4.44$ for Miura and Eggbox, respectively. 
So this   expression  gets the trends right, namely that Eggbox is stiffer than the  Miura under small transverse loads, but also provides reasonable quantitative agreement in spite of all the simplifications.  

One nice feature of capturing the trends in a simplified setting is that it enables  design optimization without taxing a supercomputer. As a quick example, we can replace the expressions for $\lambda_i(\theta)$ in (\ref{eq:lambdaMO}) and (\ref{eq:lambdaEO}) with $\lambda_i(\theta+ \psi)$ for an angle field $\psi$, which alters the reference folded state of the unit cells, and then optimize the stiffness as a function of $\psi$ using our simplified model. In the Miura case,  the stiffness is unsurprisingly monotonic --- it achieves its minimum at the fully flat state and  becomes infinite at the fully folded state. The Eggbox too  has a stiffness that approaches infinity at the fully folded state. However, it is  most compliant  at the nontrivial value $\psi \approx 0.21$. 

We end by  pointing out an intriguing connection between pinching and linearized bending in these metamaterials: Both situations are well characterized by underlying second-order PDEs whose ``PDE type", and thus its basic qualitative properties, is completely dictated by the pattern's auxeticity.  Specifically, the actuation field in the pinching simulations is well approximated by the PDE for $\theta_{\text{g}}$ in (\ref{eq:EllipHyperbolic}), while the out-of-plane displacement $pw$ under small transversed load is constrained by (\ref{eq:PDEOutOfPlane}).  Both PDEs are elliptic for Miura Origami (or really any auxetic parallelogram origami) and hyperbolic for Eggbox (any pattern that is not auxetic). More generally, they emerge as different analytical case-studies of the geometric constraints $\mathbf{I}(\mathbf{y}) = \mathbf{A}^T(\theta) \mathbf{A}(\theta)$ and $\big(\mathbf{u}(\theta) \cdot \mathbf{u}'(\theta)\big) \tilde{\mathbf{v}}_0 \cdot \mathbf{II}(\mathbf{y})\tilde{\mathbf{v}}_0 = -   \big(\mathbf{v}(\theta) \cdot \mathbf{v}'(\theta)\big) \tilde{\mathbf{u}}_0 \cdot \mathbf{II}(\mathbf{y})\tilde{\mathbf{u}}_0$ that underly our modeling framework. In fact, we anticipate that any such case-study of these  constraints will lead to PDE compatibility conditions that illuminate an ``elliptic $=$ auxetic" versus ``hyperbolic $=$ non-auxetic" dichotomy in the response, although we leave this point to be examined more carefully for another time. 

\section{Discussion}\label{sec:Conclusion}

This paper brought together  rigorous theory, modeling, numerical method development, simulations, and mechanical analysis to address the challenging problem  of understanding the effective elastic behavior of parallelogram origami. We first described a pair of coarse-graining rules that produced effective   geometric constraints for  origami soft modes and  a corresponding plate energy for their elasticity.  Then, we introduced a simplified elastic model  that drove the pattern’s effective deformation and actuation towards the geometric constraints of this theory, while accounting for higher-order sources of elasticity  --- the plate energy ---  through an appropriate regularization. 
 Next, we provided a finite element formulation of this model using the $C^0$ interior penalty method.  Finally, through numerical implementation in \texttt{Firedrake}, we demonstrated a model and numerical framework capable of simulating the effective behavior of parallelogram origami metamaterials across a variety of loading conditions
 
 \begin{figure}[t!]
\centering
\includegraphics[width=0.9\textwidth]{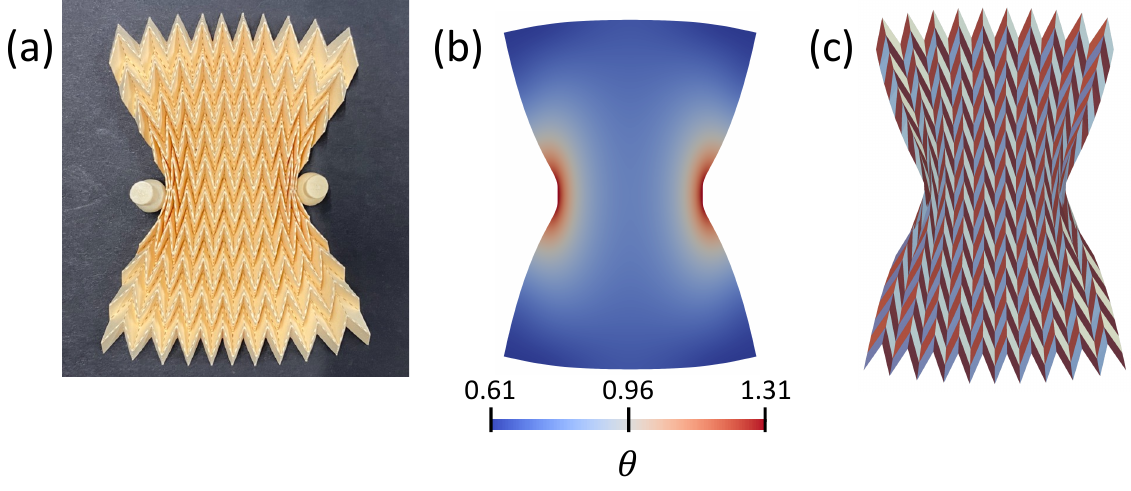} 
\caption{Comparing the model to an experiment. (a) A Miura origami made of perforated and folded paper is pinched at $\approx 63\%$ engineering strain. (b) Simulations of the effective model, with the parameters tuned to fit the experiment. (c) The corresponding origami deformation of the simulation.} 
\label{Fig:ExperimentalComparison}
\end{figure}

Challenges and opportunities remain. It would certainly be appealing to compare our model and numerics to experiments, both qualitatively and quantitatively. We anticipate that, with proper tuning of the parameters, our approach will capture the gross-shape change observed in experiments across designs and loading conditions. As a demonstration, Fig.\;\ref{Fig:ExperimentalComparison} highlights the model's ability to capture qualitative features of the pinched Miura origami sample  in Fig.\;\ref{Fig:IntroFig}(a). In brief, the experimental sample has panels of length $1$ cm and sector angle $\pi/3$  (see Fig.\;\ref{Fig:CellsRefSec5}(a)), and its natural partly folded reference configuration has a cell width $\lambda_2(0) \approx 1.25$ cm (see Fig.\;\ref{Fig:CellsRefSec5}(b)). Thus,  for the simulations, we set $\lambda_{1,2}(\theta) = \lambda_{1,2}^{\text{(MO)}}(0.35 \pi + \theta)$ from (\ref{eq:lambdaMO}), $\Omega = (0,\lambda_1(0)) \times (0, \lambda_2(0))$ for $\lambda_1(0) \approx 1.19$ and $\lambda_2(0) \approx 1.25 $, and $L_{\Omega} = 1.2$. The sample is pinched to $\approx 63\%$ engineering compressive strain by the pins, which are roughly $1$ unit cell in length. So we used the pinching boundary conditions in Section \ref{ssec:PinchingSimulations} for the simulations, with $\Gamma_{L,R}$ making up $10\%$ of the boundaries  and with $\varepsilon_1(\overline{\theta})  = 0.63$ (see (\ref{eq:normForce})).  The simulations are performed for  $c_1 = 5$ and $d_3 = 10^{-1}$ for simplicity\footnote{One really needs access to  experimental loading curves to choose these parameters more systematically.}, and for various  values of $d_1 = d_2$ to find one  that compares well with the experiment. Fig\,\ref{Fig:ExperimentalComparison}(b) shows our best match, which is for $d_1 = d_2 = 4 \times 10^{-3}$;  Fig.\;\ref{Fig:ExperimentalComparison}(c) shows the corresponding origami deformation. All told, the agreement is excellent.  

This success notwithstanding, matching  loading curves to experiments may be more challenging.   Potential improvements are readily available. We took a  purposefully simple approach to modeling the higher-order effects in the theory.  It is possible, for instance, to replace the $\int_{\Omega} ( d_1L_{\Omega}^2 |\nabla \nabla \mathbf{y}|^2 + d_2L_{\Omega}^2 |\nabla \theta|^2) \dif x$ term in the energy with the more complicated but asymptotically correct effective plate energy in  (\ref{eq:effectivePlateTheory}). A mathematical challenge  is to prove the existence of minimizers in this setting --- the proof  in Theorem \ref{ExistenceTheorem} makes key use of the fact that the $L^2$-norm of $\nabla \nabla \mathbf{y}$ is controlled by the elastic energy, and will need to be adapted.   Another natural modification is to replace the actuation energy $\int_{\Omega} d_3 |\theta|^2 \dif x$ with a weighted average $\int_{\Omega} d_3 \sum_{i =1,\ldots,4} \lambda_i |\gamma_i(\theta)|^2 \dif x$, where $\gamma_i \colon (\theta^{-}, \theta^+) \rightarrow \mathbb{R}$ parameterize the change in the four distinct dihedral angles of the unit cell and $\lambda_i$ are the weights.\footnote{Presumably, each $\lambda_i$ should be determined geometrically by the corresponding creases perimeter or some notion of  its influence area.}  Whether such modifications are ``worth the bang for your buck" will require synergistic modeling and experimental efforts going forward.

 It would also be appealing to use our framework to explore large swaths design space of parallelogram origami, beyond the canonical Miura and Eggbox origami. Importantly, the basic setup of our model allows for any parallelogram origami design. So all that needs to be done to begin this exploration is to feed the appropriate design dependent parameterizations of $\mathbf{u}(\theta)$ and $\mathbf{v}(\theta)$ into the model. Another intriguing research direction is to study ``locally periodic" origami patterns, which are obtained by slowly varying the design parameters of the origami cells and often show an uncanny ability to approximate a rich variety of surfaces on folding \cite{dang2022inverse,dudte2016programming,sardas2024continuum}. The lattice vectors in our model $\mathbf{u}_0 \equiv \mathbf{u}_0(\mathbf{d}),\ldots,  \mathbf{v}(\theta) \equiv \mathbf{v}(\theta, \mathbf{d})$ implicitly depend on design angles and lengths, labeled $\mathbf{d}$. In a locally periodic setting, $\mathbf{d} = \mathbf{d}(\mathbf{x})$ spatially varies in the reference domain $\Omega$, introducing a material anisotropy to the model. We also anticipate the emergence of an effective elasticity proportional to $\nabla \mathbf{d}$  on coarse-graining the mechanical response of these patterns, although pinning down the exact details  of this elasticity seems far from trivial. 

As a final point, flexible mechanical metamaterials are appealing in applications both for their large design space and their highly design dependent  nonlinear mechanical response. It is thus important to have   modeling tools that efficiently map the relationship between design and bulk mechanical  properties in these systems. Our results here, and those of our prior works \cite{xu2024derivation,zheng2022continuum,zheng2023modelling}, show that    homogenization is a predictive and useful tool in this setting, while also  pushing the boundaries of the theory of elasticity in new and exciting directions.  We hope our  work will serve as an exemplar for further exploration along these lines.

\subsection*{Data availability}
\noindent The  code for all numerical simulations in this paper is available on GitHub at:\\
 \url{https://github.com/HuUSC/ContinuumModelOri.git}

\subsection*{Acknowledgments}
\noindent  H.X. and P.P. acknowledge support from the Army Research Office (ARO-W911NF2310137). P.P. also acknowledges support from the National Science Foundation (CMMI-CAREER-2237243). \\
\noindent 
F.M. acknowledges support from the National Science Foundation (DMS-2409926).

\appendix

\section{Proof of the existence of minimizers for the constitutive model}
\label{sec:proof}

This section proves Theorem \ref{ExistenceTheorem} on the existence of minimizers to the  energy in Section \ref{sec:modelForSimulations}. We break the proof up into several steps. 
 
\noindent \textbf{Step 1. Extensions of the energy densities.}  We find it useful to  express the internal energy $E_{\text{int}}(\mathbf{y}, \theta)$ defined by (\ref{eq:Eint}-\ref{eq:VDef}) as 
\begin{equation}\label{eq:getWepsilon}
\begin{aligned}
E_{\text{int}}(\mathbf{y}, \theta) = \int_{\Omega} W^{\varepsilon}( \theta, \nabla \theta, \nabla \mathbf{y}, \nabla \nabla \mathbf{y}) \dif x,  \quad (\mathbf{y}, \theta) \in V^{\varepsilon}
\end{aligned}
\end{equation}
for an energy density $W^{\varepsilon} \colon \mathbb{R} \times \mathbb{R}^2 \times \mathbb{R}^{3\times2} \times \mathbb{R}^{3  \times 2 \times 2} \rightarrow \mathbb{R}$ that is  extended to all of space and  continuous in every  argument. In particular, we  construct a  $W^{\varepsilon}$ that coincides with the original energy density for $E_{\text{int}}(\mathbf{y}, \theta)$ whenever $(\mathbf{y}, \theta) \in V^{\varepsilon}$, but deviates from this density in a continuous way otherwise. This modification allows us to appeal to general results in the direct method of the calculus of variations to establish the existence of minimizers of the energy.  The construction makes use of the $\varepsilon$-dependence of the set $V^{\varepsilon}$.

To begin, recall from Section \ref{ssec:Design} that $\mathbf{u}(\theta), \mathbf{v}(\theta), \mathbf{A}(\theta), \mathbf{u}'(\theta), \mathbf{v}'(\theta)$ are well-defined and smooth on the interval $(\theta^{-}, \theta^+)$. However, $\mathbf{u}'(\theta)$ and/or $\mathbf{v}'(\theta)$ can blow up as $\theta \rightarrow \theta^{+}$ and/or $\theta^{-}$ for some parallelogram origami designs. As we are really focused on the behavior of this actuation  on the interval $[\theta^{-} + \varepsilon, \theta^+ - \varepsilon]$ per $V^{\varepsilon}$,  we introduce convenient $\varepsilon$-extensions of these functions $\mathbf{u}_{\varepsilon}(\theta)$, $\mathbf{v}_{\varepsilon}(\theta)$, $\mathbf{A}_{\varepsilon}(\theta)$, $[\mathbf{u}']_{\varepsilon}(\theta)$ and $[\mathbf{v}']_{\varepsilon}(\theta)$  as follows. Define $\mathbf{u}_{\varepsilon} \colon \mathbb{R} \rightarrow \mathbb{R}^3$ as 
\begin{equation}
\begin{aligned}\label{eq:uEpsilonDef}
\mathbf{u}_{\varepsilon}(\theta) := \begin{cases}
\mathbf{u}(\theta) & \text{ if } \theta \in [\theta^{-} + \varepsilon, \theta^+ - \varepsilon] \\
\mathbf{u}(\theta^{-} + \varepsilon) & \text{ if } \theta < \theta^{-} + \varepsilon \\
\mathbf{u}(\theta^+ - \varepsilon) & \text{ if } \theta > \theta^+ - \varepsilon,
\end{cases}
\end{aligned}
\end{equation}
and $\mathbf{v}_{\varepsilon} \colon \mathbb{R} \rightarrow \mathbb{R}^3$ analogously. Then,  define $\mathbf{A}_{\varepsilon} \colon \mathbb{R} \rightarrow \mathbb{R}^{3 \times2}$ as the unique linear transformation that satisfies 
\begin{equation}
\begin{aligned}\label{eq:Aeps}
\mathbf{A}_{\varepsilon} (\theta) \tilde{\mathbf{u}}_0 = \mathbf{u}_{\varepsilon}(\theta) \quad \text{ and } \quad \mathbf{A}_{\varepsilon} (\theta) \tilde{\mathbf{v}}_0 = \mathbf{v}_{\varepsilon}(\theta),
\end{aligned}
\end{equation}
and note that this definition implies that $\mathbf{A}_{\varepsilon}(\theta) = \mathbf{A}(\theta)$ for all $\theta \in [\theta^{-} + \varepsilon, \theta^+ - \varepsilon]$. Finally, define $[\mathbf{u}']_{\varepsilon}\colon \mathbb{R} \rightarrow \mathbb{R}$ similar to (\ref{eq:uEpsilonDef}) via
\begin{equation}
\begin{aligned}\label{eq:uEpsilonPrimeDef}
[\mathbf{u}']_{\varepsilon}(\theta) := \begin{cases}
\mathbf{u}'(\theta) & \text{ if } \theta \in [\theta^{-} + \varepsilon, \theta^+ - \varepsilon] \\
\mathbf{u}'(\theta^{-} + \varepsilon) & \text{ if } \theta < \theta^{-} + \varepsilon \\
\mathbf{u}'(\theta^+ - \varepsilon) & \text{ if } \theta > \theta^+ - \varepsilon,
\end{cases}
\end{aligned}
\end{equation}
and define $[\mathbf{v}']_{\varepsilon} \colon \mathbb{R} \rightarrow \mathbb{R}^3$ analogously. Note that these extensions  are not  chosen to be $\mathbf{u}_{\varepsilon}'(\theta), \mathbf{v}'_{\varepsilon}(\theta)$ because the latter  are discontinuous at $\theta = \theta^{\mp} \pm \varepsilon$. 

Next, we address the fact that the surface normal $\mathbf{n}(\mathbf{y})$ in (\ref{eq:secFundDef}) is not well-defined when $\partial_1 \mathbf{y} \times \partial_2 \mathbf{y} = \mathbf{0}$. Embracing again the $\varepsilon$-dependence in $V^{\varepsilon}$, we  regularize  the surface normal via the continuous vector field $\boldsymbol{\nu}_{\varepsilon} \colon \mathbb{R}^{3\times2} \rightarrow \mathbb{R}^3$ defined by 
\begin{equation}
\begin{aligned}
\boldsymbol{\nu}_{\varepsilon}(\mathbf{F}) := \begin{cases}
\frac{\mathbf{F} \mathbf{e}_1 \times \mathbf{F} \mathbf{e}_2}{|\mathbf{F} \mathbf{e}_1 \times \mathbf{F} \mathbf{e}_2|} & \text{ if } |\mathbf{F} \mathbf{e}_1 \times \mathbf{F} \mathbf{e}_2| \geq \varepsilon \\
\varepsilon^{-1} \mathbf{F} \mathbf{e}_1 \times \mathbf{F} \mathbf{e}_2 & \text{ otherwise},
\end{cases}
\end{aligned}
\end{equation}
where $\{ \mathbf{e}_1, \mathbf{e}_2\}$ denotes the standard basis on $\mathbb{R}^2$.  It follows that $\boldsymbol{\nu}_{\varepsilon}(\nabla \mathbf{y}) = \mathbf{n}(\mathbf{y})$ whenever $|\partial_1 \mathbf{y} \times \partial_2 \mathbf{y}| \geq \varepsilon$. Likewise, we address the potential blowup of  $1/\sqrt{\det \mathbf{I}(\mathbf{y})}$ in $W_1$ of (\ref{eq:W12Def}) as follows. Noting that $\det \mathbf{I}(\mathbf{y}) = |\partial_1 \mathbf{y} \times \partial_2 \mathbf{y}|^2$,  we define 
\begin{equation}
\begin{aligned}
    j_{\varepsilon}(\mathbf{F}) = \begin{cases}
|\mathbf{F} \mathbf{e}_1 \times \mathbf{F} \mathbf{e}_1| & \text{ if } | \mathbf{F} \mathbf{e}_1 \times \mathbf{F} \mathbf{e}_2| \geq \varepsilon  \\
\varepsilon & \text{ otherwise}.
    \end{cases} 
\end{aligned}
\end{equation}
It follows that $1/j_{\varepsilon}(\mathbf{F})$ is continuous and  $\sqrt{\det \mathbf{I}(\mathbf{y})} = j_{\varepsilon}(\nabla \mathbf{y})$ whenever $|\partial_1 \mathbf{y} \times \partial_2 \mathbf{y}|\geq \varepsilon.$

We now define $W^{\varepsilon}$. For this purpose and throughout the remainder of this section, the variables $\mathbf{g} \in \mathbb{R}^2 , \mathbf{F} \in \mathbb{R}^{3\times2} , \mathcal{F} \in  \mathbb{R}^{3\times2 \times2}, \mathcal{F}_{\alpha \beta} \in \mathbb{R}^3$ denote placeholders for  $\nabla \theta, \nabla \mathbf{y}, \nabla \nabla \mathbf{y}$, $\partial_{\alpha} \partial_{\beta} \mathbf{y}$, respectively. Hence, we take $W^{\varepsilon} \colon \mathbb{R} \times \mathbb{R}^2 \times \mathbb{R}^{3\times2} \times \mathbb{R}^{3  \times 2 \times 2} \rightarrow \mathbb{R}$  to satisfy
\begin{equation}
\begin{aligned}\label{eq:WEpsDef}
W^{\varepsilon}(\theta, \mathbf{g} ,\mathbf{F}, \mathcal{F}) := W_1^{\varepsilon}( \theta, \mathbf{F}) + W_{2}^{\varepsilon}( \theta, \mathbf{F}, \mathcal{F}) + d_1 |\mathcal{F}|^2+   d_2 | \mathbf{g} |^2 + d_3 |\theta|^2
\end{aligned}
\end{equation} 
for energy densities $W_{1,2}^{\varepsilon}$ that are  continuous extensions of $W_{1,2}$ in (\ref{eq:W12Def}) and defined by  
\begin{equation}
\begin{aligned}\label{eq:W12Eps}
 &W_1^{\varepsilon}( \theta, \mathbf{F}) := \frac{c_1}{j_{\varepsilon}(\mathbf{F})} \big| \mathbf{F}^T \mathbf{F} - \mathbf{A}_{\varepsilon}^T(\theta) \mathbf{A}_{\varepsilon}(\theta) \big|^2,  \\
 &W_2^{\varepsilon}( \theta,  \mathbf{F}, \mathcal{F}) := \frac{c_2 L_{\Omega}^2}{|\tilde{\mathbf{u}}_0|^4|\tilde{\mathbf{v}}_0|^4}\Big(\big[ \mathbf{v}_{\varepsilon}(\theta) \cdot [\mathbf{v}']_{\varepsilon}(\theta)\big]  \big[\tilde{\mathbf{u}}_0 \cdot  \mathbf{K}_{\varepsilon}( \mathbf{F} , \mathcal{F}) \tilde{\mathbf{u}}_0\big]  + \big[\mathbf{u}_{\varepsilon}(\theta) \cdot [\mathbf{u}']_{\varepsilon}(\theta)\big] \big[ \tilde{\mathbf{v}}_0 \cdot  \mathbf{K}_{\varepsilon}(\mathbf{F}, \mathcal{F})  \tilde{\mathbf{v}}_0 \big] \Big)^2  ,
\end{aligned}
\end{equation}
where $\mathbf{K}_{\varepsilon} \colon \mathbb{R}^{3\times2} \times \mathbb{R}^{3\times2\times2}  \rightarrow \mathbb{R}^{2\times2}_{\text{sym}}$ is the  tensor field
\begin{equation}\label{eq:Keps}
\begin{aligned}
\mathbf{K}_{\varepsilon}(\mathbf{F}, \mathcal{F}) := \begin{pmatrix}  \mathcal{F}_{11} \cdot  \boldsymbol{\nu}_{\varepsilon}(\mathbf{F}) &  \mathcal{F}_{12} \cdot  \boldsymbol{\nu}_{\varepsilon}(\mathbf{F}) \\  \mathcal{F}_{12} \cdot  \boldsymbol{\nu}_{\varepsilon}(\mathbf{F}) &  \mathcal{F}_{22} \cdot  \boldsymbol{\nu}_{\varepsilon}(\mathbf{F}) \end{pmatrix}.
\end{aligned}
\end{equation}
As with the surface normal,  $\mathbf{K}_{\varepsilon}(\nabla \mathbf{y}, \nabla \nabla \mathbf{y}) = \mathbf{II}(\mathbf{y})$ whenever $|\partial_1 \mathbf{y} \times \partial_2 \mathbf{y}| \geq \varepsilon$; it is also linear in $\mathcal{F}$. This completes our construction of the desired  extended and continuous energy density $W^{\varepsilon}$ in (\ref{eq:getWepsilon}).

\noindent \textbf{Step 2.\;Coercivity, growth,  and convexity.} We turn to establishing several pointwise properties of $W^{\varepsilon}$, including upper and lower bounds and convexity in certain arguments. These properties  are key to applying the direct method of the calculus of variations to prove the existence of minimizers.   

Our first result concerns estimates on the energy density $W_{1}^{\varepsilon}$ in (\ref{eq:W12Eps}).
\begin{lemma}\label{firstLemmaBound}
There are constants $0 < c_{1,\varepsilon} < C_{1,\varepsilon}$ such that 
\begin{equation}
\begin{aligned}
c_{1,\varepsilon} | \mathbf{F}|^2 - \frac{1}{c_{1,\varepsilon}} \leq W^{\varepsilon}_1(\theta, \mathbf{F}) \leq C_{1,\varepsilon}( | \mathbf{F}|^4 +1)  \quad \text{ for all } \quad  (\theta, \mathbf{F}) \in \mathbb{R} \times \mathbb{R}^{3\times2}.
\end{aligned}
\end{equation}
\end{lemma}
\begin{proof}
For the upper bound, observe  by standard estimates that 
\begin{equation}
\begin{aligned}
W^{\varepsilon}_1(\theta, \mathbf{F}) = \frac{c_1}{j_{\varepsilon}(\mathbf{F})}  |\mathbf{F}^T \mathbf{F} + \mathbf{A}_{\varepsilon}^T(\theta) \mathbf{A}_{\varepsilon}(\theta)|^2 \leq 2c_1 \varepsilon^{-1}\big( |\mathbf{F}^T\mathbf{F}|^2 + |\mathbf{A}_{\varepsilon}^T(\theta) \mathbf{A}_{\varepsilon}(\theta)|^2 \big) \leq  2 c_1 \varepsilon^{-1} \big( |\mathbf{F}|^4 + |\mathbf{A}_{\varepsilon}(\theta)|^4).
\end{aligned}
\end{equation}
Let $M := \sup_{\theta \in (\theta^{-}, \theta^+)} |\mathbf{A}(\theta)|^4$ and note that $\sup_{\theta \in \mathbb{R}} |\mathbf{A}_{\varepsilon}(\theta)|^4 \leq M$ by the definition of $\mathbf{A}_{\varepsilon}(\theta)$ in (\ref{eq:Aeps}).  Moreover,  $M$ is finite by the definition of $\mathbf{A}(\theta)$ in (\ref{eq:shapeTensor}), since $|\mathbf{u}(\theta)| \leq |\mathbf{t}_1(\theta)| + |\mathbf{t}_3(\theta)| = |\mathbf{t}_1^r| + |\mathbf{t}_3^r|$ and $|\mathbf{v}(\theta)| \leq |\mathbf{t}_2(\theta)| + |\mathbf{t}_4(\theta)| = |\mathbf{t}_2^r| + |\mathbf{t}_4^r|$.  Thus, the upper bound asserted is true  for $C_{1,\varepsilon} = 2 c_1\varepsilon^{-1} \max \{ 1, M\}$. 

For the lowerbound, set $\mathbf{f}_i := \mathbf{F} \mathbf{e}_i$ and $\mathbf{a}_{i,{\varepsilon}}(\theta) := \mathbf{A}_{\varepsilon}(\theta)\mathbf{e}_i$ for $i = 1,2$. We have 
\begin{equation}
\begin{aligned}
W^{\varepsilon}_1( \theta, \mathbf{F}) &= \frac{c_1}{j_{\varepsilon}(\mathbf{F})}  \Big| \begin{pmatrix} |\mathbf{f}_1|^2 & \mathbf{f}_1 \cdot \mathbf{f}_2 \\ \mathbf{f}_1 \cdot \mathbf{f}_2 & |\mathbf{f}_2|^2 \end{pmatrix} - \begin{pmatrix} | \mathbf{a}_{1,\varepsilon}(\theta) |^2 & \mathbf{a}_{1,\varepsilon}(\theta) \cdot \mathbf{a}_{2,\varepsilon}(\theta) \\ \mathbf{a}_{1,\varepsilon}(\theta) \cdot \mathbf{a}_{2,\varepsilon}(\theta) & |\mathbf{a}_{2,\varepsilon}(\theta)|^2  \end{pmatrix} \Big|^2 \geq \frac{c_1}{j_{\varepsilon}(\mathbf{F})} \sum_{i =1,2}  \big| |\mathbf{f}_i|^2 - |\mathbf{a}_{i,\varepsilon}(\theta) |^2\big|^2.
\end{aligned}
\end{equation}
Using Young's inequality, we have  for any $\delta > 0$ that
\begin{equation}
\begin{aligned}
\sum_{i =1,2}  \big| |\mathbf{f}_i|^2 - |\mathbf{a}_{i,\varepsilon}(\theta) |^2\big|^2 &=  \sum_{i =1,2} \left(|\mathbf{f}_i|^4 - 2 |\mathbf{f}_i|^2 |\mathbf{a}_{i,\varepsilon}(\theta) |^2 + |\mathbf{a}_{i,\varepsilon}(\theta)|^4 \right)\geq  \sum_{i =1,2} \left( (1 - \delta) |\mathbf{f}_i|^4  + \big(1 - \frac1{\delta}\big) |\mathbf{a}_{i,\varepsilon}(\theta)|^4 \right).
\end{aligned}
\end{equation}
Taking $\delta = \frac12$ and employing the definition of $M$ gives
\begin{equation}
    \begin{aligned}
 \sum_{i =1,2}  \big| |\mathbf{f}_i|^2 - |\mathbf{a}_{i,\varepsilon}(\theta) |^2\big|^2\ge  \sum_{i =1,2}  \big( \frac12 |\mathbf{f}_i|^4 -  \frac12|\mathbf{a}_{i,\varepsilon}(\theta) |^4 \big)  \ge \frac{1}{4} |\mathbf{F}|^4 - \frac{1}{2} M. 
\end{aligned}
\end{equation}
It follows that 
\begin{equation}
    \begin{aligned}
        W_1^{\varepsilon}(\theta, \mathbf{F}) \geq \frac{c_1}{j_{\varepsilon}(\mathbf{F})}\big(\frac{1}{4} |\mathbf{F}|^4 - \frac{1}{2}M\Big) \geq  \frac{c_1}{4} \frac{|\mathbf{F}|^4}{|\mathbf{F} \mathbf{e}_1||\mathbf{F}\mathbf{e}_2|} - \frac{c_1}{2} \varepsilon^{-1} M  \geq \frac{c_1}{4} |\mathbf{F}|^2 - \frac{1}{2} c_1 \varepsilon^{-1} M.
    \end{aligned}
\end{equation}
Therefore, the desired lower bound holds with $c_{1,\varepsilon} := \min\{ \frac{c_1}{4}, \tfrac{2}{c_1 \varepsilon^{-1} M} \}$. 
\end{proof}

We now develop similar results for $W_2^{\varepsilon}$  in (\ref{eq:W12Eps}).
\begin{lemma}\label{secondLemmaBound}
There is a constant $C_{2,\varepsilon}> 0$ such that 
\begin{equation}
\begin{aligned}\label{eq:firstW2EpsBound}
0 \leq W_2^{\varepsilon}(\theta, \mathbf{F}, \mathcal{F}) \leq C_{2,\varepsilon} | \mathcal{F}|^2  \quad \text{ for all } (\theta, \mathbf{F}, \mathcal{F}) \in \mathbb{R} \times \mathbb{R}^{3\times2} \times \mathbb{R}^{3\times2 \times2} .
\end{aligned}
\end{equation}
\end{lemma}
\begin{proof}
We only need to prove the upperbound. By standard estimates,
\begin{equation}
\begin{aligned}
W_2^{\varepsilon}(\theta, \mathbf{F}, \mathcal{F})  \leq 2c_2 L_{\Omega}^2 ( \big[ \mathbf{v}_{\varepsilon}(\theta) \cdot [\mathbf{v}']_{\varepsilon}(\theta)\big]^2 +  \big[ \mathbf{u}_{\varepsilon}(\theta) \cdot [\mathbf{u}']_{\varepsilon}(\theta)\big]^2)  |\mathbf{K}_{\varepsilon}(\mathbf{F},\mathcal{F})|^2. 
\end{aligned}
\end{equation}
Let $M_{\varepsilon} := \sup_{\theta \in \mathbb{R}}   ( \big[ \mathbf{v}_{\varepsilon}(\theta) \cdot [\mathbf{v}']_{\varepsilon}(\theta)\big]^2 +  \big[ \mathbf{u}_{\varepsilon}(\theta) \cdot [\mathbf{u}']_{\varepsilon}(\theta)\big]^2)$. This quantity is clearly finite by the smoothness of $\mathbf{u}(\theta), \mathbf{v}(\theta)$ on $(\theta^{-}, \theta^+)$ and by the definitions of their extensions $\mathbf{u}_{\varepsilon}(\theta), \mathbf{v}_{\varepsilon}(\theta), [\mathbf{u}']_{\varepsilon}(\theta)$ and $[\mathbf{v}']_{\varepsilon}(\theta)$ in (\ref{eq:uEpsilonDef}) and (\ref{eq:uEpsilonPrimeDef}). Standard estimates on  $\mathbf{K}_{\varepsilon}(\mathbf{F}, \mathcal{F})$ in (\ref{eq:Keps}) then yield
\begin{equation}
\begin{aligned}
|\mathbf{K}_{\varepsilon}(\mathbf{F}, \mathcal{F})|^2 \leq (|\mathcal{F}_{11}|^2 + 2|\mathcal{F}_{12}|^2 + |\mathcal{F}_{22}|^2) |\boldsymbol{\nu}_{\varepsilon}(\mathbf{F})|^2  \leq |\mathcal{F}|^2
\end{aligned}
\end{equation}
where  the last inequality uses that $|\mathcal{F}_{11}|^2 + 2|\mathcal{F}_{12}|^2 + |\mathcal{F}_{22}|^2= |\mathcal{F}|^2$ and that  $|\boldsymbol{\nu}_{\varepsilon}(\mathbf{F})| = 1$ when $|\mathbf{F} \mathbf{e}_1 \times \mathbf{F} \mathbf{e}_2| \geq \varepsilon$ and $< 1$ otherwise. We conclude that $W_2^{\varepsilon}(\theta, \mathbf{F}, \mathcal{F})  \leq 2c_2 L_{\Omega}^2 M_{\varepsilon} |\mathcal{F}|^2$ and thus $C_{2,\varepsilon} := 2c_2 L_{\Omega}^2 M_{\varepsilon}$. 
\end{proof}

We state one final result highlighting the key properties of $W^{\varepsilon}$  that go into proving Theorem \ref{ExistenceTheorem}. 
\begin{proposition}\label{WepsProp}
There are positive constants $c_\varepsilon, C_{\varepsilon}> 0$ such that 
\begin{equation}
\begin{aligned}\label{eq:growthCoercivity}
c_{\varepsilon}\Big(  |\theta|^2 +   | \mathbf{g} |^2 +  | \mathbf{F}|^2 +  |\mathcal{F}|^2 \Big)  - \frac{1}{c_{\varepsilon}}    \leq  W^{\varepsilon}(\theta, \mathbf{g} ,\mathbf{F}, \mathcal{F})  \leq  C_{\varepsilon}   (  |\theta|^2  + |\mathbf{g}|^2 + |\mathbf{F}|^4 + |\mathcal{F}|^2 + 1) 
\end{aligned}
\end{equation}
for all $(\theta, \mathbf{g}, \mathbf{F}, \mathcal{F}) \in \mathbb{R} \times \mathbb{R}^2 \times \mathbb{R}^{3\times2}  \times \mathbb{R}^{3 \times 2 \times 2}$. In addition,  $W^{\varepsilon} (\theta, \mathbf{g} , \mathbf{F} , \mathcal{F})$ is continuous in all its arguments and convex in $(\mathbf{g}, \mathcal{F})$ for each $(\theta, \mathbf{F})$. 
\end{proposition}
\begin{proof}
For the lowerbound, we have by (\ref{eq:WEpsDef}) and Lemma's \ref{firstLemmaBound}-\ref{secondLemmaBound} that 
\begin{equation}
\begin{aligned}
W^{\varepsilon}(\theta, \mathbf{g}, \mathbf{F}, \mathcal{F}) \geq c_{1,\varepsilon} | \mathbf{F}|^2 - \frac{1}{c_{1,\varepsilon}}  + d_1 |\mathcal{F}|^2  +  d_2 | \mathbf{g} |^2 + d_3 |\theta|^2.
\end{aligned} 
\end{equation}
The choice $c_\varepsilon := \min\{ c_{1,\varepsilon} , d_1 L_{\Omega}^2, d_2 L_{\Omega}^2, d_3\}$ gives the desired result.  For the upperbound, we again have by (\ref{eq:WEpsDef}) and Lemma's \ref{firstLemmaBound}-\ref{secondLemmaBound} that 
\begin{equation}
\begin{aligned}
W^{\varepsilon}(\theta, \mathbf{g}, \mathbf{F}, \mathcal{F})  &\leq C_{1,\varepsilon}( | \mathbf{F}|^4 +1) +   d_2 L_{\Omega}^2 | \mathbf{g} |^2 + (d_1 L_{\Omega}^2 + C_{2,\varepsilon})  |\mathcal{F}|^2 + d_3 |\theta|^2 \\
&\leq  (C_{1,\varepsilon}+ d_1L_{\Omega}^2  + d_2 L_{\Omega}^2 + d_3 + C_{2,\varepsilon}) (  |\theta|^2  + |\mathbf{g}|^2 + |\mathbf{F}|^4 + |\mathcal{F}|^2 + 1) .
\end{aligned}
\end{equation}
Thus, $C_{\varepsilon} := (C_{1,\varepsilon}+ d_1L_{\Omega}^2  + d_2 L_{\Omega}^2 + d_3 + C_{2,\varepsilon})$ gives  the desired upperbound. 

As for the convexity statement, it is clear that $d_1 L_{\Omega}^2|\mathcal{F}|^2 + d_2 L_{\Omega}^2 |\mathbf{g}|^2$ is convex and quadratic in $(\mathbf{g}, \mathcal{F})$, so the only questions concern $W_2^{\varepsilon}(\theta, \mathbf{F}, \mathcal{F})$ in (\ref{eq:W12Eps}). However,  $W_2^{\varepsilon}$ satisfies $W_2^{\varepsilon}(\theta, \mathbf{F}, \mathcal{F}) = c_2 L_{\Omega}^2 ( \boldsymbol{\nu}_{\varepsilon}(\mathbf{F}) \cdot \mathcal{F} \colon \mathbf{C}_{\varepsilon}(\theta))^2$ for all $(\theta, \mathbf{F}, \mathcal{F}) \in \mathbb{R} \times \mathbb{R}^{2\times2} \times  \mathbb{R}^{3 \times 2 \times 2}$, where $\mathbf{C}_{\varepsilon}(\theta) := \big( \mathbf{v}_{\varepsilon}(\theta) \cdot [\mathbf{v}']_{\varepsilon}(\theta)\big) \tilde{\mathbf{u}}_0 \otimes \tilde{\mathbf{u}}_0 +  \big( \mathbf{u}_{\varepsilon}(\theta) \cdot [\mathbf{u}']_{\varepsilon}(\theta)\big) \tilde{\mathbf{v}}_0 \otimes \tilde{\mathbf{v}}_0$. This function is clearly convex in $\mathcal{F}$.  This completes the proof. 
\end{proof}

\noindent \textbf{Step 3. Energy estimates.} We now go from pointwise estimates on the various energy densities that make up $E_{\text{int}}(\mathbf{y}, \theta)$ to upper and lower bounds on the overall energy $E(\mathbf{y}, \theta)$.  For the coercivity result, we follow a presentation by Ball (\cite{ball1976convexity}, Theorem 7.6) for mixed boundary conditions. 

\begin{proposition}\label{EnergyProp}
Assume that $\Gamma_{\emph{n}} \subset \Gamma_{\emph{d}}$ are measurable,  $\Gamma_{\emph{d}}$ has nonzero measure, and that $V^{\varepsilon}_{\Gamma}$ is non-empty. Let $\overline{\mathbf{b}} \in L^2(\Omega, \mathbb{R}^3)$, $\overline{\mathbf{t}} \in L^2(\Gamma_{\emph{t}}, \mathbb{R}^3)$, $\overline{\mathbf{m}} \in L^2(\Gamma_{\emph{m}}, \mathbb{R}^{3})$. Then, for all $(\mathbf{y}, \theta) \in V^{\varepsilon}_{\Gamma}$, the following two statments hold: (1) $E(\mathbf{y} ,\theta) < +\infty$ and (2) there is a constant $c_{\varepsilon}^{\star}>0$ depending only on $c_\varepsilon$ from (\ref{eq:growthCoercivity}) and $\Omega$  such that 
\begin{equation}
\begin{aligned}\label{eq:keyCoercivity}
E(\mathbf{y} ,\theta) \geq c_{\varepsilon}^{\star}\big( \| \theta \|^2_{H^1(\Omega)} + \| \mathbf{y} \|^2_{H^2(\Omega)} \big)  -  \frac{1}{c_{\varepsilon}^{\star}} \Big(  \| \overline{\mathbf{b}}\|^2_{L^{2}(\Omega)}  +  \| \overline{\mathbf{t}}\|^2_{L^{2}(\Gamma_{\text{t}})} + \| \overline{\mathbf{m}} \|^2_{L^2(\Gamma_{\text{m}})}  + \| \overline{\mathbf{y}} \|^2_{L^2(\Gamma_{\text{d}})} + 1 \Big).
\end{aligned}
\end{equation}
\end{proposition}
\begin{remark}\label{EnergyRemark}
Since $V^{\varepsilon}_{\Gamma}$ is assumed to be non-empty above, $\| \overline{\mathbf{y}} \|_{L^2(\Gamma_{\emph{d}})} < \infty$ by the trace theorem for Sobolev maps.  Thus, from the assumptions of Theorem \ref{ExistenceTheorem} (the same as those in Proposition \ref{EnergyProp}) and the inequality in (\ref{eq:keyCoercivity}), the infimum $E^{\star} := \inf_{(\mathbf{y},\theta) \in V^{\varepsilon}_{\Gamma}} E(\mathbf{y}, \theta)$ is not $-\infty$. 
\end{remark}
\begin{proof}
Fix any $(\mathbf{y}, \theta) \in V^{\varepsilon}_{\Gamma}$ (which is non-empty by assumption) and note that $\theta \in H^1(\Omega)$ and $\mathbf{y} \in H^2(\Omega)$. 

We first prove that $E(\mathbf{y}, \theta) < +\infty$. The Cauchy-Schwarz inequality and the trace theorem provide an upper bound on $E_{\text{ext}}(\mathbf{y})$ in (\ref{eq:Eext}) of the form 
\begin{equation}
\begin{aligned}\label{eq:bcEsts}
E_{\text{ext}}(\mathbf{y}) &\leq \| \mathbf{y} \|_{L^2(\Omega)} \| \overline{\mathbf{b}}\|_{L^{2}(\Omega)}  + \| T \mathbf{y}\|_{L^2(\Gamma_{\text{t}})}  \| \overline{\mathbf{t}} \|_{L^2(\Gamma_{\text{t}})}  +  \| T \nabla \mathbf{y}\|_{L^2(\Gamma_{\text{m}})}  \| \overline{\mathbf{m}} \|_{L^2(\Gamma_{\text{m}})}  \\
&\leq  C_0 \| \mathbf{y}\|_{H^2(\Omega)} \left( \| \overline{\mathbf{t}} \|_{L^2(\Gamma_{\text{t}})} + \| \overline{\mathbf{m}} \|_{L^2(\Gamma_{\text{m}})} + \| \overline{\mathbf{b}}\|_{L^{2}(\Omega)} \right) 
\end{aligned}
\end{equation}
for some $C_0>0$  that only depends on $\Omega$, where $T \colon H^1(\Omega) \rightarrow L^2(\partial \Omega)$ in the first inequality is the trace operator mapping an $H^1$ function on $\Omega$ to its ``trace", an $L^2$ function on the boundary  $\partial \Omega$ (see, for instance, \cite{evans2022partial}). We conclude that $E_{\text{ext}}(\mathbf{y}) < \infty$ from the assumptions on the boundary fields $\overline{\mathbf{b}}$, $\overline{\mathbf{t}}$, and $\overline{\mathbf{m}}$ in the proposition and since $\mathbf{y} \in H^2(\Omega)$. Next, observe that, since $V^{\varepsilon}_{\Gamma} \subset V^{\varepsilon}$,  
\begin{equation}
\begin{aligned}\label{eq:EintUpperBound}
E_{\text{int}}(\mathbf{y}, \theta) = \int_{\Omega} W^{\varepsilon}(\theta, \nabla \theta, \nabla \mathbf{y} , \nabla \nabla \mathbf{y} ) \dif x \leq C_{\varepsilon} \Big(  \| \theta \|_{H^1(\Omega)}^2  + \| \nabla \mathbf{y}\|_{L^4(\Omega)}^4 + \| \nabla \nabla \mathbf{y} \|_{L^2(\Omega)}^2  + |\Omega| \Big) 
\end{aligned}
\end{equation}
by the definition in (\ref{eq:getWepsilon}) and  the upperbound in (\ref{eq:growthCoercivity}) of Proposition \ref{WepsProp}.  Since $\Omega \subset \mathbb{R}^2$ and $\mathbf{y} \in H^2(\Omega)$, the Sobolev embedding theorem implies that $\nabla \mathbf{y} \in L^4(\Omega)$. So it follows from (\ref{eq:EintUpperBound}) that $E_{\text{int}}(\mathbf{y}, \theta) < + \infty$. Thus, $E(\mathbf{y}, \theta) := E_{\text{int}}(\mathbf{y}, \theta) + E_{\text{ext}}(\mathbf{y}, \theta) < + \infty$, as desired.

We now prove the lowerbound in (\ref{eq:keyCoercivity}).  Since $V^{\varepsilon}_{\Gamma} \subset V^{\varepsilon}$, we have by the definition of $W^{\varepsilon}$ in (\ref{eq:getWepsilon}) and the lowerbound in  (\ref{eq:growthCoercivity}) of Proposition \ref{WepsProp} that 
\begin{equation}
\begin{aligned}\label{eq:ineqDirect1}
E(\mathbf{y}, \theta)  \geq c_{\varepsilon} \int_{\Omega}\Big( |\theta|^2 + |\nabla \theta|^2 + |\nabla \mathbf{y}|^2 + |\nabla \nabla \mathbf{y}|^2 \Big) \dif x- c_\varepsilon^{-1} |\Omega| + E_{\text{ext}}(\mathbf{y}).
\end{aligned}
\end{equation} 
Focusing on $E_{\text{ext}}(\mathbf{y})$, applying a Cauchy--Schwarz inequality, a Young's inequality for $\delta > 0$, and a trace inequality leads to 
\begin{equation} 
\begin{aligned}
E_{\text{ext}}(\mathbf{y} ,\theta)  &\geq - \delta \| \mathbf{y} \|_{L^2(\Omega)} \delta^{-1} \| \overline{\mathbf{b}}\|_{L^{2}(\Omega)}  - \delta \| T \mathbf{y}\|_{L^2(\Gamma_{\text{t}})}  \delta^{-1} \| \overline{\mathbf{t}}\|_{L^{2}(\Gamma_{\text{t}})} -  \delta \| T \nabla \mathbf{y}\|_{L^2(\Gamma_{\text{m}})}  \delta^{-1} \| \overline{\mathbf{m}} \|_{L^2(\Gamma_{\text{m}})} \\
&\geq - \frac{\delta^2}{2}  \Big( \| \mathbf{y} \|_{L^2(\Omega)}^2 + \| T \mathbf{y} \|_{L^2(\Gamma_{\text{t}})}^2 + \| T \nabla \mathbf{y}\|^2_{L^2(\Gamma_{\text{t}})}\Big)  - \frac{1}{2\delta^2} \Big( \| \overline{\mathbf{b}}\|^2_{L^{2}(\Omega)}  +  \| \overline{\mathbf{t}}\|^2_{L^{2}(\Gamma_{\text{t}})} + \| \overline{\mathbf{m}} \|^2_{L^2(\Gamma_{\text{m}})} \Big) \\
&\geq - \delta^2 C_{1}  \| \mathbf{y} \|_{H^2(\Omega)}^2 -  \frac{1}{\delta^2} \Big( \| \overline{\mathbf{b}}\|^2_{L^{2}(\Omega)}  +  \| \overline{\mathbf{t}}\|^2_{L^{2}(\Gamma_{\text{t}})} + \| \overline{\mathbf{m}} \|^2_{L^2(\Gamma_{\text{m}})} \Big)
\end{aligned}
\end{equation}
for some constant $C_{1} >0$ that depends only on $\Omega$. 
Furthermore, since $\Gamma_{\text{d}}$ has positive measure, a Poincar\'{e} type estimate (from \cite{morrey2009multiple}, page 82) gives 
\begin{equation}
\begin{aligned}\label{eq:ineqDirect2}
\int_{\Omega} |\mathbf{y}|^2 \dif x\leq C_{2} \Big( \int_{\Omega} |\nabla \mathbf{y}|^2 + \int_{\Gamma_{\text{d}}}|\overline{\mathbf{y}}|^2 ds  \Big) 
\end{aligned}
\end{equation}
for some constant $C_{2} >0$ that depends only on $\Omega$, where $\|\overline{\mathbf{y}} \|_{L^2(\Gamma_{\text{d}})}^2 < \infty$ since $\overline{\mathbf{y}}  \in H^{1/2}(\Gamma_\text{d}, \mathbb{R}^3)$.
Combining the inequalities in (\ref{eq:ineqDirect1}-\ref{eq:ineqDirect2}) gives 
\begin{equation}
\begin{aligned}\label{eq:almostDoneIneq}
E(\mathbf{y} ,\theta)  &\geq  c_\varepsilon \int_{\Omega}\Big( |\mathbf{\theta}|^2 +  |\nabla \theta|^2 + \lambda |\nabla \mathbf{y}|^2 + |\nabla \nabla \mathbf{y}|^2 \Big) \dif x+ c_\varepsilon (1-\lambda) C_{2}^{-1}  \int_{\Omega} | \mathbf{y}|^2 \dif x   -  c_\varepsilon^{-1}  |\Omega| \\
& \qquad   - \delta^2 C_{1}  \| \mathbf{y} \|_{H^2(\Omega)}^2 -  \frac{1}{\delta^2} \Big( \| \overline{\mathbf{b}}\|^2_{L^{2}(\Omega)}  +  \| \overline{\mathbf{t}}\|^2_{L^{2}(\Gamma_{\text{t}})} + \| \overline{\mathbf{m}} \|^2_{L^2(\Gamma_{\text{m}})} \Big) -  c_\varepsilon (1-\lambda) \int_{\Gamma_{\text{d}}} |\overline{\mathbf{y}}|^2 \dif s 
\end{aligned}
\end{equation}
for any choice of  $\lambda \in (0,1)$ and $\delta >0$. By choosing such $\lambda$ and $\delta$ appropriately, we  deduce from (\ref{eq:almostDoneIneq}) that there is a constant $c_{\varepsilon}^{\star} >0$  that depends only on $c_{\varepsilon}, C_1, C_2,$ and $\Omega$ (so $c_{\varepsilon}$ and $\Omega$) such that (\ref{eq:keyCoercivity}) holds. 
\end{proof}

\noindent \textbf{Step 4. Proof of Theorem \ref{ExistenceTheorem}.} As the assumptions of the theorem are the same as that of Proposition \ref{EnergyProp}, this proposition and Remark \ref{EnergyRemark} establish that $E^{\star} := \inf_{(\mathbf{y}, \theta) \in V_{\Gamma}^{\varepsilon}}  E(\mathbf{y}, \theta)$ is finite. Hence, there exists a minimizing sequence $\{ (\mathbf{y}_j, \theta_j) \} \subset V^{\varepsilon}_{\Gamma}$ such that $E(\mathbf{y}_j, \theta_j) \leq E^{\star} + \frac{1}{j}$ for all positive integers $j$. Moreover, the lowerbound in (\ref{eq:keyCoercivity}) of Proposition \ref{EnergyProp} furnishes the estimates on this sequence 
\begin{equation}
\begin{aligned}
 \| \theta_j \|^2_{H^1(\Omega)} + \| \mathbf{y}_j \|^2_{H^2(\Omega)} \leq  \frac{1}{c_{\varepsilon}^{\star}} (E^{\star} + 1 ) +   \frac{1}{(c_{\varepsilon}^{\star})^2} \Big(  \| \overline{\mathbf{b}}\|^2_{L^{2}(\Omega)}  +  \| \overline{\mathbf{t}}\|^2_{L^{2}(\Gamma_{\text{t}})} + \| \overline{\mathbf{m}} \|^2_{L^2(\Gamma_{\text{m}})}  + \| \overline{\mathbf{y}} \|^2_{L^2(\Gamma_{\text{d}})} + 1 \Big)
\end{aligned}
\end{equation}
for a $c_{\varepsilon}^{\star} >0$ independent of $j$. As the bound on the right is finite (by the assumptions of the theorem) and independent of $j$,  the sequence  $\{ (\mathbf{y}_j, \theta_{j})\}$ is uniformly bounded in the $H^2 \times H^1$ norm. It follows that  there is a subsequence $\{ (\mathbf{y}_{j_k}, \theta_{j_k})\}$ that converges weakly in $H^2\times H^1$, i.e.,  
\begin{equation}
\begin{aligned}\label{eq:weakConvergence}
\mathbf{y}_{j_k} \rightharpoonup \mathbf{y} \text{ in } H^2(\Omega, \mathbb{R}^3) \quad \text{ and } \quad \theta_{j_k} \rightharpoonup \theta \text{ in } H^1(\Omega, \mathbb{R}).
\end{aligned}
\end{equation}

We now show that $(\mathbf{y}, \theta)$ belongs to $V_{\Gamma}^{\varepsilon}$, first by verifying that these fields satisfy the $\varepsilon$-inequalities contained in this set and then by verifying the boundary conditions. By Rellich's theorem, there is subsequence (not relabeled) that satisfies $\nabla \mathbf{y}_{j_k} \rightarrow \nabla \mathbf{y}$ in $L^2(\Omega)$ and $\theta_{j_k} \rightarrow \theta$ in $L^2(\Omega)$.  By this strong $L^2$-convergence, there is an even further subsequence (not relabeled) such that $\nabla \mathbf{y}_{j_k}(\mathbf{x}) \rightarrow \nabla \mathbf{y}(\mathbf{x})$ and $\theta_{j_k}(\mathbf{x}) \rightarrow \theta(\mathbf{x})$ for a.e.\;$\mathbf{x} \in \Omega$. Since $\theta_{j_k}(\mathbf{x}) \geq \theta^{-} + \varepsilon$ for a.e. $\mathbf{x} \in \Omega$ for all $j_k$, 
\begin{equation}
\begin{aligned}
\theta(\mathbf{x}) = \lim_{k \rightarrow \infty} \theta_{j_k}(\mathbf{x}) \geq \theta^{-} + \varepsilon 
\end{aligned}
\end{equation}
for a.e.\;$\mathbf{x} \in \Omega$. Arguing similarly, we conclude that $\theta(\mathbf{x}) \leq \theta^{+} - \varepsilon$ for a.e.\;$\mathbf{x} \in \Omega$. Finally, since the function $f(\mathbf{F}) := |\mathbf{F}\mathbf{e}_1 \times \mathbf{F} \mathbf{e}_2|$ is continuous  and since $f(\nabla \mathbf{y}_{j_k}(\mathbf{x}) ) = |\partial_1  \mathbf{y}_{j_k}(\mathbf{x}) \times \partial_2  \mathbf{y}_{j_k}(\mathbf{x})| \geq \varepsilon$ for a.e.\;$\mathbf{x} \in \Omega$ for all $j_k$, we conclude that 
\begin{equation}
\begin{aligned}
|\partial_1 \mathbf{y}(\mathbf{x}) \times \partial_2 \mathbf{y}(\mathbf{x})| = f\Big( \lim_{k \rightarrow \infty} \nabla \mathbf{y}_{j_k}(\mathbf{x})\Big) = \lim_{k \rightarrow \infty} f(\nabla \mathbf{y}_{j_k}(\mathbf{x}))  \geq \varepsilon
\end{aligned}
\end{equation}
for a.e.\;$\mathbf{x} \in \Omega$. We now verify the boundary conditions. Since the trace operator $T \colon H^1(\Omega) \rightarrow L^2(\partial \Omega)$ is bounded and linear, it is continuous and it is follows from (\ref{eq:weakConvergence}) that  
\begin{equation}
\begin{aligned}\label{eq:weakTrace}
T \mathbf{y}_{j_k} \rightharpoonup T\mathbf{y}  \text{ in $L^2(\partial \Omega)$} \quad  \text{ and } \quad T\nabla \mathbf{y}_{j_k} \rightharpoonup T \nabla \mathbf{y} \text{ in $L^2(\partial \Omega)$}.
\end{aligned}
\end{equation}
Since $T \mathbf{y}_{j_k} = \overline{\mathbf{y}}$ a.e.\;on $\Gamma_{\text{d}}$ and $T \nabla \mathbf{y}_{j_k} \mathbf{n} = \overline{\mathbf{s}}$ a.e.\;on $\Gamma_{\text{n}}$, we conclude from (\ref{eq:weakTrace}) that $T\mathbf{y} = \overline{\mathbf{y}}$ a.e.\;on $\Gamma_{\text{d}}$ and $T \nabla \mathbf{y}\mathbf{n} = \bar{\mathbf{s}}$ a.e.\;on $\Gamma_{\text{n}}$. Hence, $(\mathbf{y}, \theta) \in V_{\Gamma}^{\varepsilon}$. 

We are now ready to complete the proof. The properties  established by Proposition \ref{WepsProp} place us in a setting where we can directly apply a general result on weakly lower-semicontinuous functionals (Dacorogna \cite{dacorogna2007direct}, Theorem 3.23). In particular, we conclude from (\ref{eq:weakConvergence}) that 
\begin{equation}
\begin{aligned}\label{eq:liminf1}
\liminf_{k \rightarrow \infty} E_{\text{int}}(\mathbf{y}_{j_k}, \theta_{j_k}) &= \liminf_{k \rightarrow \infty} \int_{\Omega} W^{\varepsilon} (\theta_{j_k}, \nabla \theta_{j_k}, \nabla \mathbf{y}_{j_k}, \nabla \nabla \mathbf{y}_{j_k} ) \dif x\\
& \geq \int_{\Omega} W^{\varepsilon} (\theta,  \nabla \theta, \nabla \mathbf{y}, \nabla \nabla \mathbf{y} ) \dif x= E_{\text{int}}(\mathbf{y}, \theta),
\end{aligned}
\end{equation}
where the last equality follows since $(\mathbf{y}, \theta) \in V_{\Gamma}^{\varepsilon}$.
The basic reasoning behind this inequality  is that $W^{\varepsilon}$ is continuous in the lower order terms $(\theta, \nabla \mathbf{y})$ and convex in these terms gradients $ (\nabla \theta,\nabla \nabla \mathbf{y})$.
Next, we note that  the potential $E_{\text{ext}}(\mathbf{y}_{j_k})$ associated to applied forces converges trivially to $E_{\text{ext}}(\mathbf{y})$ given (\ref{eq:weakConvergence}) and (\ref{eq:weakTrace})  since the terms  $\mathbf{y}_{j_k}$, $T \mathbf{y}_{j_k}$ and $T\nabla \mathbf{y}_{j_k}$ are paired with functions $\overline{\mathbf{b}}$, $\overline{\mathbf{t}}$ and $\overline{\mathbf{m}}$ in their dual space, i.e., 
\begin{equation}
\begin{aligned}\label{eq:liminf2}
\lim_{k \rightarrow \infty} E_{\text{ext}}(\mathbf{y}_{j_k}) = E_{\text{ext}}(\mathbf{y}).  
\end{aligned}
\end{equation}
Recalling that $\{ (\mathbf{y}_{j_k}, \theta_{j_k})\}$ is a minimizing sequence that satisfies  $E(\mathbf{y}_{j_k}, \theta_{j_k}) \leq E^{\star} + \frac{1}{j_k}$ for positive and increasing  $j_k$ that  $\rightarrow \infty$ as $k \rightarrow \infty$, we obtain from (\ref{eq:liminf1}) and (\ref{eq:liminf2}) that 
\begin{equation}
\begin{aligned}
E^{\star} \geq \liminf_{k \rightarrow \infty} E(\mathbf{y}_{j_k}, \theta_{j_k}) = \liminf_{k \rightarrow \infty} \big( E_{\text{int}}(\mathbf{y}_{j_k}, \theta_{j_k}) + E_{\text{ext}}(\mathbf{y}_{j_k})\big) \geq  E_{\text{int}}(\mathbf{y}, \theta)  + E_{\text{ext}}(\mathbf{y}) = E(\mathbf{y}, \theta).
\end{aligned}
\end{equation}
Since $E^{\star} = \inf_{V_{\Gamma}^{\varepsilon}}E$ and  $(\mathbf{y},\theta) \in  V_{\Gamma}^{\varepsilon}$, the inequalities above are actually equalities, i.e., $E^{\star} = E(\mathbf{y}, \theta)$, and thus  $(\mathbf{y}, \theta)$ is a minimizer. This completes the proof of the existence of minimizers.

\section{On the derivation of the governing equations}\label{sec:DeriveStrongForm}
Here, we derive a key identity used to obtain the strong form of the equilibrium equations and natural boundary conditions from weak form, namely, that 
\begin{equation}
\begin{aligned}\label{eq:importForStrongForm}
&\int_{\Omega} \Big\{  \mathbf{P}(\mathbf{y}, \theta) \colon  \nabla \mathbf{w}   +  \big \langle \mathcal{H}(\mathbf{y}, \theta) ,  \nabla \nabla \mathbf{w} \big \rangle  + q(\mathbf{y}, \theta)  \eta  + \mathbf{j}(\theta) \cdot \nabla \eta  \Big\} \dif x \\
&\quad = \int_{\Omega}   \big(\Div \big[ \Div \mathcal{H}(\mathbf{y}, \theta)  - \mathbf{P} (\mathbf{y}, \theta)   \big] \big)\cdot \mathbf{w}  \dif x +  \int_{\partial \Omega}  \big[\mathbf{P} (\mathbf{y}, \theta)  - \nabla  \mathcal{H}(\mathbf{y}, \theta) \colon (\mathbf{I} + \mathbf{n}^{\perp} \otimes \mathbf{n}^{\perp}) \big]  \mathbf{n} \cdot \mathbf{w} \dif s   \\
&\qquad \quad + \int_{\partial \Omega}  \big(\mathcal{H}(\mathbf{y}, \theta)  \colon  ( \mathbf{n} \otimes \mathbf{n}) \big) \cdot (\nabla \mathbf{w}) \mathbf{n} \dif s  + \int_{\Omega} \big(  q(\mathbf{y}, \theta)  - \nabla \cdot \mathbf{j}(\theta) \big) \eta  \dif x +  \int_{\partial \Omega} \big( \mathbf{j} (\theta)\cdot \mathbf{n} \big)  \eta \dif s.
\end{aligned}
\end{equation}
holds  for all sufficiently smooth $\mathbf{y}, \mathbf{w} \colon \Omega \rightarrow \mathbb{R}^3$ and $\theta, \eta \colon \Omega \rightarrow (\theta^{-}, \theta^+)$, where $\mathbf{n}$ is the outward normal to $\partial \Omega$ and $\mathbf{n}^{\perp}$ is the corresponding unit tangent vector. (See  also below (\ref{eq:manipWeakForm}) for any questions about the tensor notation.) 

To begin, use integration by parts and the divergence theorem to the term on the left in (\ref{eq:importForStrongForm}) to obtain  
\begin{equation}
\begin{aligned}\label{eq:importForStrongForm1}
&\int_{\Omega} \Big\{  \mathbf{P}(\mathbf{y}, \theta) \colon  \nabla \mathbf{w}   +  \big \langle \mathcal{H}(\mathbf{y}, \theta) ,  \nabla \nabla \mathbf{w} \big \rangle  + q(\mathbf{y}, \theta)  \eta  + \mathbf{j}(\theta) \cdot \nabla \eta  \Big\} \dif x \\
&\qquad = \int_{\Omega}   \big(\Div \big[ \Div \mathcal{H}(\mathbf{y}, \theta)  - \mathbf{P} (\mathbf{y}, \theta)   \big] \big)\cdot \mathbf{w}  \dif x +  \int_{\partial \Omega}  \big[\mathbf{P} (\mathbf{y}, \theta)  - \Div  \mathcal{H}(\mathbf{y}, \theta) \big]  \mathbf{n} \cdot \mathbf{w} \dif s   \\
&\qquad \quad + \int_{\partial \Omega}  \big(\mathcal{H}(\mathbf{y}, \theta)  \cdot \mathbf{n} \big) \colon \nabla \mathbf{w} \dif s  + \int_{\Omega} \big(  q(\mathbf{y}, \theta)  - \nabla \cdot \mathbf{j}(\theta) \big) \eta  \dif x +  \int_{\partial \Omega} \big( \mathbf{j} (\theta)\cdot \mathbf{n} \big)  \eta \dif s
\end{aligned}
\end{equation}
where $[\mathcal{H} \cdot \mathbf{n}]_{i \alpha} = [\mathcal{H}]_{i\alpha \beta} n_\beta$. Next, write $\nabla \mathbf{w}  = \partial_{\mathbf{n}} \mathbf{w} \otimes \mathbf{n} + \partial_{\mathbf{n}^{\perp}} \mathbf{w} \otimes \mathbf{n}^{\perp}$ and observe that 
\begin{equation}
\begin{aligned}\label{eq:NablaWTerm}
\int_{\partial \Omega}  \big(\mathcal{H}(\mathbf{y}, \theta)  \cdot \mathbf{n} \big) \colon \nabla \mathbf{w} \dif s  = \int_{\partial \Omega} \Big\{   \big(\mathcal{H}(\mathbf{y}, \theta) \colon (\mathbf{n} \otimes \mathbf{n}) \big) \cdot \partial_{\mathbf{n}} \mathbf{w}   + \big(\mathcal{H}(\mathbf{y}, \theta) \colon (\mathbf{n} \otimes \mathbf{n}^{\perp}) \big) \cdot \partial_{\mathbf{n}^{\perp}} \mathbf{w} \Big\} \dif s.
\end{aligned}
\end{equation}
Another round of integration by parts on the last term in (\ref{eq:NablaWTerm}) furnishes 
\begin{equation}
\begin{aligned}
\int_{\partial \Omega} \big(\mathcal{H}(\mathbf{y}, \theta) \colon (\mathbf{n} \otimes \mathbf{n}^{\perp}) \big) \cdot \partial_{\mathbf{n}^{\perp}} \mathbf{w} \dif s =  -  \int_{\partial \Omega} \partial_{\mathbf{n}^{\perp}}  \big(\mathcal{H}(\mathbf{y}, \theta) \colon (\mathbf{n} \otimes \mathbf{n}^{\perp}) \big) \cdot  \mathbf{w}  \dif s.
\end{aligned}
\end{equation}
In particular, $ \int_{\partial \Omega} \partial_{\mathbf{n}^{\perp}} \big[  \big(\mathcal{H}(\mathbf{y}, \theta) \colon (\mathbf{n} \otimes \mathbf{n}^{\perp}) \big) \cdot  \mathbf{w} \big] \dif s = 0$ because the tangential derivative of a field integrated on a closed boundary always vanishes.  To finish, note that  $\partial_{\mathbf{n}^{\perp}} \mathbf{n} = \kappa \mathbf{n}^{\perp}$, where $\kappa$ is  the curvature of $\partial \Omega$.  Likewise,  $\partial_{\mathbf{n}^{\perp}} \mathbf{n}^{\perp} = - \kappa \mathbf{n}$. Thus, 
\begin{equation}
\begin{aligned}\label{eq:getWeirdNaturalBC}
\partial_{\mathbf{n}^{\perp}} \big[ \mathcal{H}(\mathbf{y}, \theta) \colon (\mathbf{n} \otimes \mathbf{n}^{\perp} )\big]  &= \partial_{\mathbf{n}^{\perp}} \big[ \mathcal{H}(\mathbf{y}, \theta) \big] \colon (\mathbf{n} \otimes \mathbf{n}^{\perp} ) + \kappa  \mathcal{H}(\mathbf{y}, \theta) \colon  ( \mathbf{n}^{\perp} \otimes \mathbf{n}^{\perp} - \mathbf{n} \otimes \mathbf{n})  \\
&=\big[\nabla  \mathcal{H}(\mathbf{y}, \theta)   \colon ( \mathbf{n}^{\perp} \otimes \mathbf{n}^{\perp} ) \big] \mathbf{n},
\end{aligned}
\end{equation}
where the term proportional to $\kappa$ in the first line vanishes since $\mathcal{H}$ is symmetric ($[\mathcal{H}]_{i\alpha \beta} = [\mathcal{H}]_{i\beta \alpha}$) and $( \mathbf{n}^{\perp} \otimes \mathbf{n}^{\perp} - \mathbf{n} \otimes \mathbf{n})$ is skew symmetric. We complete the derivation by substituting the results of (\ref{eq:NablaWTerm}-\ref{eq:getWeirdNaturalBC}) into (\ref{eq:importForStrongForm1}) and using that $\Div \mathcal{H} = \nabla \mathcal{H} \colon \mathbf{I}$ and $\partial_{\mathbf{n}} \mathbf{w} = (\nabla \mathbf{w})\mathbf{n}$.

\section{Derivation of the consistency term}
\label{sec:appendix}
This appendix derives the consistency terms introduced in the interior penalty method in Section \ref{sec:WeakFormFEM}. The basic idea is that, if our numerical method happens upon an exact solution to the governing equations in (\ref{eq:governingEquations}), then it should identify these fields as solutions to the discrete numerical problem as well. This demand leads naturally to the form of $\mathscr{A}_{\text{con}}^h$ in  (\ref{eq:consistency}). We refer to Section \ref{sec:NotationFEM} for any questions about the notation in this derivation. 

 Let $(\mathbf{y}_h, \theta_h) \in V_{\Gamma_{\text{d}}}^h$ and $(\mathbf{w}_h, \eta_h) \in V_{0}^h$. 
Observe using integration by parts that $\mathscr{A}_0^h$ in (\ref{eq:A0h})  satisfies  
\begin{equation}
\begin{aligned}
 \mathscr{A}_0^h( (\mathbf{y}_h, \theta_h);(\mathbf{w}_h, \eta_h)) &= \sum_{T \in \mathcal{T}^h} \bigg\{ \int_T \Div\big[ 
 \Div \mathcal{H}(\mathbf{y}_h, \theta_h)  
 - \mathbf{P}(\mathbf{y}_h, \theta_h) \big] \cdot \mathbf{w}_h \dif x   \\
 &\qquad +  \int_{\partial T}  \big[  \big( \mathbf{P}(\mathbf{y}_h, \theta_h)   - \Div \mathcal{H}(\mathbf{y}_h ,\theta_h) \big) \mathbf{n} \cdot \mathbf{w}_h   + \big(\mathcal{H} (\mathbf{y}_h, \theta_h) \cdot \mathbf{n} \big)  \colon \nabla \mathbf{w}_h \big] \dif s \\
 &\qquad + \int_T \big( q(\mathbf{y}_h, \theta_h) - \nabla \cdot \mathbf{j}(\theta_h) \big) \eta_h  \dif x + \int_{\partial T} \big( \mathbf{j}(\theta_h) \cdot \mathbf{n}  \big) \eta_h \dif s \bigg\}   .  
\end{aligned}
\end{equation}
This integral can be reorganized by replacing the boundary integrals over each $\partial T$ of the triangulation  $\mathcal{T}^h$  with an identical set of integrals along the edges $e \in \mathcal{E}^h$, leading to 
\begin{equation}
\begin{aligned}\label{eq:grungyCalc}
 \mathscr{A}_0^h( (\mathbf{y}_h, \theta_h);(\mathbf{w}_h, \eta_h)) &=  \sum_{T \in \mathcal{T}^h} \int_T  \Big\{  \Div\big[ 
 \Div \mathcal{H}(\mathbf{y}_h, \theta_h)  
 - \mathbf{P}(\mathbf{y}_h, \theta_h) \big] \cdot \mathbf{w}_h + \big( q(\mathbf{y}_h, \theta_h) - \nabla \cdot \mathbf{j}(\theta_h) \big) \eta_h  \Big\} \dif x \\
 &\qquad + \sum_{e \in \mathcal{E}_{\text{int}}^h}  \int_e  \Big\{  \llbracket  \mathbf{P}(\mathbf{y}_h, \theta_h)   - \Div \mathcal{H}(\mathbf{y}_h ,\theta_h) \rrbracket \mathbf{n}_e  \cdot \mathbf{w}_h   + \big(\ldblbrace\mathcal{H} (\mathbf{y}_h, \theta_h)\rdblbrace \cdot \mathbf{n}_e \big)  \colon  \llbracket \nabla \mathbf{w}_h \rrbracket \Big\} \dif s \\
 &\qquad +  \int_{\Gamma_\text{t}} \Big\{ \big(   \mathbf{P}(\mathbf{y}_h, \theta_h)   - \Div \mathcal{H}(\mathbf{y}_h ,\theta_h) \big) \mathbf{n}  \cdot \mathbf{w}_h \Big\} \dif s    +  \int_{\partial \Omega}  \big( \mathcal{H}(\mathbf{y}_h, \theta_h) \cdot \mathbf{n} \big) \colon \nabla \mathbf{w}_h  \dif s \\
 &\qquad +  \sum_{e \in \mathcal{E}_{\text{int}}^h}  \int_e  \big( \llbracket \mathbf{j}(\theta_h) \rrbracket \cdot \mathbf{n}_e  \big) \eta_h  \dif s  + \int_{\partial \Omega} \big(\mathbf{j} (\theta_h) \cdot \mathbf{n} \big)\eta_h \dif s.   
\end{aligned}
\end{equation}
The jump conditions on the interior edges of the mesh are due to the lack of $H^2$-regularity in the deformation $\mathbf{y}_h$ and test field $\mathbf{w}_h$ in our discrete formulation.  However, an exact solution $(\mathbf{y}, \theta) \in V_\Gamma$ to the strong form of the governing equation in (\ref{eq:governingEquations}) satisfies 
\begin{equation}
\begin{aligned}
&\llbracket  \mathbf{P}(\mathbf{y}, \theta)   - \Div \mathcal{H}(\mathbf{y} ,\theta) \rrbracket \mathbf{n}_e = \mathbf{0}   \quad \text{ and } \quad  \llbracket \mathbf{j}(\theta) \rrbracket \cdot \mathbf{n}_e = 0 
\end{aligned}
\end{equation}
on each  $e \in \mathcal{E}_{\text{int}}^h$. Thus, 
replacing  $(\mathbf{y}_h, \theta_h)$ in (\ref{eq:grungyCalc}) with an exact solution and subtracting off the boundary terms $\mathscr{B}_0(\mathbf{w}_h)$ gives 
\begin{equation}
\begin{aligned}
 \mathscr{A}_0^h( (\mathbf{y}, \theta);(\mathbf{w}_h, \eta_h))  - \mathscr{B}_0(\mathbf{w}_h) =\sum_{e \in \mathcal{E}_{\text{int}}^h}  \int_e \big(\ldblbrace\mathcal{H} (\mathbf{y}, \theta)\rdblbrace \cdot \mathbf{n}_e \big)  \colon  \llbracket \nabla \mathbf{w}_h \rrbracket  \dif s + \int_{\Gamma_{\text{n}}}  \big( \mathcal{H}(\mathbf{y}, \theta) \cdot \mathbf{n} \big) \colon \nabla \mathbf{w}_h  \dif s .
\end{aligned}
\end{equation}
The right-hand side here is actually $-\mathscr{A}_{\text{con}}^h( (\mathbf{y}, \theta);(\mathbf{w}_h, \eta_h))$ for $\mathscr{A}_{\text{con}}^h$ defined  in (\ref{eq:consistency}). In adding this term to left side of our discrete formulation (along with $\mathscr{A}_{\text{sta}}^h$ and $\mathscr{A}_{\text{bnd}}^h$), it follows that 
\begin{equation}
\begin{aligned}
 \mathscr{A}^h( (\mathbf{y}, \theta);(\mathbf{w}_h, \eta_h))  =  \mathscr{B}_0(\mathbf{w}_h) \quad \text{ for all } (\mathbf{w}_h, \eta_h)  \in V^{h}_{0}
\end{aligned}
\end{equation}
for any $(\mathbf{y}, \theta) \in V_{\Gamma}$ that solves (\ref{eq:governingEquations}), since $\mathscr{A}_{\text{sta}}^h( (\mathbf{y}, \theta);(\mathbf{w}_h, \eta_h))$ and $\mathscr{A}_{\text{bnd}}^h((\mathbf{y},\theta), (\mathbf{w}_h, \eta_h))$ vanish trivially on an exact solution. This is the desired consistency.

\section{Initialization of the boundary conditions}
\label{sec:InitialGuess}
For investigating large deformations and angle fields in our numerical framework, the Newton solver is improved when we initialize it with a pair of fields $(\mathbf{y}^0_h, \theta^0_h)$ that \textit{a priori} satisfy the boundary conditions. We therefore perform two preliminary (and much simpler) numerical calculations to initialized such fields. 

We start by solving the minimization problem  $\min \{ \int_{\Omega}\frac{1}{2} |\nabla \nabla \mathbf{y} |^2 \dif x + E_{\text{ext}}(\mathbf{y})  \colon \mathbf{y} = \overline{\mathbf{y}} \text{ on } \Gamma_{\text{d}}\; (\nabla \mathbf{y}) \mathbf{n} = \overline{\mathbf{s}} \text{ on } \Gamma_{\text{n}} \}$ using $C^0$ finite elements and the interior penalty method.  Referencing again the notation in Section \ref{sec:NotationFEM} and mimicking the exposition in Section \ref{sec:WeakFormFEM} and   Appendix \ref{sec:appendix}, this problem is formulated in the weak form as: Find a $\mathbf{y}^0_h \in \widetilde{V}_{\Gamma_{\text{d}}}^h$ such that 
\begin{equation}
\begin{aligned}\label{eq:weakFormInitialize}
\mathscr{A}_{\text{pre}}^h(\mathbf{y}^0_h; \mathbf{w}_h) = \mathscr{B}_{\text{pre}}^h ( \mathbf{w}_h) \quad \text{ for all } \mathbf{w}_h \in \widetilde{V}^h_0,
\end{aligned}
\end{equation}
where $\widetilde{V}^h_{\Gamma_{\text{d}}}$ and $\widetilde{V}^h_0$ are the approximation spaces given by restricting $V_{\Gamma_{\text{d}}}^h$ and $V^h_0$, respectively, to the deformation variable (i.e., neglecting the angle variable). In this setting, $\mathscr{A}_{\text{pre}}^h$ is a bilinear form defined by 
\begin{equation}
\begin{aligned}
\mathscr{A}_{\text{pre}}^h(\mathbf{y}^0_h; \mathbf{w}_h) &:=    \sum_{T \in \mathcal{T}^h} \int_T \langle \nabla_h \nabla \mathbf{y}_h^0 , \nabla_h \nabla \mathbf{w}_h \rangle \dif x  - \sum_{e \in \mathcal{E}_{\text{int}}^h} \int_e 
 \big( \ldblbrace \nabla_h \nabla \mathbf{y}^0_h \rdblbrace \cdot  \mathbf{n}_e \big) \colon \llbracket \nabla \mathbf{w}_h \rrbracket 
 \dif s   \\
 &\qquad - \sum_{e \in \mathcal{E}_{\text{int}}^h} \int_e 
 \big( \ldblbrace \nabla_h \nabla \mathbf{w}_h \rdblbrace \cdot  \mathbf{n}_e \big) \colon \llbracket \nabla \mathbf{y}_h^0 \rrbracket 
 \dif s  + \sum_{e \in \mathcal{E}_{\text{int}}^h} \frac{\alpha}{|e|}  \int_{e}  \llbracket \nabla \mathbf{y}^0_h \rrbracket :\llbracket \nabla \mathbf{w}_h \rrbracket \dif s  \\
 &\qquad - \int_{\Gamma_{\text{n}}} \big(\nabla_h \nabla \mathbf{y}^0_h \colon ( \mathbf{n} \otimes \mathbf{n} ) \big) \cdot \nabla \mathbf{w}_h \mathbf{n} \dif s  - \int_{\Gamma_{\text{n}}} \big(\nabla_h \nabla \mathbf{w}_h \colon (\mathbf{n} \otimes \mathbf{n})\big)  \cdot (\nabla \mathbf{y}^0_h)\mathbf{n} \dif s\\
 &\qquad + \sum_{e \in \mathcal{E}^h_{\Gamma_{\text{n}}}}\frac{\alpha}{|e|} \int_{e}  (\nabla \mathbf{y}_h^0) \mathbf{n} \cdot (\nabla \mathbf{w}_h) \mathbf{n} \dif s    
\end{aligned}
\end{equation}
The linear form $\mathscr{B}_{\text{pre}}^h$ on the right side of (\ref{eq:weakFormInitialize}) is defined by 
\begin{equation}
    \begin{aligned}
      \mathscr{B}^h_{\text{pre}}( \mathbf{w}_h) :=   \mathscr{B}_0(\mathbf{w}_h) - \int_{\Gamma_{\text{n}}} \big(\nabla_h \nabla \mathbf{w}_h \colon (\mathbf{n} \otimes \mathbf{n})\big) \cdot \bar{\mathbf{s}} \dif s  + \sum_{e \in \mathcal{E}^h_{\Gamma_{\text{n}}}}\frac{\alpha}{|e|} \int_{e}  \overline{\mathbf{s}} \cdot (\nabla \mathbf{w}_h) \mathbf{n} \dif s.
      \end{aligned}
\end{equation}
In both $\mathscr{A}_{\text{pre}}^h$ and $\mathscr{B}^h_{\text{pre}}$, we include symmetry terms that mirror  $(\mathbf{y}^0_h; \mathbf{w}_h) \equiv (\mathbf{w}_h; \mathbf{y}^0_h)$   in the argument for $\mathscr{A}_{\text{pre}}^h$. These terms are standard (see \cite{di2011mathematical,engel2002continuous}) and easy to implement when using the interior penalty method to solve a linear  PDE  because the weak form of the internal energy is a symmetric bilinear form.  Their inclusion, while not necessary, can improve the effectiveness and efficiency of the  numerical solver. In the typical initialization step to obtain a $\mathbf{y}_h^0$, the value of the penalty parameter is chosen as $\alpha = 20$ in an effort to strongly impose the slope boundary conditions $\nabla \mathbf{y}_h^0 \mathbf{n} = \overline{\mathbf{s}}$; it is subsequently reduced to $\alpha = 0.1$ when solving (\ref{eq:discrete problem}).

Given a $\mathbf{y}_h^0$ solving the weak form above, we produce an initial guess $\theta^0_h$ for the actuation field by approximating the metric constraint. Specifically we compute $\theta^0_h$ as 
\begin{equation}
\begin{aligned}\label{eq:getTheta0h}
     \theta^0_h := \argmin_{\theta_h \in W^h} \Big\{  \int_{\Omega} \left\lvert (\nabla \mathbf{y}^0_h)^T \nabla \mathbf{y}^0_h - \mathbf{A}^T(\theta_h) \mathbf{A}(\theta_h) \right\rvert^2 \dif x + d_2 \int_{\Omega} \left\lvert \nabla \theta_h \right\rvert^2 \dif x \Big\}
\end{aligned}
\end{equation}
in the space  $W^h := \{ \theta_h \in C^0(\Omega, \mathbb{R}) \colon \theta_h\vert_{T} \in P^1(T)\;\;  \forall \;\; T \in \mathcal{T}^h \}$, where the second term with $d_2 > 0$ is the regularizing term in (\ref{eq:regTerm}) associated to the gradient of actuation.
 As (\ref{eq:getTheta0h}) has a standard finite element implementation, we do not detail it here for the sake of brevity.

\section{Computation of consistent reaction forces}
\label{sec:forces}
Here we explain how to compute reaction forces in the simulations. To start, consider a model problem where $(\mathbf{y}, \theta)$ solve the governing equations in (\ref{eq:governingEquations}) so that the weak form  $\mathscr{A}_0((\mathbf{y}, \theta); (\mathbf{w}, \eta)) = \mathscr{B}_0(\mathbf{w})$ holds for all $  (\mathbf{w}, \theta) \in  V_0$  (see the beginning of Section \ref{ssec:Weakform} for the definition of $V_0$  and Section \ref{sec:WeakFormFEM} for the definitions of $\mathscr{A}_0$ and $\mathscr{B}_0$). The reaction forces emerge when we replace the test functions $\mathbf{w}$ in the weak form with a generic test function $\mathbf{w}_{\text{f}} \in H^2(\Omega, \mathbb{R}^3)$ that does not satisfy the vanishing Dirichlet boundary conditions contained in $V_0$. Indeed, using the identity in (\ref{eq:importForStrongForm}) and that  $(\mathbf{y}, \theta)$ solve (\ref{eq:governingEquations}), this replacement yields
 \begin{equation}
 \begin{aligned}\label{eq:forceTrick1}
   \mathscr{A}_0((\mathbf{y}, \theta); (\mathbf{w}_\text{f}, 0)) &=  \int_{\Omega} \overline{\mathbf{b}} \cdot \mathbf{w}_{\text{f}}  \dif x + \int_{\Gamma_{\text{t}}} \overline{\mathbf{t}} \cdot \mathbf{w}_{\text{f}}  \dif s +   \int_{\Gamma_{\text{d}}}  \big[  \big( \mathbf{P}(\mathbf{y}, \theta)   - \nabla \mathcal{H}(\mathbf{y} ,\theta) \colon (\mathbf{I} + \mathbf{n}^{\perp} \otimes \mathbf{n}^{\perp}) \big) \mathbf{n} \big] \cdot \mathbf{w}_{\text{f}} \dif s   \\
   &\qquad  + \int_{\Gamma_{\text{m}}} \overline{\mathbf{m}} \cdot (\nabla \mathbf{w}_{\text{f}} ) \mathbf{n}   \dif s  +   \int_{\Gamma_{\text{n}}} \big(\mathcal{H} (\mathbf{y}, \theta) \colon (\mathbf{n} \otimes \mathbf{n}) \big)  \cdot (\nabla \mathbf{w}_\text{f}) \mathbf{n} \dif s   
 \end{aligned}
\end{equation}
The definition of $\mathscr{B}_0(\mathbf{w})$ then allows us to conclude that 
\begin{equation}
\begin{aligned}
   \mathscr{A}_0((\mathbf{y}, \theta); (\mathbf{w}_\text{f}, 0)) - \mathscr{B}_0(\mathbf{w}_{\text{f}}) &= \int_{\Gamma_{\text{d}}}  \big[  \big( \mathbf{P}(\mathbf{y}, \theta)   - \nabla \mathcal{H}(\mathbf{y} ,\theta) \colon (\mathbf{I} + \mathbf{n}^{\perp} \otimes \mathbf{n}^{\perp}) \big) \mathbf{n} \big] \cdot \mathbf{w}_{\text{f}} \dif s  \\
   &\qquad   +  \int_{\Gamma_{\text{n}}} \big(\mathcal{H} (\mathbf{y}, \theta) \colon (\mathbf{n} \otimes \mathbf{n}) \big) \cdot (\nabla \mathbf{w}_\text{f}) \mathbf{n} \dif s.
\end{aligned}
\end{equation}

Let us now specialize to the pinching boundary conditions of Section \ref{ssec:PinchingSimulations}, where  this method is employed. The Dirichlet boundary set is $\Gamma_{\text{d}} = \Gamma_{L} \cup \Gamma_{R}$ and $\Gamma_{\text{n}} = \emptyset$. Clearly, the force on $\Gamma_{L}$ in the $\mathbf{e}_1$-direction is balanced by that on $\Gamma_R$. So we compute the boundary force in the plots by choosing a test function such that $\mathbf{w}_{\text{f}}  = \mathbf{0}$ on $\Gamma_{L}$ and $\mathbf{w}_{\text{f}}=\mathbf{e}_1$ on $\Gamma_R$. It follows that 
\begin{equation}
\begin{aligned}
   \mathscr{A}_0((\mathbf{y}, \theta); (\mathbf{w}_\text{f}, 0)) - \mathscr{B}_0(\mathbf{w}_{\text{f}}) =  \int_{\Gamma_{\text{R}}}  \big[  \big( \mathbf{P}(\mathbf{y}, \theta)   - \nabla \mathcal{H}(\mathbf{y} ,\theta) \colon (\mathbf{I} + \mathbf{n}^{\perp} \otimes \mathbf{n}^{\perp}) \big) \mathbf{n} \big] \cdot \mathbf{e}_1 \dif s, 
\end{aligned}
\end{equation}
which is the desired total generalized force on the right boundary. 

The numerical formulation possess a modified weak formulation, but the same strategy applies. Once we have a solution $(\mathbf{y}_h, \theta_h)$ in the sense that $\mathscr{A}^h((\mathbf{y}_h, \theta_h); (\mathbf{w}_h, \eta_h)) = \mathscr{B}_0(\mathbf{w}_h)$ for all $(\mathbf{w}_h, \theta_h) \in V_0^h$ (see Section \ref{sec:WeakFormFEM}), we then test this solution with a  $(\mathbf{w}_{h,\text{f}}, 0)$ to reveal the reaction forces. For the pinching simulations,  we choose a test function such that $\mathbf{w}_{h,\text{f}}  = \mathbf{0}$ on $\Gamma_{L}$ and $\mathbf{w}_{h,\text{f}}=\mathbf{e}_1$ on $\Gamma_R$, analogous to the prior choice. The normalized compressive force plotted for these simulations is thus 
\begin{equation}
\begin{aligned}
f_1 = -\frac{1}{d_3} \Big( \mathscr{A}^h((\mathbf{y}_h, \theta_h); (\mathbf{w}_{h,\text{f}}, 0)) - \mathscr{B}_0(\mathbf{w}_{h,\text{f}}) \Big).
\end{aligned}
\end{equation}

\section{Constructing origami deformations from the effective fields}\label{sec:oriDefs}
In Section \ref{sec:Examples}, we plot several examples of (approximate) origami deformations. Here we  explain how to construct such deformations using only the effective fields $(\mathbf{y}, \theta)$ as input.  Along the way, we make liberal use of the notation in Section \ref{ssec:Design} for the design of a  parallelogram origami and its perfect mechanism.

The general idea is sketched in Fig.\;\ref{Fig:OriConstruct}.  Start by recognizing that a reference parallelogram origami unit cell of characteristic length $\sim \ell$ is defined by nine vertices:
\begin{equation}
\begin{aligned}
&\mathbf{x}_0^{\ell} = \mathbf{0},\\
&\mathbf{x}_k^{\ell} = \ell \mathbf{t}^r_k, \quad k = 1,\ldots, 4, \\
&\mathbf{x}_5^{\ell} =  \ell (\mathbf{t}^r_1 + \mathbf{t}_{2}^r), \quad  \mathbf{x}_6^{\ell} =  \ell (\mathbf{t}^r_2 + \mathbf{t}_{3}^r), \quad \mathbf{x}_7^{\ell} =  \ell (\mathbf{t}^r_3 + \mathbf{t}_{4}^r), \quad \mathbf{x}_8^{\ell} =  \ell (\mathbf{t}^r_4 + \mathbf{t}_{1}^r).
\end{aligned}
\end{equation}
Now fix a 2D domain $\Omega$. 
 An overall parallelogram origami pattern can be built systematically on top of $\Omega$ by first defining the index set on $\Omega$
\begin{equation}
\begin{aligned}
      I^{\ell} =\big\{ (m,n) \in \mathbb{Z}^2 \colon m \tilde{\mathbf{u}}_0 + n \tilde{\mathbf{v}}_0 \in \Omega \big\}    
    \end{aligned}
\end{equation}
and then he collection of vertices 
\begin{equation}
\begin{aligned}
      \Omega^{\ell} := \{ \mathbf{x}_k^{\ell} + m\ell \mathbf{u}_0 + n \ell \mathbf{v}_0 \colon k =0,\ldots, 9  \text{ and } (m,n) \in I^{\ell} \} \subset \mathbb{R}^3,    
      \end{aligned}
\end{equation}
which is the  desired parallelogram origami pattern. A stylized version of this construction  is illustrated in Fig\;\ref{Fig:OriConstruct}(a). Unlike the figure, $\Omega^{\ell}$ is a 3D collection of points overlaid on a 2D domain $\Omega$.

\begin{figure}[t!]
\centering
\includegraphics[width=0.8\textwidth]{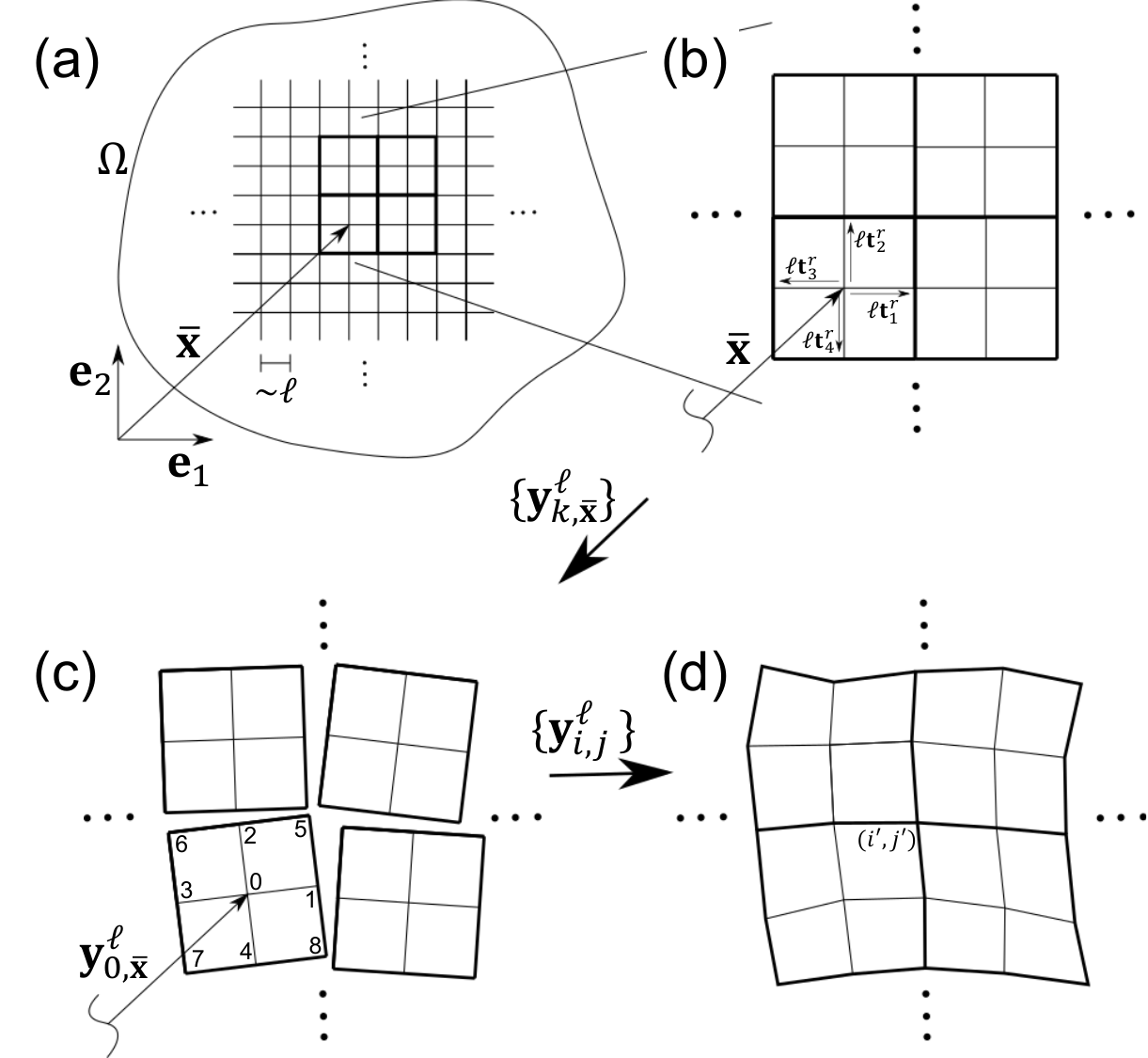} 
\caption{Stylized sketch and notation for constructing origami deformations. (a) A sketch of an origami pattern overlaid onto a domain $\Omega$. Each square of characteristic length $\sim \ell$ represents a panel; four squares represent a cell; $\bar{\mathbf{x}}$ labels the unit cells.  (b) Sketch illustrating $\bar{\mathbf{x}}$-cell and its neighbors to the right and above.  (c) Cell-based anastz is constructed for each $\bar{\mathbf{x}}$-cell. It generically produces gaps between neighboring cells. (d) The gaps are closed by averaging.   }
\label{Fig:OriConstruct}
\end{figure}

Our goal is to deform the vertices in $\Omega^{\ell}$ as approximate origami, characterized by the effective fields $(\mathbf{y}, \theta)$. For this purpose, we assume that $(\mathbf{y}, \theta) \colon \Omega \rightarrow \mathbb{R}^3 \times (\theta^{-}, \theta^+)$ either  solve the equilibrium equations in (\ref{eq:governingEquations})  (analytically or numerically) or satisfy the purely geometric constraints in (\ref{eq:purelyGeom}). We also assume that $|\partial_1 \mathbf{y} \times \partial_2 \mathbf{y}| >0$ to avoid singular surfaces.  Our construction then employs the following  rotation field built from $\mathbf{y}$. Observe that $\mathbf{F} := (\nabla \mathbf{y} , \mathbf{n}(\mathbf{y}))$ for $\mathbf{n}(\mathbf{y})$ in (\ref{eq:secFundDef})  is  a field on $\mathbb{R}^{3\times3}$ with positive determinant. Thus, the Polar decomposition theorem furnishes the stretch tensor field
$\mathbf{U} := \sqrt{\mathbf{F}^T \mathbf{F}}$.  The desired rotation field  on $\Omega$  is then $\mathbf{R} := \mathbf{F} \mathbf{U}^{-1}$.

We build the origami deformation from the fields $\mathbf{y}, \theta, \mathbf{R}$ on $\Omega$ in two steps. The  first step is a cell-based construction, where each unit cell deforms as an origami mechanism. Fix  $\bar{\mathbf{x}} = \bar{m} \tilde{\mathbf{u}}_0 + \bar{n} \tilde{\mathbf{v}}_0$ for $(\bar{m}, \bar{n}) \in I^{\ell}$ as illustrated in Fig.\;\ref{Fig:OriConstruct}(a). The $\bar{\mathbf{x}}$-cell contains nine vertices
\begin{equation}
    \begin{aligned}
      \mathbf{x}_{k,\bar{\mathbf{x}}}^{\ell} = \begin{pmatrix}  \bar{\mathbf{x}} \\ 0 \end{pmatrix} + \mathbf{x}_k^{\ell}, \quad i = 0 ,\ldots, 9.
    \end{aligned}
\end{equation}
We deform the cell  by taking each $\mathbf{x}_{k,\bar{\mathbf{x}}}^{\ell}$ to a $\mathbf{y}_{k,\bar{\mathbf{x}}}^{\ell}$ given by 
\begin{equation}
    \begin{aligned}\label{eq:cellBaased}
    &\mathbf{y}_{0,\bar{\mathbf{x}}}^{\ell} = \mathbf{y}(\bar{\mathbf{x}})  \\
    &\mathbf{y}_{k,\bar{\mathbf{x}}}^{\ell} = \mathbf{y}(\bar{\mathbf{x}}) + \ell \mathbf{R}(\bar{\mathbf{x}}) \mathbf{t}_k(\theta(\bar{\mathbf{x}})), \quad k = 1,\ldots,4,\\
    &\mathbf{y}_{5,\bar{\mathbf{x}}}^{\ell} = \mathbf{y}(\bar{\mathbf{x}}) + \ell\mathbf{R}(\bar{\mathbf{x}})\big[ \mathbf{t}_1(\theta(\bar{\mathbf{x}})) + \mathbf{t}_2(\theta(\bar{\mathbf{x}}))\big], \\
    &\mathbf{y}_{6,\bar{\mathbf{x}}}^{\ell} = \mathbf{y}(\bar{\mathbf{x}}) + \ell\mathbf{R}(\bar{\mathbf{x}})\big[ \mathbf{t}_2(\theta(\bar{\mathbf{x}})) + \mathbf{t}_3(\theta(\bar{\mathbf{x}}))\big], \\
    &\mathbf{y}_{7,\bar{\mathbf{x}}}^{\ell} = \mathbf{y}(\bar{\mathbf{x}}) + \ell\mathbf{R}(\bar{\mathbf{x}})\big[ \mathbf{t}_3(\theta(\bar{\mathbf{x}})) + \mathbf{t}_4(\theta(\bar{\mathbf{x}}))\big],\\
    &\mathbf{y}_{8,\bar{\mathbf{x}}}^{\ell} = \mathbf{y}(\bar{\mathbf{x}}) + \ell\mathbf{R}(\bar{\mathbf{x}})\big[ \mathbf{t}_4(\theta(\bar{\mathbf{x}})) + \mathbf{t}_1(\theta(\bar{\mathbf{x}}))\big].
    \end{aligned}
\end{equation}
Applying these formulas  to all the $\bar{\mathbf{x}}$-cells defined by $\Omega^{\ell}$ produces a cell-based ansatz $\{ \mathbf{y}^{\ell}_{k,\bar{\mathbf{x}}}\}$.  Each cell deforms as a perfect mechanism informed by the fields $\mathbf{y}, \theta$ and $\mathbf{R}$. However, intercell gaps generically arise since we do not enforce any compatibility conditions at the cell boundaries. Fig.\;\ref{Fig:OriConstruct}(b) and (c) qualitatively sketches this ansatz and its gaps.   

The final step of the construction is to average the gaps, and thus replace the cell-based ansatz $\{ \mathbf{y}_{k,\bar{\mathbf{x}}}^{\ell}\}$ by a vertex based one $\{\mathbf{y}_{i,j}^{\ell}\}$. This is basically an accounting problem, involving a number of distinct cases. We detail one of them to explain the main idea. Let $(i',j')  \in \mathbb{Z}^2$ index the one vertex that belongs to all four cells  in Fig.\;\ref{Fig:OriConstruct}(b), as shown in Fig.\;\ref{Fig:OriConstruct}(d). The cell based construction in Fig.\;\ref{Fig:OriConstruct}(c) maps this point to (potentially) four distinct points by via the cell formulas in (\ref{eq:cellBaased}).  We now simply average those points to close the gaps 
\begin{equation}
    \begin{aligned}
        \mathbf{y}_{i',j'}^{\ell} = \frac{1}{4} \big(\mathbf{y}_{5,\bar{\mathbf{x}}}^{\ell} + \mathbf{y}_{6,\bar{\mathbf{x}}+ \ell \tilde{\mathbf{u}}_0}^{\ell} + \mathbf{y}_{8,\bar{\mathbf{x}}+ \ell \tilde{\mathbf{v}}_0}^{\ell} + \mathbf{y}_{7,\bar{\mathbf{x}}+ \ell (\tilde{\mathbf{u}}_0 +   \tilde{\mathbf{v}}_0)}^{\ell} \big) 
    \end{aligned}
\end{equation}
A similar strategy is employed to close all the gaps in the pattern, completing the construction. 

Note that we have added panel strain by this averaging step. However, the strain should be quite small. Basically, it should be imperceptible to the naked eye, so long as the cell size  $\ell$ is $\ll$   than the characteristic length of the domain $L_{\Omega}$ and the  fields   $(\mathbf{y},\theta)$  satisfy or approximately satisfy (\ref{eq:purelyGeom}). For those interested,  a rigorous justification of this statement is one of the main topics of  our prior work  \cite{xu2024derivation}. 

\bibliographystyle{plain}
\bibliography{bib}

\end{document}